\documentclass[12pt,notitlepage]{article}
\usepackage[latin9]{inputenc}
\usepackage[active]{srcltx}
\usepackage{xcolor}
\usepackage{mathrsfs}
\usepackage{amsmath}
\usepackage{amsthm}
\usepackage{amssymb}
\usepackage[letterpaper]{geometry}
\usepackage{setspace}
\usepackage[authoryear]{natbib}
\PassOptionsToPackage{normalem}{ulem}
\usepackage{ulem}
\makeatletter

\providecolor{lyxadded}{rgb}{0,0,1}
\providecolor{lyxdeleted}{rgb}{1,0,0}
\DeclareRobustCommand{\mklyxadded}[1]{\textcolor{lyxadded}\bgroup#1\egroup}
\DeclareRobustCommand{\mklyxdeleted}[1]{\textcolor{lyxdeleted}\bgroup\mklyxsout{#1}\egroup}
\DeclareRobustCommand{\mklyxsout}[1]{\ifx\\#1\else\sout{#1}\fi}

\theoremstyle{plain}
\newtheorem{assumption}{\protect\assumptionname}
\theoremstyle{definition}
\newtheorem{defn}{\protect\definitionname}[section]
\theoremstyle{plain}
\newtheorem{thm}{\protect\theoremname}[section]
\theoremstyle{plain}
\newtheorem{lem}{\protect\lemmaname}[section]
\theoremstyle{plain}
\newtheorem{cor}{\protect\corollaryname}[section]

\usepackage{threeparttable}
	
\usepackage{lscape}
\usepackage{amsfonts}

\usepackage{graphicx}
\usepackage{latexsym}
\usepackage{epsfig}\usepackage{epsfig}
\usepackage{graphicx}\usepackage{tabularx}
\usepackage{rotating}
\usepackage{hyperref}
\usepackage{setspace}
\usepackage{geometry}
\usepackage{caption}
\providecommand{\U}[1]{\protect\rule{.1in}{.1in}}

\renewenvironment{thebibliography}[1]{
  \begin{oldthebibliography}{#1}
    \setlength{\itemsep}{0em}
    \setlength{\parskip}{0em}
}
{
  \end{oldthebibliography}
}

\makeatother

\providecommand{\assumptionname}{Assumption}
\providecommand{\corollaryname}{Corollary}
\providecommand{\definitionname}{Definition}
\providecommand{\lemmaname}{Lemma}
\providecommand{\theoremname}{Theorem}

\begin{document}
\title{Differential Test Performance and Peer Effects\thanks{Ying Zeng gratefully acknowledges financial support from the Young
Scientists Fund of the Fujian Natural Science Foundation (Grant No.
2023J05010)}}
\author{Guido M. Kuersteiner\thanks{Department of Economics, University of Maryland, College Park, MD
20742, e-mail: gkuerste@umd.edu}, Ingmar R. Prucha\thanks{Department of Economics, University of Maryland, College Park, MD
20742, e-mail: prucha@umd.edu}, Ying Zeng\thanks{Wang Yanan Institute for Studies in Economics, Xiamen University,
zengying17@xmu.edu.cn}}
\maketitle
\begin{abstract}
We use variation of test scores measuring closely related skills to
isolate peer effects. The intuition for our identification strategy
is that the difference in closely related scores eliminates factors
common to the performance in either test while retaining idiosyncratic
test specific variation. Common factors include unobserved teacher
and group effects as well as test invariant ability and factors relevant
for peer group formation. Peer effects work through idiosyncratic
shocks which have the interpretation of individual and test specific
ability or effort. We use education production functions as well as
restrictions on the information content of unobserved test taking
ability to formalize our approach. An important implication of our
identifying assumptions is that we do not need to rely on randomized
group assignment. We show that our model restrictions are sufficient
for the formulation of linear and quadratic moment conditions that
identify the peer effects parameter of interest. We use Project STAR
data to empirically measure peer effects in kindergarten through third
grade classes. We find evidence of highly significant peer effects
with magnitudes that are at the lower end of the range of estimates
found in the literature.
\end{abstract}
\newpage{}

\section{Introduction}

\global\long\def\diag{\operatorname{diag}}%
\global\long\def\tr{\operatorname{tr}}%
\global\long\def\Var{\operatorname{Var}}%
\global\long\def\det{\operatorname{det}}%
\global\long\def\Cov{\operatorname{Cov}}%
\global\long\def\argmin{\operatorname{argmin}}%
\global\long\def\plim{\operatorname{plim}}%

We develop a framework to analyze, identify and estimate the effect
of peer groups on performance measures for individuals who are allocated
to groups. Our analysis centers around the idea of measuring the performance
of individuals in closely related tasks and within a short period
of time. An example are aptitude tests administered at the end of
the school year and in related areas such as reading, writing and
word comprehension. By considering the individual quasi-difference
in scores we are able to eliminate unobserved ability and unobserved
group effects. Our key identifying assumption postulates that different
tests measure similar skills and that variation in differential scores
is, apart from variation induced by observed covariates, due to idiosyncratic
cross-sectional variation in ability or effort of individuals. That
variation is not related to prior performance and other systematic
and possibly unobserved factors correlated with group formation such
as teacher quality or parental support.

A leading example where our method can be applied are test scores
of students who are allocated to classrooms. However, the theory could
equally be applied to outcome measures of members of sports teams,
teams of workers or other groups. Identification is based on the availability
of multiple performance measures of comparable quality or information
content. We allow team formation to be endogenous but also cover the
case of randomly selected groups. For expositional purposes we refer
to groups as classrooms and individuals selected into groups as students
with the understanding that our procedure applies to a broader set
of scenarios.

The first result of our paper shows that our key identifying restriction
implies orthogonality conditions that are at the core of our identification
strategy. These restrictions are obtained without any assumptions
about random group selection or independence of individual characteristics
in the population of individuals. Combined with the assumption of
linear latent outcomes and mean group peer effects this leads to an
empirical model that falls within the class of linear peer effects
models that have been prominently studied in the literature, see for
example \citet{manski_identification_1993,calvo-armengol_peer_2009,blume_linear_2015}
and \citet{angrist_perils_2014} for a critique of these models. We
expand on this literature by explicitly accounting for endogenous
peer group selection and individual heterogeneity in unobserved test
taking ability. Most of the empirical literature uses observable characteristics
as well as proxies for unobserved ability to measure the quality of
peer groups. Our baseline model is formulated for unobserved, at least
to the analyst, peer characteristics. We use multiple performance
measures for similar skills to difference out common test taking ability.
We allow group selection to depend on unobserved as well as observed
characteristics. Recent contributions to the econometrics literature
accounting for endogenous group and network formation include \citet{goldsmith-pinkham_social_2013},
\citet{hsieh_social_2016} and \citet{griffith_random_2022} who use
an explicit network formation model, \citet{qu_estimating_2015},
\citet{johnsson_estimation_2019}, \citet{auerbach_identification_2022}
who use control function approaches and \citet{kuersteiner_dynamic_2020}
who use an instrumental variables approach.\footnote{The literature on peer effects is part of a larger literature on network
effects. It is well recognized that the adjacency matrix frequently
used in modeling peer effect is a special case of the weight matrices
used in a class of models introduced by \citet{cliff_spatial_1973a,cliff_spatial_1981a},
which were originally intended for modeling spatial network effects;
see \citet{anselin_thirty_2010} and Lee and Yu (2015)\nocite{Lee2015}
for a review of important contribution of this literature on identification
and estimation of network effects from cross sectional and panel data,
and \citet{kuersteiner_dynamic_2020} for a recent contribution that
connects those strands of network literature.}

We apply our approach to the Project STAR data set of the Tennessee
class size experiment. We exploit variation on closely related test
score outcomes for kindergarten to third grade students. Our empirical
analysis focuses on the identification, estimation and statistical
inference for the parameter determining the marginal effect of unobserved
peer quality. More specifically, this parameter measures the impact
of a unit increase in peer quality, measured in terms of latent peer
outcomes, on individual outcomes. Obtaining data with convincing exogenous
peer group variation is generally difficult and may require costly
experimental designs. Yet, most of the empirical literature estimating
peer group composition effects relies on conditionally randomly assigned
peers. An advantage of the approach proposed in this paper is that
we do not rely on random group assignment. Our framework enables us
to identify the marginal parameter without variation in peer group
composition. The random assignment assumption has been criticized
for Project STAR data, for example because of attrition in higher
grades, see \citet{hanushek_failure_2003}. Our approach is designed
to work despite these data limitations. In addition, and unlike in
related studies identifying similar parameters such as \citet{kelejian_2sls_2002},
\citet{lee_identification_2007}, \citet{graham_identifying_2008}
or \citet{kuersteiner_efficient_2023} we do not rely on group size
variation or variation in group type variances.

Our estimators are related to methods that use panel data to identify
peer effects, see \citet{mas_peers_2009}, \citet{arcidiacono_estimating_2012},
\citet{cornelissen_peer_2017}, \citet{miraldo_identification_2021}
and \citet{braun_estimation_2023}. In contrast with panel methods,
we emphasize applications with cross-sectional data where measures
of closely related skills are observed essentially simultaneously
in time. For one, this framework alleviates problems with sample attrition
which are well documented for Project STAR. More importantly, within
this framework the source of identifying variation is due to remaining
randomness, after conditioning on ability, in the responses of an
individual to test questions rather than due to random selection of
students into classrooms. An advantage of our framework is that its
credibility can be assessed and influenced by a testing protocol that
complies with its basic premise.

There is an extensive empirical literature studying educational outcomes
using Project STAR data. Peer effects specifically were considered
by numerous authors. \citet{boozer_black_2001} use within and between
class variation as well as controls for whether students currently
are or previously were in a small classroom to estimate the parameter
for endogenous peer effects. Variation in the exposure of classmates
to being previously in small classrooms provides variation in peer
quality that is used to identify the endogenous peer effect. They
find large peer effects in second and third grade and negative but
insignificant effects in first grade. \citet{whitmore_resource_2005}
studies the effects of variation in the ratio of girls on test scores
and finds mixed results depending on grades. A decomposition of the
effect into endogenous and exogenous peer effects results in an estimated
increase in own test score of 0.6 points for every point increase
in the average peer test score. \citet{graham_identifying_2008} uses
a model similar to ours and estimates endogenous peer effects based
on differences of the between and within variances in small and regular
size classes. He finds stronger peer effects of being randomly assigned
to smaller as compared to larger classes. \citet{kuersteiner_efficient_2023}
interpret the variance approach of \citet{graham_identifying_2008}
as being part of a class of more general random group effects models,
see also \citet{rose_identification_2017} for an approach similar
to \citet{graham_identifying_2008} using variance restrictions. \citet{chetty_how_2011}
link Project STAR data with 1996-2008 tax records to investigate classroom,
teacher and peer effects on future earnings. Their empirical approach
relies on random assignment to Kindergarten and elementary school
grades which they test using additional individual level data obtained
from tax records. Using analysis of variance and regression based
methods exploiting within classroom variation they find significant
effects of class quality on earnings. In work that looks at a related
question \citet{bietenbeck_longterm_2020} investigates the link between
low ability repeaters in kindergarten classes on educational outcomes
later in life and finds positive effects of being exposed to repeaters
on the probability to graduate from high school and taking a college
entrance exam. \citet{pereda-fernandez_social_2017} relies on a conditional
double randomization assumption resulting in conditional independence
restrictions between student and teacher effects. These restrictions
imply covariance as well as higher order restrictions that are exploited
to estimate both the endogenous peer effect as well as the distribution
of teacher effects. The paper finds sizable peer effects and investigates
optimal teacher and class size allocation rules. \citet{lewbel_social_2023}
also estimate a linear peer effects model for Project STAR data, but
assuming that actual interaction between students is unobserved. They
use restrictions on reduced form parameters to recover the endogenous
peer effect. Empirically, they find large peer effects for second
and third grade math scores.

In our empirical work we use Project STAR data and focus on students
in kindergarten to the third grade. We use individual SAT scores in
reading and word study skills as our outcome measures. We document
the high correlation between the two scores, unconditionally as well
as conditionally on a full set of controls. The correlation coefficients
are about 0.9 for all four grades The peer effects coefficient of
our model is estimated precisely and is stable across the different
grades except for the first grade where we observe a lower degree
of peer effects. Quantitatively, we estimate that an increase of average
unobserved peer quality measured in SAT equivalent scores of 10 SAT
points leads to an increase of 2 to 4 SAT points for an individual
classmate. These measurements are quite robust to the inclusion or
omission of additional controls. When we measure exogenous peer effects,
in other words the effects of average age, race and gender, as well
as an indicator of economic background, we find mostly insignificant
results on differential scores. This lack of significance is further
evidence that our differencing strategy eliminates much of the systematic
factors explaining test score variation. However, it does not mean
that exogenous controls have no effect on the level of individual
scores.

The paper is organized as follows. In Section \ref{Model} we develop
the identification strategy for a model without covariates. Section
\ref{subsec:Latent-Group} defines performance measures and group
selection, Section \ref{subsec:Peer-Effects} introduces the education
production function and peer effects and Section \ref{subsec:Identification}
discusses the key identifying restriction that defines closely related
skills. Section \ref{sec:Model-with-Covariates} extends the specification
to include covariates and introduces the estimators we propose. Section
\ref{sec:Empirical-Results} contains the empirical analysis of peer
effects using Project STAR data. We report a battery of specification
checks in Section \ref{subsec:Plausibility-of-Assumptions}. Our main
empirical findings are reported in Section \ref{subsec:Main_Results}.
Section \ref{sec:Technical} develops formal identification proofs.
Section \ref{sec:Conclusions} contains conclusions. Proofs and tables
are contained in the appendix. An online appendix contains additional
lemmas used in the proofs. In the online appendix we also give results
that formally establish the consistency and asymptotic normality of
the considered estimators, as well as some additional empirical results.

\section{Model\label{Model}}

We assume that for each student $i=1,...,n\text{ }$we observe results
for two different tests, or more generally that we observe two measures
of academic achievement, say, $y_{it}$ where $t=1,2$ indexes the
test. Our empirical work uses Project STAR data where these measures
are scores of the Stanford Achievement Tests (SAT) in word study skills
and reading tests given to kindergarten through third grade students.
Our approach uses two outcome measures to control for unobserved and
possibly correlated student characteristics as well as group level
effects related to classrooms, teachers and schools.

Let $\tilde{y}_{it}^{*}$ be the unobserved ability or effort of student
$i$ taking test $t$. We use $*$ to denote variables that are not
observed by the analyst, but may be observed by the individuals, and
we use a surmount\ $\tilde{}$\  to denote variables in the original
sample before group assignment. The notation $\tilde{y}_{it}^{*}$
emphasizes that the variable refers to the $i$-th student in the
population of students prior to group allocation, and that the index
$i$ is an individual student identifier which is not directly tied
to classrooms, e.g., $i$ may correspond to an alphabetical ordering
of all students by name. There are two steps that link the vector
of latent outcomes $\left(\tilde{y}_{1t}^{*},...,\tilde{y}_{nt}^{*}\right)$
to a vector of observed performance $\left(y_{1t},...,y_{nt}\right)$.
The first step consists of allocating student $i$ in the original
sample to classroom $c$ where student $i$ now receives the in-class
identifier $r.$ For example, $r$ could be the alphabetical rank
of student $i$ in class $c.$ Let $n_{c}$ denote the size of classroom
$c$ and suppose there are $C$ classrooms so that $n=n_{1}+...+n_{C}$.
We then map the $r$-th student in classroom $c$ into a new index,
say, $i'=n_{1}+...+n_{c-1}+r$. In the following we will, abusing
notation slightly, refer, e.g., to the unobserved ability of the $r$-th
student in classroom $c\text{ }$ taking test $t$ either as $y_{crt}^{*}$
or $y_{i't}^{*}$, depending on the context. We will furthermore use
$i$ rather than $i'$ for the new index for convenience of notation.
The second step is an education production function or a more general
group interaction model $\psi$ that relates latent ability $y_{crt}^{*}$
of classmates to actual test performance $y_{crt}$. Actual performance
depends on class level unobserved effects, own latent ability, observed
individual, class and school level covariates, as well as potential
peer interaction. In what follows, we use the term class level effects
for all factors affecting outcomes that are invariant at the classroom
level. Our goal is to isolate the portion of observed performance
that is due to peer effects.

\subsection{Latent Performance and Group Selection\label{subsec:Latent-Group}}

Selection into classrooms can be correlated with a latent baseline
performance measure $\tilde{\zeta}_{i}^{*}$ that can also be thought
of individual test taking ability or latent performance, or effort
prior to taking the observed tests. The purpose of introducing baseline
test taking ability is to account for unobserved student characteristics
that are invariant for the two tests. Our key identifying assumption
of the two tests measuring essentially identical skills detailed in
Section \ref{subsec:Identification} implies that the ability $\tilde{y}_{it}^{*}$
of taking test $t$ can be decomposed into test invariant factors
that include baseline ability $\tilde{\zeta}_{i}^{*}$ and idiosyncratic
errors that are uncorrelated across students and tests. Our method
also accounts for observable student and classroom specific characteristics.
We omit the observed covariates for now to simplify the discussion
of our identification strategy. We discuss a more complete model that
includes covariates in Section \ref{sec:Model-with-Covariates}. For
notational convenience we collect the latent baseline performance
measures for all $n$ students in the sample in the following $n\times1$
vector: $\tilde{\zeta}^{\ast}=\left(\tilde{\zeta}_{1}^{\ast},...,\tilde{\zeta}_{n}^{\ast}\right)'$.
Selection into classrooms can, in addition to depending on $\tilde{\zeta}^{*}$
and observed individual characteristics, depend on observed class
level characteristics as well as on a vector $\alpha=\left(\alpha_{1},...,\alpha_{C}\right)'$
of unobserved class level characteristics. We do not impose any restrictions
on the cross-sectional dependence in $\tilde{\zeta}^{\ast}$, $\alpha$
or observed individual and class level characteristics. Latent baseline
performance of student $i$ may be correlated with latent performance
of student $j$ for a variety of reasons including similar educational,
socioeconomic, cultural, religious or geographic backgrounds that
may or may not be known to the analyst.

We now describe the selection of students into classrooms. This assignment
process results in a reordering of the students characterized by a
one-to-one mapping of the original index set to the new index set.
Let $\tilde{y}_{t}^{*}=\left(\tilde{y}_{1t}^{*},...,\tilde{y}_{nt}^{*}\right)'$
denote the vector of unobserved test-taking ability of students taking
test $t$. The vector $\tilde{y}_{t}^{*}$ is indexed corresponding
to the the original sample or population. Let $S_{c}$ be the $n_{c}\times n\text{ }$
selection matrix that allocates students to classroom $c$, and let
$S=\left[S_{1}',...,S_{C}'\right]'$ be the selector matrix for all
students. The set of all possible classroom allocations is denoted
by $\mathscr{S}$. A random assignment mechanism can be thought of
as selecting one element of $\mathscr{S}$ at random. Corresponding
to $S$ the vector $y_{t}^{*}=\left(y_{1t}^{*},\ldots,y_{nt}^{*}\right)^{\prime}$
representing the sample or population of unobserved ability for test
$t$ of students ordered by classrooms is given by 
\[
y_{t}^{*}=S\tilde{y}_{t}^{*}.
\]
As remarked above, in slight but obvious abuse of notation, let $i=n_{1}+...+n_{c-1}+r$.
Then we denote the $i$-th student's unobserved ability for test $t$
interchangeably by either $y_{it}^{*}$ or $y_{crt}^{*}$. The latter
indexing convention, which we also employ analogously for other variables,
is convenient when an analysis at the classroom level is required.
Now let $y_{ct}^{*}=\left(y_{c1t}^{*},...,y_{cn_{c}t}^{*}\right)',$
then 
\[
y_{ct}^{*}=S_{c}\tilde{y}_{t}^{*}.
\]
The latent performance vector $\zeta^{*}$ is defined analogously.

We allow for $S$, $\tilde{\zeta}^{\ast}$, $\tilde{y}_{t}^{*}$,
for $t=1,2$, and $\alpha$ to be mutually correlated and also to
be correlated with observed individual and class level characteristics
for the generalized model presented in Section \ref{sec:Model-with-Covariates}.
This assumption includes scenarios where $S$ is selected completely
at random, scenarios where $S$ is selected based on student and class
level characteristics, as well as scenarios where students respond
to the allocation $S$. Our assumptions, formally spelled out in Section
\ref{subsec:Identification}, require that selection happens before
testing which is a mild restriction that should be satisfied in most
scenarios we have in mind. We do not otherwise specify or restrict
the mechanism $S$ that selects students into classrooms. We also
do not require that $S$ is observed. This is relevant in situations
where we do not have information identifying individual students such
as through their names, social security numbers or residential address.
On the other hand we do require that we know who, among all test takers,
is allocated to the same classroom. For project STAR data, this information
is available for the majority of students in grades K through three
through a class identifier\footnote{The identifier is a teacher ID that is unique for each classroom.},
but not for students in higher grades.

\subsection{Peer Effects\label{subsec:Peer-Effects}}

We relate educational outputs to observed and unobserved inputs via
an educational production function denoted by $\psi$. Educational
production functions were considered by \citet{krueger_experimental_1999}
and \citet{todd_specification_2003} among others. Our specification
of the educational production function explicitly allows for peer
effects and models observed performance $y_{crt}$ of student $r$
in classroom $c$ as a function of latent ability or effort $y_{c1t}^{*},\ldots,y_{cn_{c}t}^{*}$
of all students in class $c$. The assumptions we impose on $y_{crt}^{*}$
allow it to result from a class assignment process that depends on
both observed and unobserved student and class level characteristics.

We interpret $y_{crt}^{*}$ as the unobserved test taking ability
in the absence of peer effects, but accounting for possible selection
of student $r$ into classroom $c,$ taking test $t.$ The function
$\psi_{crt}\left(.\right)$ accounts for the fact that test performance
varies by student and thus depends on $r$ as well as additional student,
teacher, classroom and school characteristics not captured by $y_{crt}^{*}$.
Observed performance is determined by $y_{crt}=\psi_{crt}\left(y_{c1t}^{*},\ldots,y_{cn_{c}t}^{*}\right)$.
The formulation of $\psi_{crt}\left(.\right)$ explicitly accounts
for the possibility of peer effects by allowing for individual performance
to depend on the characteristics of all peers. In the absence of peer
effects the function $\psi_{crt}\left(.\right)$ simplifies to $y_{crt}=\psi_{crt}\left(y_{crt}^{*}\right).$
In other words, the performance of the $r$-th student in class $c$
only depends on own characteristics and, for the generalized model
presented below, on class characteristics such as teachers and resources
but not on the characteristics of other students in the class.

At the class level let $y_{ct}=\left(y_{c1t},\ldots,y_{cn_{c}t}\right)^{\prime}$,
$y_{ct}^{*}=\left(y_{c1t}^{*},\ldots,y_{cn_{c}t}^{*}\right)'$, and
let $\psi_{ct}=\left(\psi_{c1t},\ldots,\psi_{cn_{c}t}\right)^{\prime}$.
We postulate the following education production function, which accounts
for possible peer effects.
\begin{assumption}
\label{assu:Production-Function}Assume that 
\begin{equation}
y_{ct}=\psi_{ct}\left(y_{ct}^{*}\right)=\alpha_{c}\mathbf{1}_{c}+\varphi_{c}+\left(I_{c}+\rho M_{c}\right)y_{ct}^{*},\label{eq:Ed_ProdF}
\end{equation}
where $\mathbf{1}_{c}=\left(1,....,1\right)^{\prime}$ is an $n_{c}\times1$
vector and $M_{c}$ is some matrix of dimension $n_{c}\times n_{c}$
with zero diagonal elements and $\rho$ is a fixed parameter. The
vector $\varphi_{c}=\left(\varphi_{c1},...,\varphi_{cn_{c}}\right)^{\prime}$
contains measurement errors.
\end{assumption}
The part of the expression involving $M_{c}$ in \eqref{eq:Ed_ProdF}
models peer effects. The vector $M_{c}y_{ct}^{*}$ contains weighted
averages of the latent performance of classroom peers, and the parameter
$\rho$ specifies the degree to which peer effects influence actual
outcomes. \footnote{In the usual terminology of Cliff-Ord (1973,1981) models, $M_{c}y_{ct}^{*}$
is called a spatial lag and $\rho$ the corresponding parameter.}

Peer effects specifications such as the one in Assumption \ref{assu:Production-Function}
have been motivated as solutions to Nash games in the literature,
see \citet{calvo-armengol_peer_2009}, \citet{blume_linear_2015}
or \citet{pereda-fernandez_social_2017}. Our model allows for measurement
errors $\varphi_{c}$ affecting test scores. Consistent with our assumption
that tests $t=1$ and $t=2$ measure essentially identical skills,
measurement errors are assumed to be test invariant. We do not impose
any restrictions on whether $\varphi_{c}$ is correlated with $y_{ct}^{*}$
or whether $\varphi_{c}$ is cross-sectionally correlated. The model
in \eqref{eq:Ed_ProdF} is equivalent to $y_{ct}=\alpha_{c}\mathbf{1}_{c}+\left(I_{c}+\rho M_{c}\right)\left(y_{ct}^{*}+\varphi_{c}\right)$
where the measurement errors are in terms of ability $y_{ct}^{*}$
rather than outcomes $y_{ct}.$

In line with the spatial and social network literature building on
the models introduced by Cliff-Ord (1973,1981), the matrix $M_{c}$
in Assumption \ref{assu:Production-Function} could be an arbitrary
weight matrix describing peer interaction subject to certain measurability
assumptions specified below. For exposition and concreteness we focus
on the case where average characteristics of peers determine peer
effects. Define $M_{c}=\left(\mathbf{1}_{c}\mathbf{1}_{c}^{\prime}-I_{c}\right)/\left(n_{c}-1\right)$
where $I_{c}$ is the $n_{c}\times n_{c}$ identity matrix. Then,
$M_{c}$ is the operator that computes the leave out average of peer
characteristics. Combined with the functional form in Assumption \ref{assu:Production-Function}
this leads, using scalar notation, to an education production function
$y_{crt}=\psi_{crt}\left(y_{ct}^{*}\right)$ for observed outcomes
in terms of unobservables as given below ($t=1,2$) 
\begin{equation}
y_{crt}=\alpha_{c}+\varphi_{cr}+y_{crt}^{*}+\rho\left[(n_{c}-1)^{-1}\sum_{l=1,l\neq r}^{n_{c}}y_{clt}^{*}\right].\label{eq:Granham}
\end{equation}

\subsection{Closely Related Skills\label{subsec:Identification}}

In this paper we propose an identification strategy for the parameter
$\rho$ that does not require randomization over $\mathscr{S}$ or
variation in $S_{c}.$ Our argument proceeds in two steps. First,
we propose moment restrictions on unobserved test taking ability $y_{it}^{*}$
and show that these restrictions imply linear and quadratic moment
restrictions for quantities we are able to observe in the data. We
then show that $\rho$ can be identified by considering a quasi-differencing
transformation of \eqref{eq:Granham} and a GMM estimator that exploits
the restrictions on the conditional mean and variance implied by our
identifying assumption. The asymptotic theory for a general class
of GMM estimators of this type was developed by \citet{kuersteiner_dynamic_2020}.
Here, we rely on their theory to derive sharp identification results
and inference procedures that are specific to this application.

We formalize the information structure of our model in the following
definition. Recall that $\zeta^{*}=\left(\zeta_{1}^{*},\ldots,\zeta_{n}^{*}\right)^{\prime}=S\tilde{\zeta}^{\ast}$,
which stands for unobserved individual characteristics, latent test
performance, effort or ability prior to taking tests $t=1,2.$ We
use the notation $-t$ in the definition below to denote the other
test. We define the following:
\begin{defn}
\label{def:Filtration-no-cov}Let $\varphi=\left(\varphi_{1}',...,\varphi_{C}'\right)'$,
$y_{-it}^{*}=\left(y_{1t}^{*},...,y_{i-1,t}^{*},y_{i+1,t}^{*},...,y_{n,t}^{*}\right)'$
and $y_{t}^{*}=\left(y_{1t}^{*},...,y_{nt}^{*}\right)^{\prime}$.
Define the sigma fields (information sets) $\mathcal{F}_{n,i,t}=\sigma\left(S,\alpha,\varphi,\zeta^{\ast},y_{-t}^{\ast},y_{-i,t}^{*}\right)\text{\,for \ensuremath{t=1,2}}$.
Let $\mathcal{Z}_{n}^{*}=\sigma\left(S,\alpha,\varphi,\zeta^{\ast}\right)$
be the sigma field of all conditioning variables.
\end{defn}
The information set $\mathcal{Z}_{n}^{*}$ consists of classroom allocations
$S,$ as well as of unobserved variables $\zeta^{\ast}$, $\alpha$
and $\varphi.$ The following restriction on the conditional means
of the latent outcomes is at the core of our proposed method of identifying
peer effects. We note that the assumption maintains implicitly that
the two test scores are measured on the same scale.
\begin{assumption}
\label{assu:Proportional}Assume that $E\left[y_{it}^{*}|\mathcal{F}_{n,i,t}\right]=\kappa_{i}$
where $\kappa_{i}$ is a random variable that is invariant to $t$
and is measurable with respect to $\mathcal{Z}_{n}^{*}$ . In addition,
assume that the random variable $y_{it}^{*}$ is either bounded, or
there exists an $\eta>0$ and a random variable $y$ such that $\left|y_{it}^{*}\right|+\left|\kappa_{i}\right|\leq y$
with $E\left[\left|y\right|^{4+\eta}|\mathcal{Z}_{n}^{*}\right]\leq K_{y}<\infty$
for all $i$ and $t.$ In addition, $E\left[\left|\alpha_{c}\right|^{2+\eta}\right]\leq K_{\alpha}<\infty$
and $E\left[\left|\varphi_{cr}\right|^{2+\eta}\right]\leq K_{\varphi}<\infty$
for all $c$ and $r$.
\end{assumption}
The interpretation of the condition is that conditional on a hypothetical
or actual baseline of test results $\zeta^{*}$, as well as information
about group formation $S$, unobserved class level characteristics
$\alpha$, as well as measurement errors $\varphi$ , additional test
results do not change expected performance. A second implication of
Assumption \ref{assu:Proportional} is that expected performance does
not depend on the specific test $t.$ In other words, $\kappa_{i}$
only varies by individual $i$ but not by test $t.$ It is in this
sense that tests $t=1$ and $t=2$ measure essentially identical skills.
Note that $E\left[y_{it}^{*}|\mathcal{Z}_{n}^{*}\right]=\kappa_{i}$
by iterated expectations and $\mathcal{Z}_{n}^{*}\subset\mathcal{F}_{n,i,t}.$
If $y_{it}^{*}$ is correlated with measurement errors $\varphi_{c}$
then this dependence is absorbed into $\kappa_{i}$. The restriction
Assumption \ref{assu:Proportional} imposes on the relationship between
measurement error and $y_{it}^{*}$ is that any type of mean dependence
is common to both tests $t$. A special case covered by Assumption
\ref{assu:Proportional} is that $y_{it}^{*}$ may be independent
of measurement error.

Relative to a baseline represented by $\mathcal{Z}_{n}^{*}$ the latent
test results $y_{it}^{*}$ are in essence as good as randomly assigned.
Conditional on observed covariates as well as the information in $\mathcal{Z}_{n}^{*}$,
actual outcomes $y_{it}$ which are determined by the educational
production function in \eqref{eq:Ed_ProdF}, then are only correlated
in the cross-section because of peer effects, i.e. when $\rho\neq0.$

The variable $\kappa_{i}$ is a function of the entire cross section
of latent performance measures $\zeta^{*}$, of the allocation $S$
and unobserved characteristics that include classroom effects and
measurement errors $\varphi$. It is generally cross-sectionally dependent
in ways we do not restrict or specify and captures such common effects
as the socioeconomic background of students and classmates, their
family background including parental education and support, their
interaction with peers, the exposure to their teachers and resources
available in their classroom. It also depends on a student's own ability.
The restriction we impose is that these factors do not change between
tests $t=1$ and $t=2$.

Baseline performance $\zeta^{*}$ does not necessarily have to be
observed or realized and $\zeta^{*}$ can alternatively be interpreted
as latent ability prior to taking the actual tests $t=1,2$. In this
sense, $\kappa_{i}$ and $E\left[y_{it}^{*}|\mathcal{Z}_{n}^{*}\right]$
are also unobserved. A scenario where Assumption \ref{assu:Proportional}
is realistic arises when the same type of test is taken multiple times
without additional training between iterations of the test, or when
several tests are given that focus on related skills, as is the case
in Project STAR data for grades K-3 and tests for reading, and word
study skills.\footnote{Another example of a test with multiple scores is the Comprehensive
Testing Program (CTP) administered by the Educational Records Bureau
and given to Grades 1-11. The CTP test consists of several main categories
such as ``Auditory Comprehension'', ``Reading Comprehension''
and ``Mathematics'' with each category consisting of additional
subcategories that each receive separate scores.} The assumption is less plausible if comparisons are attempted across
different subjects in a high school or college setting, for tests
given at different times during the school year, or measuring different
skills in a professional or team setting.

The following theorem formally establishes a decomposition of $y_{it}^{*}$
into its conditional mean, $E\left[y_{it}^{*}|\mathcal{F}_{n,i,t}\right]$,
that is dependent cross-sectionally as well as across tests, and into
uncorrelated idiosyncratic noise $u_{it}$.
\begin{thm}
\label{thm:decomp} Suppose Assumption \ref{assu:Proportional} holds.
Define $u_{it}=y_{it}^{*}-E\left[y_{it}^{*}|\mathcal{F}_{n,i,t}\right]=y_{it}^{*}-\kappa_{i},$
then by construction $E\left[u_{it}|\mathcal{F}_{n,i,t}\right]=0$
and $\Cov\left(u_{it},u_{js}|\mathcal{F}_{n,i,t}\right)=0$ for any
$j\neq i$ or $t\neq s$. Furthermore, let $\mathcal{G}_{n,i,t}$
be any sigma field with $\mathcal{G}_{n,i,t}\subseteq\mathcal{F}_{n,i,t}$,
then $E\left[u_{it}|\mathcal{G}_{n,i,t}\right]=0$ and $\Cov\left(u_{it},u_{js}|\mathcal{G}_{n,i,t}\right)=0$
for any $j\neq i$ or $t\neq s$, and $\Cov\left(u_{it},\kappa_{i}\right)=0$
etc. If $E\left[y_{it}^{*}-\kappa_{i}|\mathcal{G}_{n,i,t}\right]\neq0$
for some $\mathcal{G}_{n,i,t}$, then Assumption \ref{assu:Proportional}
cannot hold.
\end{thm}
Examples for the information sets $\mathcal{G}_{n,i,t}$ in the above
theorem are $\mathcal{G}_{n,i,t}=\sigma\left(S,\alpha,\varphi,\zeta^{*},y_{-t}^{*},y_{-i,t}^{*s}\right)$
where $y_{-i,t}^{*s}$ is a subvector of $y_{-i,t}^{*}$ or $\mathcal{G}_{n,i,t}=\mathcal{Z}_{n}^{*}.$
The above theorem makes clear that the main restriction we impose
on $y_{it}^{*}$ is the assumption that latent test performance is
not predictable conditional on observable characteristics and unobserved
baseline ability. The properties $E\left[u_{it}|\mathcal{Z}_{n}^{*}\right]=0,$
$\Cov\left(u_{it},u_{js}|\mathcal{Z}_{n}^{*}\right)=0$ for any $j\neq i$
or $t\neq s$ and $\Cov\left(u_{it},\kappa_{i}\right)=0$ follow directly
from Assumption \ref{assu:Proportional} and are not additional assumptions
imposed on the distribution of $y_{it}^{*}.$ In other words, Assumption
\ref{assu:Proportional} is necessary and sufficient for the decomposition
of $y_{it}^{*}=\kappa_{i}+u_{it}$ with uncorrelated errors $u_{it}$.
Note that $M_{c}$ is implicitly a function of class assignment $S$
and thus orthogonal to $u_{it}.$

We do not attach a specific economic or behavioral interpretation
to the decomposition of $y_{it}^{*}$ into $\kappa_{i}$ and $u_{it}$.
It is merely a statistical representation for the purpose of isolating
variation in test scores that is correlated within groups but invariant
across tests and variation that is idiosyncratic to the individual
student and test. Let $\kappa_{c}=\left(\kappa_{c1},\ldots.\kappa_{cn_{c}}\right)'$,
$u_{ct}=(u_{c1t},\ldots,u_{cn_{c}t})^{\prime}$ and combine Assumptions
\ref{assu:Production-Function} and \ref{assu:Proportional} such
that 
\begin{equation}
y_{ct}=\mu_{c}^{*}+\left(I_{c}+\rho M_{c}\right)u_{ct}\label{eq:Graham-v1}
\end{equation}
with the unobserved class and individual effect defined as 
\begin{equation}
\mu_{c}^{\ast}=\alpha_{c}\mathbf{1}_{c}+\varphi_{c}+\left(I_{c}+\rho M_{c}\right)\kappa_{c}.\label{eq:mu_c}
\end{equation}
 By differencing \eqref{eq:Graham-v1}, which eliminates $\mu_{c}^{\ast}$,
we obtain
\begin{eqnarray}
y_{c1}-y_{c2} & = & \left(I_{c}+\rho M_{c}\right)(u_{c1}-u_{c2}).\label{eq:Graham-2v}
\end{eqnarray}
As seen from the above, differential measures of test scores purge
performance measures from common unobserved influences captured by
$\mu_{c}^{*}$ while still maintaining relevant information that can
shed light on the existence of peer effects. Let $y_{t}=\left(y_{1t}^{\prime},\ldots,y_{Ct}^{\prime}\right)'$
and $u_{t}=\left(u_{1t}^{\prime},\ldots,u_{Ct}^{\prime}\right)'$
for $t=1,2$. In similar fashion define $M=\diag_{c=1}^{C}\left(M_{c}\right)$.
Then we can write \eqref{eq:Graham-2v} for the entire sample more
compactly as 
\begin{equation}
y_{1}=y_{2}+(I+\rho M)(u_{1}-u_{2}).\label{eq:lin_IV-1}
\end{equation}

Our framework is in contrast with much of the econometrics literature
on the identification of peer effects where assumptions about the
orthogonality of $u_{it}$ and $\kappa_{i}$ and covariates are usually
imposed directly. These unobservable quantities often are given economic
interpretations. Examples of papers with such assumptions include
\citet{lee_identification_2007} who assumes i.i.d. errors, \citet{graham_identifying_2008}
who assumes independent classroom and idiosyncratic errors conditional
on class type, \citet{arcidiacono_estimating_2012} who assume uncorrelated
idiosyncratic errors and fixed effects that are orthogonal to idiosyncratic
errors, \citet{rose_identification_2017} who assumes uncorrelated
idiosyncratic errors, \citet{kuersteiner_efficient_2023} who assume
a random effects specification and \citet{lewbel_social_2023} who
assume exogeneity of group formation. Often these assumptions are
justified by random group selection. In our setting we neither assume
that individuals are randomly assigned to groups, nor that they are
randomly selected from a population or that they have otherwise observed
or unobserved characteristics that are independent in the cross-section.

Let $\kappa=\left(\kappa_{1},\ldots,\kappa_{n}\right)'$ and $\tilde{\kappa}=\left(\tilde{\kappa}_{1},\ldots,\tilde{\kappa}_{n}\right)'=S^{-1}\kappa$,
then it is of interest to note that our theory allows for the variables
$\tilde{\kappa}_{i}$ to be arbitrarily influential on the process
that forms groups. An example of such correlation potentially arises
in the Project Star data. Despite random allocation of students and
teachers to classrooms in kindergarten it is possible that parents
who are unhappy with their allocation try and succeed to move their
child to a different class or school, or that those parents provide
additional training for their child. In such a scenario individual
and classroom specific effects $\tilde{\kappa}_{i}$ are not independent
of classroom allocations.

\section{Estimation with Covariates\label{sec:Model-with-Covariates}}

\subsection{Model with Covariates}

We extend the model in Section \ref{Model} by adding observable characteristics
and by specifying the test average scores in a more flexible way.
We collect \textcolor{black}{test invariant and test varying student
level characteristics} for students in class $c$ in the matrices
$v_{c}^{p}$ and $w_{ct}^{p},t=1,2$ \textcolor{black}{with row dimension
}$n_{c}$. We also collect \textcolor{black}{observed test invariant
and test varying class level characteristics for class $c$ which
are not related to student characteristics in the matrices $v_{c}^{c}$
and }$w_{ct}^{c},t=1,2$\textcolor{black}{, with row dimension }$n_{c}$\textcolor{black}{.
Class level characteristics may include class size, the gender, education
and experience of the teacher for subject $t$, etc.} In the following
let \textcolor{black}{$v^{p}=\left(v_{1}^{p\prime},\ldots,v_{C}^{p\prime}\right)^{\prime}$,
$w_{t}^{p}=\left(w_{1t}^{p\prime},\ldots,w_{Ct}^{p\prime}\right)^{\prime}$,
$v^{c}=\left(v_{1}^{c\prime},\ldots,v_{C}^{c\prime}\right)^{\prime}$,
$w_{t}^{c}=\left(w_{1t}^{c\prime},\ldots,w_{Ct}^{c\prime}\right)^{\prime}$,
$t=1,2,$ and let }$X=\left(v^{p},w_{1}^{p},w_{2}^{p},v^{c},w_{1}^{c},w_{2}^{c}\right)$\textcolor{black}{{}
denote the matrix of observations on all covariates.} Below is an
extended version of Assumption \ref{assu:Production-Function}.
\begin{assumption}
\label{assu:Potential_Outcomes-cov}Let $X_{ct}^{c}=\left(v_{c}^{c},w_{ct}^{c}\right)$,
$\beta_{t}^{c}=\left(\beta_{vt}^{c\prime},\beta_{wt}^{c\prime}\right)^{\prime}$,
$X_{ct}^{p}=\left(v_{c}^{p},w_{ct}^{p}\right)$, $\beta_{t}^{p}=\left(\beta_{vt}^{p\prime},\beta_{wt}^{p\prime}\right)'$.
Assume that 
\begin{equation}
y_{ct}=\psi_{ct}\left(y_{ct}^{*}\right)=\alpha_{c}\mathbf{1}_{c}+\varphi_{c}+X_{ct}^{c}\beta_{t}^{c}+\left(I_{c}+\rho M_{c}\right)\left(X_{ct}^{p}\beta_{t}^{p}+y_{ct}^{*}\right),\label{eq:Ed_ProdF-cov}
\end{equation}
where $\mathbf{1}_{c}=\left(1,....,1\right)^{\prime}$ is an $n_{c}\times1$
vector and $M_{c}$ is some matrix of dimension $n_{c}\times n_{c}$
with zero diagonal elements and $\rho$ is a fixed parameter.
\end{assumption}
The vector $M_{c}X_{ct}^{p}\beta_{t}^{p}$ represents contextual peer
effects. We update the filtrations given in Definition \ref{def:Filtration-no-cov}
to include observable covariates.
\begin{defn}
\label{def:Filtration}Define the sigma fields $\mathcal{F}_{n,i,t}=\sigma\left(S,\alpha,\varphi,X,z,\zeta^{\ast},y_{-t}^{\ast},y_{-i,t}^{*}\right)\text{\,for \ensuremath{t=1,2,}}$
where $z$ denotes variables excluded from $X$ that may be used as
instruments. Let $\mathcal{Z}_{n}^{*}=\sigma\left(S,\alpha,\varphi,X,z,\zeta^{\ast}\right)$
be the sigma field of all conditioning variables and let $\mathcal{Z}_{n}=\sigma\left(X,z\right)\subset\mathcal{Z}_{n}^{*}$
be the subset of observable information.
\end{defn}
We maintain Assumption \ref{assu:Proportional} without change, but
using the expanded versions of $\mathcal{F}_{n,i,t}$, $\mathcal{Z}_{n}^{*}$
and $\mathcal{Z}_{n}$ in Definition \ref{def:Filtration}. Then,
under Assumptions \ref{assu:Proportional} and \ref{assu:Potential_Outcomes-cov}
the conclusions of Theorem \ref{thm:decomp} continue to hold. Furthermore,
we have $\Cov\left(u_{it},X\right)=0$ and $\Cov(u_{it},z)=0$. We
note that Assumptions \ref{assu:Proportional}, \ref{assu:Potential_Outcomes-cov}
and Theorem \ref{thm:decomp} imply the peer effects model with covariates
for test score outcomes 
\begin{eqnarray}
y_{ct} & =\mu_{c}^{*} & +v_{c}^{c}\beta_{vt}^{c}+w_{ct}^{c}\beta_{wt}^{c}+\left(I_{c}+\rho M_{c}\right)(v_{c}^{p}\beta_{vt}^{p}+w_{ct}^{p}\beta_{wt}^{p}+u_{ct}).\label{eq:Graham-v3}
\end{eqnarray}
We note that \eqref{eq:Graham-v3} imposes the restriction that average
residualized test scores for both tests $t=1$ and $t=2$ are the
same. This is in line with our core identification strategy of the
two tests testing essentially identical skills.

\subsection{\label{subsec:Empirical-Specification}Empirical Specification}

Our empirical model is a slight generalization of $\eqref{eq:Graham-v3}$
in the sense that we are parameterizing possible differences in test
score averages. The model we estimate in the data has the following
form
\begin{equation}
y_{ct}=\mu_{c}^{*}f_{t}+v_{c}^{c}\beta_{vt}^{c}+w_{ct}^{c}\beta_{wt}^{c}+\left(I_{c}+\rho M_{c}\right)(v_{c}^{p}\beta_{vt}^{p}+w_{ct}^{p}\beta_{wt}^{p}+u_{ct}),\label{eq:Graham-v3e}
\end{equation}
and where our assumptions imply that $f_{t}=1$ for $t=1,2.$ Empirically
we treat $f_{1}$ as a free parameter to be estimated, while we maintain
$f_{2}=1$ as a normalization. An estimated value of $f_{1}$ that
differs from $f_{1}=1$ is an indication, although not a rigorous
specification test, that Assumption $\ref{assu:Proportional}$ is
violated in the data.

By quasi-differencing \eqref{eq:Graham-v3e}, which eliminates $\mu_{c}^{\ast}$,
we obtain
\begin{eqnarray}
y_{c1}-y_{c2}f_{1} & = & v_{c}^{c}(\beta_{v1}^{c}-f_{1}\beta_{v2}^{c})+w_{c1}^{c}\beta_{w1}^{c}-f_{1}w_{c2}^{c}\beta_{w2}^{c}\label{eq:Graham-2v-1}\\
 & + & \left(I_{c}+\rho M_{c}\right)\left(v_{c}^{p}(\beta_{v1}^{p}-f_{1}\beta_{v2}^{p})+w_{c1}^{p}\beta_{w1}^{p}-f_{1}w_{c2}^{p}\beta_{w2}^{p}\right)+\left(I_{c}+\rho M_{c}\right)(u_{c1}-f_{1}u_{c2}).\nonumber 
\end{eqnarray}
Collecting terms with $f_{1}$ leads to an interpretation of \eqref{eq:Graham-2v-1}
where $y_{c2}-v_{c}^{c}\beta_{v2}^{c}-w_{c2}^{c}\beta_{w2}^{c}-\left(I_{c}+\rho M_{c}\right)(v_{c}^{p}\beta_{v2}^{p}+w_{c2}^{p}\beta_{w2}^{p})$
is used to control for unobserved $\mu_{c}^{\ast}$ similar to a control
function approach. The difference to a conventional control function
approach is that \eqref{eq:Graham-v3e} for $t=2$ cannot be consistently
estimated by least squares because unobserved components may be correlated
with observed covariates, a problem that our GMM estimators address.

We next write \eqref{eq:Graham-2v-1} more compactly by stacking observations
across classrooms. Let $v^{p},v^{c}$ and $y_{t},u_{t},w_{t}^{p},w_{t}^{c}$,
$t=1,2$, as well as $M$ be as defined above. Furthermore, let $\underline{X}=\left(v^{c},w_{1}^{c},w_{2}^{c},v^{p},w_{1}^{p},w_{2}^{p},Mv^{p},Mw_{1}^{p},Mw_{2}^{p}\right)$
and define
\begin{equation}
\delta(f_{1},\rho,\beta)=\left(\delta_{v}^{c\prime},\beta_{w1}^{c\prime},-f_{1}\beta_{w2}^{c\prime},\delta_{v}^{p\prime},\beta_{w1}^{p\prime},-f_{1}\beta_{w2}^{p\prime},\rho\delta_{v}^{p\prime},\rho\beta_{w1}^{p\prime},-\rho f_{1}\beta_{w2}^{p\prime}\right)^{\prime},\label{eq:delta}
\end{equation}
where $\delta_{v}^{c}=\beta_{v1}^{c}-f_{1}\beta_{v2}^{c}$, $\delta_{v}^{p}=\beta_{v1}^{p}-f_{1}\beta_{v2}^{p}$
and $\beta$ is the vector of all $\beta$s above. Then we can write
\eqref{eq:Graham-2v-1} for the entire sample more compactly as 
\begin{equation}
y_{1}=f_{1}y_{2}+\underline{X}\delta+(I+\rho M)(u_{1}-f_{1}u_{2}).\label{eq:lin_IV}
\end{equation}
where $\delta=\beta_{1}-f_{1}\beta_{2}$. Here $\beta_{t}$ collects
the corresponding coefficients for each variable in $\text{\ensuremath{\underline{X}}}$
with elements of $\beta_{t}$ set to 0 if a variable is absent in
\eqref{eq:Graham-v3e} for given $t$. Let $\mu^{\ast}=(\mu_{1}^{\ast\prime},....,\mu_{C}^{\ast\prime})^{\prime}$.
Then we can write \eqref{eq:Graham-v3e} compactly as

\begin{equation}
y_{t}=\mu^{\ast}f_{t}+\text{\ensuremath{\underline{X}}\ensuremath{\ensuremath{\beta_{t}}+(I+\ensuremath{\rho}M)\ensuremath{u_{t}}}}.\label{eq:Graham-v1-compact}
\end{equation}

With Assumptions \ref{assu:Proportional} and \ref{assu:Potential_Outcomes-cov},
and full column rank conditions on $\underline{X}$, we can identify
$\delta$, which is sufficient for the identification of $(\rho,f_{1})$.
The parameters $\beta_{w1}^{c}$ and $\beta_{w1}^{p}$ can be directly
recovered from $\delta$, and $\beta_{w2}^{c}$ and $\beta_{w2}^{p}$
can be identified when $f_{1}\neq0$. As is common in pure fixed effect
panel settings, the effects of test invariant covariates $v$, in
our case $\beta_{v1}^{c},\beta_{v2}^{c}$ and $\beta_{v1}^{p},\beta_{v2}^{p}$
are not identified when using within or differencing estimators. In
our case, except when $f_{1}=1,$ we are quasi-differencing the equation
and allowing for heterogeneity in parameters between $t=1$ and $t=2.$
This leads to a formulation of the model where $v^{c}$ and $v^{p}$
enter the quasi, or fully differenced equation. However, without further
assumptions the parameters $\beta_{v1}^{c}$, $\beta_{v2}^{c}$, $\beta_{v1}^{p}$
and $\beta_{v2}^{p}$ are not separately identified. An example of
restrictions where the parameters $\beta_{v1}^{c}$ and $\beta_{v2}^{c}$
are identified arises when $f_{1}\neq1$ and $\beta_{v1}^{c}=\beta_{v2}^{c}$.

\subsection{Estimation}

In empirical applications it is sometimes sufficient to estimate $\beta$
or even just $\delta,$ for example when $\underline{X}$ contains
contextual peer effects. Conventional regression methods applied directly
to \eqref{eq:Graham-v3e} are invalid because unobserved fixed effects
collected in $\mu_{c}^{\ast}$ may be correlated with $\underline{X}.$
We show in Section \ref{sec:Technical} that instrumenting for $y_{2}$
in \eqref{eq:lin_IV} identifies the parameters $\delta$ and $f_{1}.$
We show in Corollary \ref{cor:TS_ID} that when test scores are recorded
as non-negative values the classroom fixed effect or an overall constant
is a valid instrument. More generally, exogenous covariates that vary
at the classroom level but are excluded from $\underline{X}$ are
valid instruments by Lemma \ref{lem:id_nocov}. We consider instrument
matrices $H$ and $A$ and require that all elements of $H$ and $A$
are functions of the observable variables that are measurable with
respect to $\mathcal{Z}_{n}$. Note that by construction, these variables
may depend on the unobserved effects $\alpha$ and $\kappa$. In particular,
we set $H=[\underline{X},z]$ as an $n\times q$ instrument matrix.

A simple two-step procedure is sufficient to obtain consistent but
inefficient estimates for $\delta$, $f_{1}$, and $\rho$. Such a
two-step procedure can be readily defined adapting ideas from the
literature on spatial Cliff-Ord models; cp. e.g., \citet{kelejian_generalized_1998,kelejian_generalized_1999}.
The first step relies on linear moment conditions for the identification
of the parameters $\delta$ and $f_{1}$ which exploit that $H$ and
the disturbance vector $(I+\rho M)(u_{1}-f_{1}u_{2})$ of model \eqref{eq:lin_IV}
are orthogonal. Note that by Theorem \ref{thm:decomp} the linear
moment condition $E\left[H^{\prime}\left(y_{1}-y_{2}f_{1}-\underline{X}\delta\right)\right]=0$
holds at the true parameter values for $\delta$ and $f_{1}.$ For
simplicity consider the case where $z$ is of dimension $n\times1$.
Then this equation corresponds to a regular just identified linear
instrumental variables problem with included exogenous covariates
$\underline{X}$ and where $z$ is instrumenting for the endogenous
variable $y_{2}$. Letting $W=\left(y_{2},\text{\ensuremath{\underline{X}}}\right)$
and solving the sample analog of the moment condition yields $\left(\tilde{f}_{1},\tilde{\delta}'\right)'=\left(H^{\prime}W\right)^{-1}H^{\prime}y_{1}$.

The estimator for the coefficient $f_{1}$ has the familiar form of
just identified two stage least squares with included covariates 
\begin{equation}
\tilde{f}_{1}=\frac{z'Q_{X}y_{1}}{z'Q_{X}y_{2}}\label{eq:f_tilde}
\end{equation}
where $Q_{X}=I-\underline{X}(\text{\ensuremath{\underline{X}}}^{\prime}\text{\ensuremath{\underline{X}}})^{-1}\text{\ensuremath{\underline{X}}}^{\prime}$
is the residual operator of a projection onto $\underline{X}$. With
$\tilde{f}_{1}$ determined by \eqref{eq:f_tilde} $\tilde{\delta}$
is given by 
\begin{equation}
\tilde{\delta}=(\underline{X}'\underline{X})^{-1}\underline{X}^{\prime}(y_{1}-\tilde{f}_{1}y_{2}).\label{eq:delta_tilde}
\end{equation}
In the second step we estimate $\rho$ using a form of the quadratic
moment function with, for example, $A=M$. Quadratic moments are motivated
by Theorem \ref{thm:decomp} which indicates that the variance-covariance
matrix of residual vector $\epsilon^{+}=u_{1}-f_{1}u_{2}$ is diagonal.
Hence $E(\epsilon^{+\prime}A\epsilon^{+})=0$ if $A$ has all diagonal
elements 0. This suggests plugging $\tilde{\delta}$ and $\tilde{f}_{1}$
into $\epsilon^{+}\left(\rho,f_{1},\delta\right)=\left(I+\rho M\right)^{-1}\left(y_{1}-f_{1}y_{2}-\underline{X}\delta\right)$
and forming the quadratic moment vector $m_{\epsilon}^{q}\left(\rho\right)=\frac{1}{n}\epsilon^{+}\left(\rho,\tilde{f}_{1},\tilde{\delta}\right)^{\prime}A\epsilon^{+}\left(\rho,\tilde{f}_{1},\tilde{\delta}\right)$.
The first stage estimator for $\rho$ then is obtained as 
\begin{equation}
\tilde{\rho}=\argmin_{\rho}\left(m_{\epsilon}^{q}\left(\rho\right)\right)^{2}.\label{eq:rho_tilde}
\end{equation}
We prove identification in Section \ref{sec:Technical}. Consistency
and the asymptotic distribution of $\tilde{\theta}=\left(\tilde{\rho},\tilde{f}_{1},\tilde{\delta}^{\prime}\right)'$are
discussed in Sections \ref{sec:Consistency} and \ref{sec:AsyNormal}
in the online appendix together with a discussion of the efficient
estimator $\hat{\theta}$ defined below.

The estimator $\tilde{\theta}$ is inefficient as it is constructed
from the residuals $u_{1}-f_{1}u_{2}$ in \eqref{eq:lin_IV}, which
are not homoscedastic. To develop efficient estimators, denote the
variance-covariance matrix of $u_{1}-f_{1}u_{2}$ as $\Omega(\gamma)$,
where $\gamma$ are parameters to be specified later. Theorem \ref{thm:decomp}
indicates that $\Omega(\gamma)$ is a diagonal matrix. The form of
$\Omega(\gamma)$ and $\gamma$ are discussed under Assumption \ref{assu:PAR}
as well as Assumption \ref{assu:CondVar} in the online appendix.

Premultiplying \eqref{eq:lin_IV} by $V\left(\rho,\gamma\right)=\Omega(\gamma)^{-1/2}(I+\rho M)^{-1}$
then yields the following equation 
\begin{equation}
y^{+}(\theta,\gamma)=X^{+}(\theta,\gamma)\delta+u^{+}\label{eq:uplus-1}
\end{equation}
where $y^{+}(\theta,\gamma)=V\left(\rho,\gamma\right)(y_{1}-f_{1}y_{2}),$
$X^{+}(\theta,\gamma)=V\left(\rho,\gamma\right)\underline{X},$ and
$u^{+}=\Omega(\gamma)^{-1/2}(u_{1}-f_{1}u_{2})$ and $\Var\left(u^{+}|\mathcal{Z}_{n}^{*}\right)=I_{n}$.
In light of \eqref{eq:uplus-1} we define 
\begin{equation}
\text{\ensuremath{u^{+}\left(\theta,\gamma\right)}}=V\left(\rho,\gamma\right)(y_{1}-f_{1}y_{2}-\text{\ensuremath{\underline{X}}\ensuremath{\ensuremath{\delta}}})\label{eq:uplus-1-1}
\end{equation}
and observe that $u^{+}(\theta,\gamma)=u^{+}$ at the true parameters.
The transformation leading to $u^{+}$ consists of three components.
Quasi-differencing eliminates classroom and individual effects. The
operator $(I+\rho M)^{-1}$ removes cross-sectional correlation due
to peer effects and can be understood as a form of spatial Cochrane-Orcutt
transformation. Finally, the operator $\Omega(\gamma)^{-1/2}$ scales
the spatially uncorrelated residuals to unit variance in the case
of heteroskedasticity.

We next provide an outline of our GMM estimation methodology. For
clarity we explicitly denote in the following the true parameter vectors
as $\theta_{0}=(\rho_{0},f_{1,0},\delta_{0}')'$ and $\gamma_{0}$.
Our GMM estimator\footnote{Rigorous consistency and asymptotic normality results as well as additional
assumptions needed to establish these properties are given in Sections
\ref{sec:Consistency} and \ref{sec:AsyNormal} in the online appendix.} is based on both linear and quadratic moment conditions.

The moment function consists of a set of linear and quadratic moments
for the transformed residuals\footnote{The specification can be readily extended to include several quadratic
moment conditions along the lines of, e.g., \citet{kuersteiner_dynamic_2020}.} 
\begin{equation}
m_{n}(\theta,\gamma)=\left(\begin{array}{c}
m_{n}^{l}(\theta,\gamma)\\
m_{n}^{q}(\theta,\gamma)
\end{array}\right)=n^{-1/2}\left(\begin{array}{c}
H^{\prime}u^{+}(\theta,\gamma)\\
u^{+\prime}(\theta,\gamma)Au^{+}(\theta,\gamma)
\end{array}\right),\label{eq:Moments_SI_Graham}
\end{equation}
where the functions $m_{n}^{l}(\theta,\gamma)$ and $m_{n}^{q}(\theta,\gamma)$
denote the linear and quadratic moment conditions respectively. Under
the maintained assumptions of this paper the results in \citet{kuersteiner_dynamic_2020},
Theorem 1, imply that the linear moment function is uncorrelated with
the quadratic moments. We have 
\begin{align}
E\left[m_{n}(\theta_{0},\gamma_{0})m_{n}(\theta_{0},\gamma_{0})^{\prime}\right] & =\frac{1}{n}\left(\begin{array}{cc}
E\left[H^{\prime}H\right] & 0\\
0 & 2E\left[\tr\left(A^{2}\right)\right]
\end{array}\right)\equiv\Xi_{n}.\label{eq:Xi_n}
\end{align}
Therefore, our optimal weight matrix is the inverse of $\Xi_{n}.$
Consistent estimates of the elements of $\Xi_{n}$ can be obtained
as follows. For $V_{n}^{h}=\frac{1}{n}$$E\left[H^{\prime}H\right]$
we can use the estimator $\hat{V_{n}^{h}}=\frac{1}{n}H^{\prime}H$
and for $V_{n}^{a}=\frac{1}{n}$$E\left[\tr\left(A^{2}\right)\right]$
we can use the estimator $\hat{V_{n}^{a}}=\frac{1}{n}\tr\left(A^{2}\right)$,
and thus $\Xi_{n}$ can be estimated consistently by $\hat{\Xi}_{n}=diag\{\hat{V_{n}^{h}},2\hat{V_{n}^{a}}\}=n^{-1}diag\{H^{\prime}H,2\tr\left(A^{2}\right)\}$.
The empirical criterion function $Q_{n}\left(\theta,\gamma\right)=\frac{1}{n}m_{n}(\theta,\gamma)^{\prime}\hat{\Xi}_{n}^{-1}m_{n}(\theta,\gamma)$
of the GMM estimator then can be written as 
\begin{align}
Q_{n}\left(\theta,\gamma\right) & =\frac{1}{n}u^{+}(\theta,\gamma)^{\prime}H\left(H^{\prime}H\right)^{-1}H^{\prime}u^{+}(\theta,\gamma)+\frac{1}{n}\frac{\left(u^{+}(\theta,\gamma)'Au^{+}(\theta,\gamma)\right)^{2}}{2\tr\left(A^{2}\right)}\label{eq:Def_Qn}
\end{align}
such that the GMM estimator for $\theta_{0}$ is defined as $\hat{\theta}\left(\gamma\right)=\argmin_{\theta}Q_{n}\left(\theta,\gamma\right)$
for a fixed value of $\gamma.$ In applications $\gamma$ is replaced
with a first step consistent estimator.

\section{Peer Effects in Project STAR\label{sec:Empirical-Results}}

\subsection{Data\label{subsec:Data}}

In this section, we apply our identification strategy to the Tennessee's
Project STAR (student-teacher achievement ratio) data, exploring peer
effects among students in Kindergarten through Grade 3.\footnote{We use STAR data obtained from Harvard Dataverse \nocite{achilles_tennessee_2008}(https://dataverse.harvard.edu/dataset.xhtml?persistentId=hdl:1902.1/10766),
which differs slightly from the version used by Graham(2008) and Rose(2017).
For example, they construct class ID using teacher characteristics
while our data has class ID readily available. Their kindergarten
sample has 6,172 students, and ours has 6,325. We prefer our data
as it matches the description in Word et al. (1990) in terms of sample
size and summary statistics.} Project STAR is a randomized experiment aiming at studying the impact
of class size reduction on children's development. The data set has
been widely used in studies of class size, peer effects, teacher effectiveness
and other education-related topics (e.g., \citealt{krueger_experimental_1999,dee_teachers_2004,graham_identifying_2008,chetty_how_2011,mueller_teacher_2013}).
Here we briefly discuss the details which are relevant for our study.
For a complete description of the project and data, see \citet{word_state_1990,mosteller_tennessee_1995,boyd-zaharias_project_2007}.

Project STAR was carried out in Tennessee from 1985 to 1989, with
additional data on the participants collected after the project ended.
The experiments followed the kindergarten cohort of 1985 in participating
schools for four years, from Kindergarten to the third grade. At the
start of the 1985 academic year, participating schools randomly assigned
kindergarten students and teachers into small classes , regular classes
and regular classes with a full-time teacher\textquoteright s aide.
In the actual implementation, small classes included 11 to 20 students,
while regular classes (with or without aides) included 15 to 30 students.
Randomness of the initial class assignment is ensured by the careful
implementation of the random assignment protocol and has been examined
by a number of studies, e.g., \citet{krueger_experimental_1999,chetty_how_2011}.
However, nonrandom attrition, switching and migration might have happened
in higher grades \citep{hanushek_findings_1999}. Our analysis uses
the sample of each of the four grades, as our estimator remains consistent
even without random assignment.\footnote{A Table with summary statistics for each grade can be found in the
online appendix.} While there are 11601 students in the experiment in total, the sample
size in each grade is between 6325 and 6840 due to students migrating
in or out of the participating schools. The total count of schools
is 79 at the Kindergarten level, but it decreases to 76 at the Grade
1 level and subsequently to 75 at the Grade 2-3 level. This decline
can be attributed to schools withdrawing from the STAR program. There
are between 325 and 340 classes in each grade, with 124 to 140 of
them being small classes.

Project STAR data contains tests scores for mathematics, reading,
listening and word study skills. We use raw Stanford Achievement Tests
(SAT) scores in word study skills (word) and reading (read) as test
$t=1$ and test $t=2$ respectively in our model, as they measure
closely related skills and hence are more likely to satisfy our assumptions.
This choice is supported by diagnostic checks we report in Tables
\ref{tab:corr} \ref{tab:lagy} and Table \ref{tab:corr_all} in the
appendix. Reversing the order of $y_{1}$ and $y_{2}$ does not change
the estimate of $\rho$. Results using scores in mathematics (math)
and listening (list) are discussed in the online appendix as a comparison.
SAT is a nationally standardized test with scores comparable across
grades. The test was administered on state specified examination days,
occurring between late March and early April each year \citep{word_state_1990,krueger_experimental_1999}.
The narrow time frame for the tests lends credence to the assumption
that student's ability and preparation remain stable over the testing
window. The reading and word study skills scores have means that are
close in all grades.

Our preferred model includes a set of student and teacher characteristics.
Between $56.0\%$ and $59.2\%$ of the students receive a free lunch
during the experimental period, an indicator of low family income.
To account for the influence of race, we employ an indicator for black
students, as more than 98\% of the minority students are of black
ethnicity. The share of black students ranges from $32.6\%$ to 34.7\%.
Girls make up $48.0\%$ to $48.6\%$ of the sample. We calculate age
as of April 1st in each year, which is roughly the SAT examination
date, using students' dates of birth. Teacher's characteristics are
weighted by class size. Between 16.5\% and 20.9\% of the students
are taught by black teachers. Meanwhile, 34.6\% to 44.2\% of the students
have teachers with a master's degree or higher. Average years of experience
of the teachers for different grades are between 9.3 and 13.9.

In our preferred specification, we include four types of control variables:
school fixed effects (excluding one dummy for the first school), class
type fixed effects (excluding the dummy for small classes), student
characteristics (including indicators for free lunch, black ethnicity,
female gender, and age) and teacher characteristics (including indicators
for black teachers, having master's degree or higher, and years of
experience) and peer characteristics -- specifically, the leave-out-mean
of student characteristics. We start from the specification with only
school fixed effects and subsequently introduce other types of control
variables incrementally. We find that our results are not very sensitive
to the inclusion of additional controls. This is reassuring in that
our differencing strategy of closely related scores removes most of
the systematic variation in scores that is related to observed and
unobserved individual and classroom level factors affecting outcomes.
We allow for heteroscedasticity of $u$ between small classes and
regular classes (with or without teacher's aide). Standard errors
are clustered at the class level.

\subsection{Plausibility of Assumptions\label{subsec:Plausibility-of-Assumptions}}

Before proceeding to discussing the main estimates of peer effects,
we present several summary statistics in support of our model specifications
and assumptions. Our estimator relies on the premise that word score
$y_{1}$ and reading score $y_{2}$ measure similar skills within
a brief time frame. To document the close connection between test
scores, we use four measures to evaluate the correlation between $y_{1}$
and $y_{2}$ from Grade K to Grade 3 in Table \ref{tab:corr}. The
first two measures on the left panel are Spearman correlations between
$y_{1}$ and $y_{2}$, and between $Q_{X}y_{1}$ and $Q_{X}y_{2}$.
Here $Q_{X}=I-\underline{X}(\underline{X}'\underline{X})^{-1}\text{\ensuremath{\underline{X}}}'$
denotes the residual projection matrix of $\text{\underbar{X}}$,
where $\text{\underbar{X}}$ encompasses the complete set of control
variables: school fixed effects, class type fixed effect, student
characteristics, teacher characteristics, and peer characteristics.
The right panel of Table \ref{tab:corr} reports the pseudo $R^{2}$
from 2SLS, i.e., the Spearman correlation between $y_{1}$ and $\hat{y}_{1}$,
or between $y_{2}$ and $\hat{y}_{2}$. Here $\hat{y}_{1}=\tilde{f}_{1}y_{2}+\underline{X}\tilde{\delta}$,
with $\tilde{f_{1}}$ and $\tilde{\delta}$ defined in \eqref{eq:f_tilde}
and \eqref{eq:delta_tilde} being the 2SLS estimator when $y_{2}$
is instrumented by 1. Similarly, $\hat{y}_{2}$ corresponds to the
counterpart obtained by reversing the order of $y_{1}$ and $y_{2}$.\footnote{We also adjust for missing covariate values using a similar method
as the one for the GMM estimator. That is replacing the peer averages
of the observed individuals $\bar{x}_{c(-r)t}^{obs}=(n_{c}\bar{x}_{ct}-x_{crt}^{obs})/(n_{c}-1)$
with $\tilde{x}_{c(-r)t}^{obs}=(n_{c}\bar{x}_{ct}^{obs}-x_{crt}^{obs})/(n_{c}-1)$,
where $\bar{x}_{ct}^{obs}$ is the average of all observed $x$.These
adjustments make little difference to the results. We drop individuals
with missing outcome variables from the sample. We discuss the procedure
in more detail in Section \ref{sec:Randomly_Missing} of the online
appendix.} We find that the word study skills and reading scores exhibit substantial
correlations of approximately 0.9 across all four measures and in
all grades. This finding is consistent with these scores measuring
skills that are both related to language comprehension. On the other
hand, correlations between other pairs are smaller and range between
0.5 and 0.8, as shown in Table \ref{tab:corr_all} in the Appendix.
While a key identifying assumption, Assumption \ref{assu:Proportional}
is not directly testable because it relates to unobserved ability,
finding strong correlation between test scores is consistent with
a setting where test scores have common unobserved individual and
class effects that may be confounding measures of peer effects at
the classroom level. Our model implies that conditional on covariates
correlation between tests only is due to unobserved fixed effects
$\mu_{c}^{\ast}$ and peer effects working through idiosyncratic errors
$u_{it}.$ While not a test for correct specification, our results
are consistent with $\mu_{c}^{\ast}$ accounting for a large fraction
of the variation in test scores.

We next explore an additional approach\footnote{Two more specification checks are reported in the online appendix.}
to evaluate the plausibility of Assumption \ref{assu:Proportional}.
If unobserved individual and class effects $\mu_{c}^{\ast}$ are stable
over time then scores from previous years should be good predictors
for current year scores. Here, it is important to note that our theory
applies to pairs of closely related tests, and that $\mu_{c}^{\ast}$
is not assumed to be the same for different pairs. Thus our theory
does not assume $\mu_{c}^{\ast}$ to be time invariant between grades.
Conceptually, $\zeta^{*}$ measures unobserved ability just prior
to taking the tests in the current year. As a result, $\mu_{c}^{\ast}$
similarly reflects only current year unobserved ability and other
classroom level effects coming from teachers and peers. Unobserved
effects $\mu_{c}^{\ast}$ are expected to vary year by year for various
reasons, including the fact that instruction during the year improves
skills, that teachers may change and that students may move to different
classes. However, it is expected that at least some components of
$\mu_{c}^{\ast}$ that are more closely related to individual student
ability are time invariant. We examine this hypothesis by measuring
the explanatory power of lagged scores on current scores. The regression
sample includes stayers in each grade, defined as individuals participating
in the STAR project for both the previous and the current grades to
ensure availability of both current and lagged scores. The Grade 1
sample includes 66\% stayers, while Grade 2 and 3 samples have 73.8\%
and 79.6\% respectively. Among the stayers in Grade 1, 88.7\% stay
in the same type of classes (small v.s. regular(w/wo) aide) in the
same schools. This proportion exceeds 93\% for Grade 2 and Grade 3.\footnote{Based on the authors' calculation using the STAR data.}
The share is lower if we consider regular and regular/aid as different
types, due to a random reallocation of students in regular and regular/aid
classes to these two types within school in Grade 1. However, the
literature generally finds no significant impact of teacher's aide
(e.g., \citealp{krueger_experimental_1999}). The stability of school
and class types among the large share of stayers fosters a steady
learning environment. While newcomers and leavers might marginally
influence peer quality, the allocation of newcomers to different classrooms
is random within schools. Consequently, when accounting for school
fixed effects, we anticipate a consistent level of peer quality over
time.

The relatively stable environment that these students are in suggests
that $\mu_{c}^{\ast}$ should contain significant time invariant components
that in turn imply high correlation between observed scores and lagged
scores. As a result we expect to see relatively high $R^{2}$ in predictive
regressions with lagged scores as controls in addition to the observed
characteristics. Consequently, we regress test outcomes on the full
set of control variables (school and class type fixed effects, students,
teacher and peer characteristics and a constant term), and then add
lagged outcomes of the same or different type and their peer averages
to assess changes in $R^{2}$. The results are summarized in Table
\ref{tab:lagy}. The results show the predictive power of lagged scores
for the same test in bold, as well as when using other tests instead,
for grades one through three. The predictive power of lagged scores
is documented by comparing the $R^{2}$ with and without adding lags.
For example, for the word score, adding its own lag improves the $R^{2}$
from 0.244 to 0.435 in grade one, from 0.25 to 0.539 in grade two
and from 0.229 to 0.488 in grade three. Using the lagged reading score
leads to increases in the $R^{2}$ that are even slightly higher.
Overall, the coefficients on lagged scores when using the own lag
of word as well as the lagged reading score are large in magnitude
and highly statistically significant. The results are further evidence
that the data distributions are consistent with key features of our
model.

Table \ref{tab:f1rho} presents estimates and standard errors (in
parentheses) for $f_{1}$ (left panel) and the peer effects parameter
$\rho$ (right panel). \footnote{Results for contextual peer effects are reported in Table \ref{tab:ave}
in the online appendix.} Estimates for other pairs of scores are in Table \ref{tab:f1rho_all}
in the online appendix. These estimates are obtained by our efficient
GMM estimator for all grades and for specifications with just a school
effect in Column (1), school and class type effects in Column (2),
with a full set of controls excluding contextual peer effects in Column
(3) and the full set of controls in Column (4).

The point estimates of $f_{1}$ are remarkably stable across all specifications
and across all grades. The size of the estimated coefficient $f_{1}$
is close to one and estimated precisely. This is further evidence
that word and reading scores measure closely related skills and that
the test score differences control unobserved classroom and individual
effects. Assumption \ref{assu:Proportional} postulates that unobserved
individual specific and classroom level effects captured by $\mu_{c}^{\ast}$
affect both scores in the same way. This is credible when both scores
measure essentially the same skills. Our empirical finding of $f_{1}$
close to one for the word-read pair is consistent with this interpretation
of the data. Other score pairings reported in the appendix have estimated
$f_{1}$ parameters that are less reliably close to one, lending further
support for our choice of focusing on word-read as our preferred specification.

\subsection{Peer Effects\label{subsec:Main_Results}}

To provide some context and interpretation of $\rho$ consider the
education production function in Equation \eqref{eq:Ed_ProdF-cov}.
When $\rho=0$ such that there are no peer effects, it follows that
the observed scores are $y_{ct}=\alpha_{c}f_{t}\mathbf{1}_{c}+X_{ct}^{c}\beta_{t}^{c}+(X_{ct}^{p}\beta_{t}^{p}+y_{ct}^{*})$
where $y_{ct}^{*}$ has the interpretation of an unobserved term that
is measured in the same units as the test score variable. When $\rho$
is different from zero test score performance of student $i$ changes
by $\rho$ times the leave-out-average of $y_{ct}^{*},$ through the
term $\rho M_{c}y_{ct}^{*}.$ This interpretation implies that $\rho$
measures the effect on student $i$'s test score as a fraction of
the increase in peer ability measured in units of test score performance.
Similarly, the effect on test scores as a result of changes in peer
characteristics are measured as $\rho\beta^{p}$ times the changes
in relevant peer characteristics.

The comparison of the four specifications demonstrates the robustness
of our estimators to additional controls. We now focus on Specification
(4), which includes the full set of control variables. Across all
grade levels, our results consistently reveal positive and statistically
significant peer effects, ranging from 0.2 to 0.4. The results we
find are relatively stable across the four specifications in Columns
(1)-(4). The results for Grades $K$, two and three are quite similar,
while results for Grade one are smaller in magnitude but remain statistically
very significant. The implications of these results are that improvements
in peer quality measured in terms of potential SAT scores $y_{t}^{*}$
translate at the rate of 20 to 40 percent to improvements of individual
scores or in other words a 10 point increase in the peer potential
SAT score results in a 2 to 4 point increase in individual SAT scores.
These effects are at the lower end of the spectrum of results reported
in the literature.

The comparison with other results in the literature is complicated
by the fact that often a model with endogenous peer effects is estimated.
To focus ideas, consider the case without covariates and only one
outcome measure $y_{1}$. Then, an endogenous peer effect specification
for class room $c$ is $y_{ct}=\lambda M_{c}y_{ct}+u_{ct}.$ If peer
effects are measured with full rather than leave-one-out means, i.e.,
if $M_{c}=n^{-1}\mathbf{1}_{c}\mathbf{1}_{c}^{\prime}$, this formulation
is identical to our model, which is $y_{ct}=\left(I+\rho M_{c}\right)u_{ct}$
in this simplified stylized setting. The parameters $\rho$ and $\lambda$
satisfy the one-to-one mapping $\lambda=\rho/(1+\rho)$ which is a
consequence of the fact that $\left(I+\rho M_{c}\right)=\left(I-\lambda M_{c}\right)^{-1}$.
The inverse captures multiplier effects of peer performance inherent
in the endogenous peer effect formulation, while $\rho$ is a summary
measure of all these multiplier effects. When peer effects are measured
by leave-one-out rather than full means the relationship between the
two formulations is no longer exact. However, it can be shown that
$\left(I+\rho M_{c}\right)=\left(I-\lambda M_{c}\right)^{-1}$ continues
to hold approximately with an error that is $O\left(n_{c}^{-2}\right)$
where $n_{c}$ is the size of class $c$. Based on this approximation,
we convert our measure of $\rho$ to an endogenous peer effect $\lambda=\rho/(1+\rho),$
or equivalently, $\rho=\lambda/\left(1-\lambda\right)$ for the purpose
of the discussion that follows. Our estimates for $\rho$ in the range
of $0.2$ to $0.4$ then translate to values for $\lambda$ that range
from $0.16$ to $0.28$.

In comparison with existing results for peer effects employing STAR
data, our estimates are relatively small in magnitude but estimated
with high precision and statistical significance. Using normalized
kindergarten SAT scores, \citet{graham_identifying_2008} derives
endogenous peer effects of 0.46 and 0.56 for math and reading respectively,
which translates to $\rho=0.86$ or $\rho=1.30$ in our context.\footnote{\citet{graham_identifying_2008} does not directly estimate endogenous
peer effects. His model is $y_{ic}=\alpha_{c}+(\gamma-1)\bar{\epsilon}_{c}+\epsilon_{ic}$,
with $y_{ic}$ being the outcome of student $i$ in class $c$ and
$\bar{\epsilon}_{c}$ is the full mean of unobserved ability $\epsilon_{ic}$
in class $c$. Hence $\gamma-1$ is comparable with our estimate $\rho$,
both are coefficients of average peer unobserved ability. Due to the
full-mean specification, the equation is equivalent to $y_{ic}=\lambda\bar{y}_{c}+(1-\lambda)\alpha_{c}+\epsilon_{ic}$,
where $\lambda=1-\frac{1}{\gamma}$ is the endogenous peer effects.} Using normalized Kindergarten SAT math score as outcome, \citet{rose_identification_2017}
finds an endogenous peer effect of 0.65 or a correlated effect (measured
by coefficients of the peer average errors) of 0.71. When both are
considered simultaneously, endogenous and correlated effects are 0.90
and -0.03, albeit with standard errors over 1.7, indicating weak identification.
\citet{boozer_black_2001}, using average percentiles of math and
reading scores as outcomes and an instrumental variable methods for
estimation, identify endogenous peer effects of 0.30, 0.86 and 0.92
for Grades 1 to 3 respectively, with standard deviations of 1.0, 0.12
and 0.04. \citet{sojourner_identification_2013} uses first-grade
averages of SAT percentiles of math, reading and listening as outcome
measure, and checks the exogenous peer effects of lagged peer outcomes.
He finds an estimate of about 0.35 with a standard deviation around
0.14.

The observed peer effects in our study exhibit smaller magnitudes
compared to most of the results discussed above. It is important to
note that estimates in \citet{lewbel_social_2023}, \citet{rose_identification_2017}
and \citet{boozer_black_2001} pertain to coefficients of peer's average
contemporaneous scores, resulting in a multiplier effect as explained
above. In contrast, our estimated peer effects are coefficients of
average peer unobserved ability. This notion is more akin to the ``correlated
effects'' discussed in \citet{rose_identification_2017}. Overall,
we find a noticeable difference in peer effects' magnitude between
our study and those of \citet{rose_identification_2017} and \citet{graham_identifying_2008}.

Our proposed estimator leads to highly precise estimates compared
to all the aforementioned estimates, except that of \citet{lewbel_social_2023}
who also report precise estimates. The enhancement in precision can
be attributed in part to the elimination of individual fixed effects
through differencing. Certain individual characteristics that have
been recognized for their significant influence on scores are not
observed in the data for Project STAR. These include parent's education,
family income and most importantly unobserved IQ. Our estimator, by
removing unobserved class effects and individual fixed effects, effectively
eliminates residual variance associated with these characteristics
which may be one explanation for the precise estimates we obtain.
In addition, our GMM estimators use additional quadratic moment restrictions
efficiently.

The aforementioned studies are often limited to a particular grade.
Methods that rely on random peer assignment are typically constrained
in where they can be applied. In project STAR the random assignment
assumption is more defensible in earlier grades \citep{boozer_black_2001,graham_identifying_2008,sojourner_identification_2013}.
In contrast, our analysis does not require random assignment and facilitates
the estimation of peer effects across grade levels. Encouragingly,
our findings consistently fall within a reasonably narrow range, affirming
the robustness and reliability of our method.

\section{Identification \label{sec:Technical}}

We next discuss identification of our model. To establish identification
we introduce additional assumptions imposed on the parameter space.
We call the model identified at the parameter value $\theta_{0}$
if the population moment conditions underlying the formulation of
our estimators are only satisfied at $\theta=\theta_{0}$ for all
$\theta\in\Theta$. Throughout this section we treat $f_{1}$ as an
unknown parameter to be estimated. To accommodate the empirical specification
in Section \ref{subsec:Empirical-Specification} and to justify the
specification checks reported in Section \ref{subsec:Plausibility-of-Assumptions}
we establish that the parameter $f_{1}$ in \eqref{eq:Graham-v3e}
can be consistently estimated irrespective of whether $f_{1}=1$ holds
or not. For this purpose we modify Assumption \ref{assu:Proportional}
to enlarge the class of models we consider to cases where $f_{1}\neq1.$
\begin{assumption}
\label{assu:Prop_Gen}Assume that Assumption \ref{assu:Proportional}
holds except that $E\left[y_{it}^{*}|F_{n,i,t}\right]=\kappa_{i}f_{t}$
for $t=1,2$ where $\kappa_{i}$ is measurable with respect to $\mathcal{Z}_{n}^{*}$
and $f_{t}$ is a constant.
\end{assumption}
Assumption \ref{assu:Prop_Gen} implies that the correctly specified
model is \eqref{eq:Graham-v3e} where we impose the normalization
$f_{2}=1$ without loss of generality. When in addition $f_{1}=1$
then the correctly specified model reduces to \eqref{eq:Graham-v3}.
It is easy to see that the conclusions of Theorem \ref{thm:decomp}
continue to hold under Assumption \ref{assu:Prop_Gen} when $u_{it}=y_{it}^{*}-\kappa_{i}f_{t}.$
The next assumption defines the parameter space.
\begin{assumption}
\label{assu:PAR} We use subscript zero to denote true parameters
and treat $\delta_{0}$, ignoring the restrictions given in \eqref{eq:delta},
as an unrestricted parameter vector. Let $\theta_{0}=\left(\rho_{0},f_{1,0},\delta_{0}^{\prime}\right)^{\prime}$
with $\delta_{0}=\delta(f_{1,0},\rho_{0},\beta_{0})$. Then assume
$\theta_{0}\in\Theta$, where $\Theta$ is a compact subset of $\Theta_{0}=\left(-1,1\right)\times(-K_{f},K_{f})\times(-K_{X},K_{X})^{p_{x}}$,
and where $K_{f}$ and $K_{X}$ are finite positive constants and
$p_{x}$ is the column dimension of $\underline{X}$. In addition,
suppose there are $J$ class types denoted by $\tau_{c}\in\{1,\ldots,J\},$
and $\tau_{c}$ is assumed to be measurable w.r.t. $\mathcal{Z}_{n}$.
Furthermore, let \textup{$\Omega(\gamma)=\diag_{c}^{C}\left\{ \gamma_{\tau_{c}}^{2}I_{c}\right\} $
with $\gamma=(\gamma_{1},\ldots,\gamma_{J})^{\prime},$ then we assume
that }$\Var(u_{1}-f_{1,0}u_{2}|\mathcal{Z}_{n}^{*})=\Omega(\gamma_{0})$
for some $\gamma_{0}=(\gamma_{1,0},\ldots,\gamma_{J,0})^{\prime}\in\Gamma$
with $\Gamma=[K_{\gamma}^{-1},K_{\gamma}]^{J}$ where $K_{\gamma}>1$
is a finite constant.
\end{assumption}
The assumption implies that the parameter space for $\rho$ is a subset
of $[-K_{\rho},K_{\rho}]$ for some positive constant $K_{\rho}<1$.
The conditions postulated for $\Omega(\gamma)$ are satisfied under
Assumption \ref{assu:CondVar} in Appendix \ref{sec:Error_Term} regarding
$u_{t}$. In addition, we impose the the following assumption on $A$.
\begin{assumption}
\label{assu:quaA}(i) The matrix $A=(a_{ij})=\diag_{c=1}^{C}(A_{c})$
is measurable with respect to $\mathcal{Z}_{n}$ and has zero diagonal
elements. Moreover, $A_{c}$ is a function of $M_{c}$ and is symmetric.

(ii) $\mathbf{1}_{c}^{\prime}A_{c}\mathbf{1}_{c}/n_{c}\geq\underline{K}_{a}>0$
for some constant \textup{$\underline{K}_{a}$.}

(iii) $\sup_{i}\sum_{j=1}^{n}\left\vert a_{ij}\right\vert \leq K_{a}<\infty$
for some constant $K_{a}$.
\end{assumption}
Two valid choices for $A$ are $A=M$ and $A=M^{\prime}M-\diag\left(M^{\prime}M\right).$
\begin{assumption}
\label{assu:size}Group size $n_{c}$ satisfies $2\leq n_{c}\leq\bar{n}<\infty$
for some constant \textup{$\bar{n}$}.
\end{assumption}
Before discussing the general case, we start discussing identification
of our model in the case where no additional covariates are present.
Note that in the case of no covariates, the coefficients $\delta$
of the covariates are omitted, $\theta=(\rho,f_{1},\delta^{\prime})^{\prime}$
is reduced to $(\rho,f_{1})^{\prime}$ and \eqref{eq:uplus-1-1} becomes
$u^{+}(\theta,\gamma)=y^{+}(\theta,\gamma)$.

We next discuss identification for a given matrix $H=z\equiv\left(z_{1}^{\prime},...,z_{c}^{\prime},...,z_{C}^{\prime}\right)^{\prime}$
of $\mathcal{Z}_{n}$ measurable instruments, where the $z_{c}$ are
$n_{c}\times q$ matrices of random variables that are invariant at
the classroom level, i.e., $z_{c}=\mathbf{1}_{c}\grave{z}_{c}$, where
$\mathbf{1}_{c}$ is an $n_{c}\times1$ vector with all elements equal
to one and $\grave{z}_{c}$ is a $1\times q$ vector of class level
characteristics. We focus on the just identified case with $q=1$
for ease of exposition. Additional overidentifying restrictions can
be added but do not affect arguments related to identification. Consider
the moment vector $m_{n}\left(\theta,\gamma\right)$ defined in \eqref{eq:Moments_SI_Graham}
where $A$ is a matrix that satisfies the properties listed in Assumption
\ref{assu:quaA}.

Observe that $\Omega(\gamma)=\diag_{c=1}^{C}(\gamma_{\tau_{c}}^{2}I_{c})$
is diagonal for all admissible $\gamma$ and that $\Omega_{0}=\Omega(\gamma_{0})=\Var(u_{1}-f_{1,0}u_{2}|\mathcal{Z}_{n})$
in light of Assumptions \ref{assu:Prop_Gen} and \ref{assu:PAR} and
Theorem \ref{thm:decomp}. Also, $u^{+}(\theta_{0},\gamma)=\Omega(\gamma)^{-1/2}\Omega_{0}^{1/2}u^{+}.$
Thus it is readily seen that for all $A$ satisfying Assumption \ref{assu:quaA}(i)
we have $E\left[m_{n}\left(\theta_{0},\gamma\right)\right]=0$, where
the expectation of the quadratic moments is $\tr\left(A\Omega(\gamma)^{-1}\Omega_{0}\right)$
which is zero because diagonal elements of $A\Omega(\gamma)^{-1}\Omega_{0}$
are zero.

We next show that $E\left[m_{n}\left(\rho,f_{1},\gamma\right)\right]=0$
if and only if $(\rho,f_{1})=(\rho_{0},f_{1,0})$. Identification
of $f_{1,0}$ can be established from an inspection of the linear
moment condition. Identification of the social interaction parameter
$\rho_{0}$ follows from an analysis of the quadratic moment condition.
\begin{lem}
\label{lem:id_nocov} As a special case of \eqref{eq:Graham-v1-compact}
without covariates consider the data generating process $y_{t}=\mu^{\ast}f_{t,0}+\left(I_{n}+\rho_{0}M\right)u_{t}$
where $\mu^{*}$ is defined in \eqref{eq:mu_c}. Suppose Assumptions
\ref{assu:Potential_Outcomes-cov}, \ref{assu:Prop_Gen}, \ref{assu:PAR},
\ref{assu:quaA}, and \ref{assu:size} hold, let $m_{n}\left(\theta,\gamma\right)$
be defined in (\ref{eq:Moments_SI_Graham}) and let $H=z$ where $z$
is invariant at the classroom level. In addition, assume that $q=1$
and 
\begin{equation}
n^{-1}\left|E\left[y_{2}^{\prime}\Omega(\gamma)^{-1/2}z\left|\mathcal{Z}_{n}\right.\right]\right|\geq K_{y}>0\label{eq:Graham_ID_Cond}
\end{equation}
 for all $\gamma\in\Gamma$ and all $n$, where $K_{y}$ is a constant.
Then for all $\gamma\in\Gamma$ we have $E\left[m_{n}\left(\rho_{0},f_{1,0},\gamma\right)\left|\mathcal{Z}_{n}\right.\right]=0$
and $E\left[m_{n}\left(\rho,f_{1},\gamma\right)\left|\mathcal{Z}_{n}\right.\right]\neq0$
a.s. for all $\left(\rho,f_{1}\right)\neq\left(\rho_{0},f_{1,0}\right)$
and $\gamma\in\Gamma$.
\end{lem}
A proof of Lemma \ref{lem:id_nocov} establishing the identifiability
of $\left(\rho_{0},f_{1,0}\right)$ is given in the online appendix.
The subsequent corollary considers a special case of the lemma that
is of interest in the context of test scores, and arises when $H=\mathbf{1}_{n}$
and when test scores are non-negative. The latter is typically the
case for non-normalized test scores.
\begin{cor}
\label{cor:TS_ID}Suppose the assumptions of Lemma \ref{lem:id_nocov}
hold, $H=\mathbf{1}_{n}$, $y_{t}$$\geq0$, and suppose condition
(\ref{eq:Graham_ID_Cond}) is replaced by $n^{-1}\sum_{c=1}^{C}E\left[y_{c2}^{\prime}\mathbf{1}_{c}\left|\mathcal{Z}_{n}\right.\right]>K_{y}>0$
for some constant $K_{y}.$ Then the conclusions of Lemma \ref{lem:id_nocov}
hold.
\end{cor}
We emphasize further that no restrictions on $\mu^{*}$ (other than
the fact that $\mu^{*}$ does not depend on $t$ and $\mu^{*}\neq0$
a.s.) are needed to identify $\rho_{0}$. Also, the analysis in the
appendix shows that neither knowledge of, nor restrictions on heteroskedasticity
are necessary for the identification of $\theta_{0}$. Thus, initial
estimates of $\theta_{0}$ can be obtained with $\gamma_{j}^{2}=1$
for all $j=1,...,J$.

Next we turn to a discussion of identification of $\theta=(\rho,f_{1},\delta^{\prime})^{\prime}$
when the model contains covariates. We use instruments $H=\left[\underline{X},z\right]$
for linear moment restrictions and a matrix $A$ for quadratic moment
restrictions, with $A$ defined as before. Both $H$ and $A$ are
$\mathcal{Z}_{n}$-measurable and $\diag(A)=0$. For ease of exposition
we focus again on the case where $z$ is a vector. Recall that with
covariates, $u^{+}(\theta,\gamma)$ is defined in \eqref{eq:uplus-1-1}.
We consider the $(q+1)\times1$ vector of linear and quadratic moment
functions $m_{n}(\theta,\gamma)$ defined in \eqref{eq:Moments_SI_Graham}.

Note that our procedure is not designed to identify $\gamma$, nor
is that required for the identification of $\theta$. Theorem \ref{thm:general_id}
below shows that for all permissible $\gamma$ the moment condition
satisfies $E\left[m_{n}(\theta_{0},\gamma)\left|\mathcal{Z}_{n}\right.\right]=0$
and $E\left[m_{n}(\theta,\gamma)\left|\mathcal{Z}_{n}\right.\right]\neq0$
for $\theta\neq\theta_{0}$. The parameters $\gamma$ can be pinned
down by estimating the variance of $u_{1}-f_{1,0}u_{2}$ for different
types of groups. For identification, we further impose the following
assumption.
\begin{assumption}
\label{assu:rank-new}(i)The absolute values of the elements $h_{ij}$
of $H=\left[\underline{X},z\right]$ are uniformly bounded in $i,j,n$
by some positive constant $K_{H}<\infty$.

(ii) The smallest eigenvalue of $\underline{X}'\underline{X}/n$ is
bounded from below by some $\xi_{X}>0$ uniformly in $n$.

(iii) Define $V(\rho,\gamma)=\Omega(\gamma)^{-1/2}(I+\rho M)^{-1}$
and let 
\[
Q_{V^{1/2}X}(\rho,\gamma)=I-V(\rho,\gamma)^{1/2}\underline{X}\left(\underline{X}^{\prime}V(\rho,\gamma)\underline{X}\right){}^{-1}\underline{X}^{\prime}V(\rho,\gamma)^{1/2}
\]
be the projection matrix onto the orthogonal complement of the column
space of $V^{1/2}\underline{X}$, then
\begin{equation}
\inf_{\gamma\in\Gamma,\rho\in[-K_{\rho},K_{\rho}]}n^{-1}\left|E\left[(V(\rho,\gamma)^{1/2}z)^{\prime}Q_{V^{1/2}X}(\rho,\gamma)V(\rho,\gamma)^{1/2}y_{2}\left|\mathcal{Z}_{n}\right.\right]\right|>K_{y}>0\text{ a.s.}.\label{eq:Rank_Cond-3}
\end{equation}
\end{assumption}
Note that the elements of the matrix $H$ are by construction measurable
w.r.t. $\mathcal{Z}_{n}=\sigma(X,z)$. Since $\underline{X}$ contains
$X$, part (i) of the assumption implies that also the elements of
$X$ are bounded in absolute value by $K_{H}$. Observing that $(\underline{X},z)$
has full column rank by Assumption \ref{assu:rank-new}(iii), as setting
$\rho=0$ and $\gamma=(1,...,1)^{\prime}$ in $V(\rho,\gamma)$, \eqref{eq:Rank_Cond-3}
becomes $n^{-1}\left|E\left[z^{\prime}Q_{X}y_{2}|\mathcal{Z}_{n}\right]\right|>K_{y}>0$,
where $Q_{X}=I-\underline{X}\left(\underline{X}^{\prime}\underline{X}\right){}^{-1}\underline{X}^{\prime}$.
It is readily seen from an inspection of the proof of Theorem \ref{thm:general_id}
below that a necessary condition for Assumption \ref{assu:rank-new}(iii)
is that $\left(\underline{X},E\left[\mu^{*}|\mathcal{Z}_{n}\right]\right)$
has full column rank. The latter condition is equivalent to the condition
that $\left(\underline{X},E\left[y_{2}|\mathcal{Z}_{n}\right]\right)$
has full column rank, observing that in light of \eqref{eq:Graham-v1-compact}
we have $E\left[y_{2}|\mathcal{Z}_{n}\right]=E\left[\mu^{*}+\underline{X}\beta_{2}|\mathcal{Z}_{n}\right]$
for some vector of parameters $\beta_{2}.$

Condition (\ref{eq:Rank_Cond-3}) generalizes condition (\ref{eq:Graham_ID_Cond})
in Lemma \ref{lem:id_nocov}. Under homoskedasticity or when estimators
are not taking heteroskedasticity into account, the latter condition
takes the more familiar form $n^{-1}\left|E\left[y_{2}^{\prime}z|\mathcal{Z}_{n}\right]\right|>K_{y}>0$.
Further simplifications of \eqref{eq:Rank_Cond-3} are obtained if
in addition to heteroskedasticity, spatial correlation is ignored
by setting $\rho=0.$ Considering this case is useful when discussing
first step estimators for $\delta$ and $f_{1}.$ For the case with
covariates \eqref{eq:Rank_Cond-3} simplifies to $n^{-1}\left|E\left[z^{\prime}Q_{X}y_{2}|\mathcal{Z}_{n}\right]\right|>K_{y}>0$,
where $Q_{X}$ is defined as above. Both versions of the simplified
identification conditions are conventional rank conditions for linear
instrumental variables estimators. A second scenario under which the
conditions can be simplified is when $\gamma$ is evaluated at a fixed
constant level. For example, when $\gamma=\bar{\gamma}$ then \eqref{eq:Graham_ID_Cond}
is replaced by $n^{-1}\left|E\left[y_{2}^{\prime}\Omega(\bar{\gamma})^{-1/2}z\left|\mathcal{Z}_{n}\right.\right]\right|>K_{y}>0$
and a similar simplification obtained for \eqref{eq:Rank_Cond-3}.
The following theorem is a generalization of Lemma \ref{lem:id_nocov}.
\begin{thm}
\label{thm:general_id}Let Assumptions \ref{assu:Potential_Outcomes-cov},
\ref{assu:Prop_Gen}, \ref{assu:PAR}, \ref{assu:quaA}, \ref{assu:size}
and \ref{assu:rank-new} be satisfied. Let $m_{n}\left(\theta,\gamma\right)$
be defined as in \eqref{eq:Moments_SI_Graham}. Then for all admissible
$\gamma\in\Gamma,$ $E\left[m_{n}\left(\theta_{0},\gamma\right)\left|\mathcal{Z}_{n}\right.\right]=0$
and $E\left[m_{n}\left(\theta,\gamma\right)\left|\mathcal{Z}_{n}\right.\right]\neq0$
for $\theta\neq\theta_{0}$ a.s.
\end{thm}
We report results that establish the consistency and asymptotic normality
of our estimators in an online appendix where we also discuss consistent
estimation of $\gamma$ and standard errors.

\section{Conclusions\label{sec:Conclusions}}

We identify peer effects using differential scores of closely related
subject tests. Our analysis and identification strategy is based on
a variance decomposition that distinguishes between test invariant
components and unobserved idiosyncratic effects that are not predictable
using concurrent test performance in related tests. The additional
parametric assumption of linear education production functions and
peer effects allows to identify a common endogenous peer effects parameter.
Our method depends on the availability of closely related, yet separate
test results but has the advantage of being robust to non-random group
assignment. We prove identification of model parameters under conditions
that are typical in the linear instrumental variables literature.
These results depend on a careful exploration of the parametric structure
of linear peer effects models. Our empirical results for Project STAR
data and classrooms in kindergarten to the third grade show highly
significant peer effects, but smaller in magnitude than those found
in a number of studies that have looked at the same data.\pagebreak{}

{\setstretch{1.0}\bibliographystyle{jpe}
\bibliography{cite}

}

\pagebreak{}

 \begin{table}[htpb] \begin{center} \caption{Correlation between Word and Reading Scores} \label{tab:corr}  \begin{threeparttable}  \begin{tabular}{cccccc} \hline \hline & \multicolumn{2}{c}{Rank Correlation} & & \multicolumn{2}{c}{Pseudo $ R^2 $ of 2SLS} \\ \cline{2-3}\cline{5-6} & $ y_1,y_2 $ & $ Q_X y_1,Q_X y_2 $ & & $ y_1,\hat{y}_1 $ & $ y_2,\hat{y}_2 $ \\ \hline 
Grade K&0.908&0.881&&0.915&0.909\\
Grade 1&0.932&0.905&&0.936&0.940\\
Grade 2&0.884&0.848&&0.892&0.894\\
Grade 3&0.890&0.869&&0.901&0.895\\
 \hline \hline \end{tabular} \begin{tablenotes} \footnotesize
\item 1. The left panel is Spearman's rank correlation between $ y_1 $, $ y_2 $ and between $ Q_{X}y_{1} $, $ Q_{X}y_{2} $, where $ Q_{X}=I-\underbar{X}(\underbar{X}^{\prime}\underbar{X})^{-1}\underbar{X}^{\prime} $, $ y_{1} $ and $ y_{2} $ are raw SAT scores in word and reading, $ \underbar{X} $ are the complete set of control variables, including school fixed effects, class type fixed effects, student characteristics, teacher characteristics and peer characteristics. \item 2. The right panel is pseudo $ R^{2} $, i.e., Spearman's rank correlation between $ y_{1} $ and $ \hat{y}_{1} $ in Column 3, and  between $ y_{2} $ and $ \hat{y}_{2} $ in Column 4. Here $ \hat{y}_{1}=\tilde{f}_{1}y_{2}+\text{\ensuremath{\underbar{X}}}\tilde{\delta} $, with $ \tilde{f_{1}} $ and $ \tilde{\delta} $ defined in (\ref{eq:f_tilde}) and (\ref{eq:delta_tilde}), $ \underbar{X} $ is the complete set of controls as described in note 1. That is,  $ \hat{y}_{1} $ is the predicted value of $ y_{1} $ from 2SLS for $ y_{1}=f_{1}y_{2}+\underbar{X}\delta+\epsilon $ with $ y_2 $ instrumented by the constant term,  $ \hat{y}_{2} $ is defined in a similar manner by reversing the order of $ y_1, y_2 $.
\end{tablenotes} \end{threeparttable} \end{center} \end{table}

 \begin{table}[htpb] \begin{center} \caption{Explanatory Power of Lagged Scores} \label{tab:lagy}  \begin{threeparttable}  \begin{tabular}{cccccccc} \hline \hline & no lagged score & & \multicolumn{2}{c}{lagged read} & & \multicolumn{2}{c}{lagged word} \\ \cline{4-5} \cline{7-8} & $ R^2 $ & & coefficient & $ R^2 $ & & coefficient & $ R^2 $ \\ \hline 
\multicolumn{8}{c}{\textit{dependent variable:reading score}} \\
Grade 1&0.290&&\textbf{1.026(0.032)}&\textbf{0.544}&&0.809(0.024)&0.522\\
Grade 2&0.294&&\textbf{0.622(0.011)}&\textbf{0.648}&&0.582(0.013)&0.588\\
Grade 3&0.246&&\textbf{0.671(0.010)}&\textbf{0.644}&&0.512(0.010)&0.544\\
\multicolumn{8}{c}{\textit{dependent variable:word score}} \\
Grade 1&0.244&&0.860(0.033)&0.453&&\textbf{0.686(0.025)}&\textbf{0.435}\\
Grade 2&0.250&&0.668(0.012)&0.588&&\textbf{0.631(0.014)}&\textbf{0.539}\\
Grade 3&0.229&&0.737(0.013)&0.594&&\textbf{0.547(0.012)}&\textbf{0.488}\\
 \hline \hline \end{tabular} \begin{tablenotes} \footnotesize
\item 1. Results for the regression $ y_i=\alpha + X_i\gamma+ly_i \beta_1 +avely_i \beta_2 +\epsilon_i$, where $ y_i $ is the reading score in the top panel and word score in the bottom panel. X is a set of control variables including school fixed effects, class type fixed effects, student, teacher, and peer characteristics in the main specification. $ ly_i $ is the lagged reading score (columns 2 and 3) or lagged word score (columns 4 and 5), $ avely_i $ is the peer average of $ ly_i $. Results are bold when $ y_i $ and $ ly_i $ are of the same type.  \item 2. Column 1 is the R squared of the regression when no lagged variables are added. Columns 2 and 4 are coefficients and standard errors of lagged score. Columns 3 and 5 are R squared of the regression. \item 3.  Coefficients of peer average lagged scores, typically small and statistically insignificant, are omitted for conciseness.
\end{tablenotes} \end{threeparttable} \end{center} \end{table}

 \begin{table}[htpb] \begin{center} \caption{Estimates of $ f_1 $  and $ \rho $ for the word-reading pair} \label{tab:f1rho}  \begin{threeparttable}  \begin{tabular}{cccccccccc} \hline \hline & \multicolumn{4}{c}{$ f_1 $} & & \multicolumn{4}{c}{$ \rho $} \\ \cline{2-5}\cline{7-10} & (1) & (2) & (3) & (4) & & (1) & (2) & (3) & (4) \\ \hline 
Grade K&1.005&1.005&1.016&0.960&&0.452&0.449&0.443&0.424\\
&(0.004)&(0.005)&(0.010)&(0.051)&&(0.078)&(0.078)&(0.080)&(0.075)\\
Grade 1&0.992&0.992&0.986&0.975&&0.178&0.178&0.190&0.193\\
&(0.005)&(0.005)&(0.009)&(0.043)&&(0.054)&(0.054)&(0.054)&(0.053)\\
Grade 2&1.008&1.007&1.034&1.076&&0.399&0.398&0.363&0.348\\
&(0.004)&(0.004)&(0.008)&(0.035)&&(0.065)&(0.064)&(0.063)&(0.064)\\
Grade 3&0.987&0.987&1.012&0.985&&0.391&0.387&0.377&0.379\\
&(0.002)&(0.002)&(0.006)&(0.031)&&(0.078)&(0.077)&(0.077)&(0.076)\\
\hline school FE & Y & Y & Y & Y & & Y & Y & Y & Y \\ class type FE & & Y & Y & Y & & & Y & Y & Y \\ stu\&tch char & & & Y & Y & & & & Y & Y \\ peer char & & & & Y & & & & & Y\\ \hline \hline \end{tabular} \begin{tablenotes} \footnotesize
\item 1. Estimates and standard errors (in the parentheses) for $ f_1 $  and $ \rho $. Standard errors are clustered at the classroom level. All models allow for heteroscedasticity across small and regular classes. Estimates are adjusted for missing observations. \item 2. The scores are raw SAT scores in word study skills for $ y_1 $ and reading for $ y_2 $. \item 3. As indicated at the bottom of the table, the models may control for (1) school fixed effects (excluding one); (2) class type fixed effects (excluding small); (3) a student's characteristics (free lunch, black, girl, age) and teacher characteristics (black teacher, master, years of experience); (4) peer averages of student characteristics.
\end{tablenotes} \end{threeparttable} \end{center} \end{table}
\pagebreak{}

\appendix

\section*{Appendix}

\section{Proofs}
\begin{proof}[Proof of Theorem \ref{thm:decomp}]
 By the assumption on the boundedness of the $\left(4+\eta\right)$-th
moments all considered expectations and conditional expectations exist.
Since $u_{it}=y_{it}^{*}-E\left[y_{it}^{*}|\mathcal{F}_{n,i,t}\right]$
we have $E\left[u_{it}|\mathcal{F}_{n,i,t}\right]=0$ by construction.
Observing that $E\left[y_{it}^{*}|\mathcal{F}_{n,i,t}\right]=\kappa_{i}$
and that $\kappa_{i}$ is $\mathcal{Z}_{n}^{*}$-measurable it follows
that conditional on $\mathcal{Z}_{n}^{*}$ knowledge of $y_{it}^{*}$
is equivalent to knowledge of $u_{it}$. To see this let $u_{-i,t}=\left(u_{1t},..,u_{i-1,t},u_{i+1,t},...,u_{n,t}\right)'$
and use the notation $u_{-t}$ for the unobserved components related
to the other test. Define $U=\left(S,\alpha,\varphi,\zeta^{*},u_{-t},u_{-i,t},\kappa_{i}\right)$
assuming for simplicity of notations that vectors are row vectors,
and let $V=\left(S,\alpha,\varphi,\zeta^{*},y_{-t}^{*},y_{-i,t}^{*},\kappa_{i}\right)$.
Then $V=g(U)$ where $g(.)$ is Borel measurable and thus $\sigma(V)\subseteq\sigma(U)$.
Analogously $U=h(V)$ where $h(.)$ is Borel measurable and thus $\sigma(U)\subseteq\sigma(V)$.
Consequently $\sigma(U)=\sigma(V)$. Since $\kappa_{i}$ is measurable
w.r.t. $\sigma(S,\alpha,\varphi,\zeta^{*})$ it follows that 
\[
\sigma(S,\alpha,\varphi,\zeta^{*},u_{-t},u_{-i,t},\kappa_{i})=\sigma(S,\alpha,\varphi,\zeta^{*},u_{-t},u_{-i,t})
\]
 and $\sigma(S,\alpha,\varphi,\zeta^{*},y_{-t}^{*},y_{-i,t}^{*},\kappa_{i})=\sigma(S,\alpha,\varphi,\zeta^{*},y_{-t}^{*},y_{-i,t}^{*}),$
which shows that 
\[
\mathcal{B}_{n,i,t}=\sigma\left(S,\alpha,\varphi,\zeta^{*},u_{-t},u_{-i,t}\right)=\sigma\left(S,\alpha,\varphi,\zeta^{*},y_{-t}^{*},y_{-i,t}^{*}\right)=\mathcal{F}_{n,i,t}.
\]
 Next observe that the $u_{jt}$ are $\mathcal{F}_{n,i,t}$-measurable
for $i\neq j$, and hence by iterated expectations $\Cov\left(u_{it},u_{jt}|\mathcal{F}_{n,i,t}\right)=E\left[u_{jt}E\left(u_{it}|\mathcal{F}_{n,i,t}\right)\right]=0$.
Observe further that for $s\neq t$ the $u_{js}$ are $\mathcal{F}_{n,i,t}$-measurable
for all $i$, and hence by iterated expectations $\Cov\left(u_{it},u_{js}|\mathcal{F}_{n,i,t}\right)=E\left[u_{js}E\left[u_{it}|\mathcal{F}_{n,i,t}\right]\right]=0$.
The corresponding claims regarding the mean and covariances of the
$u_{it}$ conditional on $\mathcal{G}_{n,i,t}\subseteq\mathcal{F}_{n,i,t}$
follow immediately by iterated expectations. This proves the first
part of the theorem.

To prove the second part, let $\mathcal{G}_{n,i,t}\subseteq\mathcal{F}_{n,i,t}$
be some information set. Then given that Assumption \ref{assu:Proportional}
holds, it follows by iterated expectations that
\begin{eqnarray*}
\text{E\ensuremath{\left[y_{it}^{*}-\kappa_{i}|\mathcal{G}_{n,i,t}\right]}} & = & E\left[E\left[y_{it}^{*}|\mathcal{F}_{n,i,t}\right]|\mathcal{G}_{n,i,t}\right]-E\left[\kappa_{i}|\mathcal{G}_{n,i,t}\right]=E\left[\kappa_{i}|\mathcal{G}_{n,i,t}\right]-E\left[\kappa_{i}|\mathcal{G}_{n,i,t}\right]=0.
\end{eqnarray*}
Thus by the almost sure uniqueness of conditional expectations, Assumption
\ref{assu:Proportional} cannot hold if $E\left[y_{it}^{*}-\kappa_{i}|\mathcal{G}_{n,i,t}\right]\neq0$.
\end{proof}
Note that Lemma \ref{lem:id_nocov} and Corollary \ref{cor:TS_ID}
can be viewed as special cases of Theorem \ref{thm:general_id}. We
therefore only give an explicit proof of Theorem \ref{thm:general_id}
here. Proofs of the additional results are given in the online appendix.\footnote{Equations labeled 32 and higher, as well as Appendices B and higher
can be found in the online appendix.}
\begin{proof}[Proof of Theorem \ref{thm:general_id}]
 In light of \eqref{eq:Graham-v1-compact} it follows that,
\begin{equation}
y_{1}-f_{1}y_{2}=\mu^{\ast}(f_{1,0}-f_{1})+\underline{X}\delta(f_{1},\rho_{0},\beta_{0})+(I+\rho_{0}M)(u_{1}-f_{1}u_{2}),\label{eq:ydif}
\end{equation}
where $\delta(f_{1},\rho_{0},\beta_{0})$ is defined in \eqref{eq:delta}
and $\delta(f_{1,0},\rho_{0},\beta_{0})=\delta_{0}$. Consequently,
\begin{equation}
y_{1}-f_{1}y_{2}-\text{\underbar{X}}\delta=\mu^{\ast}(f_{1,0}-f_{1})+\underline{X}\left(\delta(f_{1},\rho_{0},\beta_{0})-\delta\right)+(I+\rho_{0}M)(u_{1}-f_{1}u_{2}).\label{eq:ydif2}
\end{equation}
Let $\mathcal{M}(\rho)=(I+\rho M)^{-1}(I+\rho_{0}M)$. Substitution
of \eqref{eq:ydif2} into \eqref{eq:uplus-1-1} yields
\begin{eqnarray}
u^{+}(\theta,\gamma) & = & \Omega(\gamma)^{-1/2}(I+\rho M)^{-1}\label{eq:uplus0-1}\\
 &  & \times\left\{ \left(f_{1,0}-f_{1}\right)\mu^{\ast}+\underline{X}\left(\delta(f_{1},\rho_{0},\beta_{0})-\delta\right)+(I+\rho_{0}M)(u_{1}-f_{1}u_{2})\right\} \nonumber \\
 & = & V(\rho,\gamma)\left(\begin{array}{cc}
\underline{X}, & \mu^{*}\end{array}\right)\left(\begin{array}{c}
\delta(f_{1},\rho_{0},\beta_{0})-\delta\\
f_{1,0}-f_{1}
\end{array}\right)+\Omega(\gamma)^{-1/2}\mathcal{M}(\rho)(u_{1}-f_{1}u_{2}),
\end{eqnarray}
recalling that $V(\rho,\gamma)=\Omega(\gamma)^{-1/2}(I+\rho M)^{-1}$
.

We first analyze the linear moment conditions. Note that 
\begin{equation}
H^{\prime}u^{+}(\theta,\gamma)=H^{\prime}V(\rho,\gamma)\left(\begin{array}{cc}
\underline{X}, & \mu^{*}\end{array}\right)\left(\begin{array}{c}
\delta(f_{1},\rho_{0},\beta_{0})-\delta\\
f_{1,0}-f_{1}
\end{array}\right)+H^{\prime}\Omega(\gamma)^{-1/2}\mathcal{M}(\rho)(u_{1}-f_{1}u_{2}).\label{eq:ml}
\end{equation}
Let 
\begin{eqnarray}
K_{n}(\rho,\gamma) & = & n^{-1}H^{\prime}V(\rho,\gamma)\begin{array}{cc}
(\underline{X}, & \mu^{*})\end{array}.\label{eq:Kn}
\end{eqnarray}
Observe that 
\begin{eqnarray}
n^{-1}E\left[H^{\prime}u^{+}(\theta,\gamma)\left|\mathcal{Z}_{n}\right.\right] & = & E\left(K_{n}(\rho,\gamma)\left|\mathcal{Z}_{n}\right.\right)\left(\begin{array}{c}
\delta(f_{1},\rho_{0},\beta_{0})-\delta\\
f_{1,0}-f_{1}
\end{array}\right).\label{eq:mom-1}
\end{eqnarray}
where the r.h.s. is seen to hold in light of \eqref{eq:ml} and \eqref{eq:Kn},
and since $H$ is $\mathcal{Z}_{n}$-measurable and thus
\[
E\left[H^{\prime}\Omega(\gamma)^{-1/2}\mathcal{M}(\rho)(u_{1}-f_{1}u_{2})\left|\mathcal{Z}_{n}\right.\right]=0.
\]
The assumptions imply that $K_{n}$ has full rank a.s. To see this,
note that in light of \eqref{eq:I-rhoM_IJ_Rep}, \eqref{eq:Omega_IJ_Rep},
Lemma \ref{lem:mat_I-J}(iv) in online Appendix \ref{sec:Auxiliary-Lemmas}
implies that
\begin{equation}
V(\rho,\gamma)=\diag_{c=1}^{C}\left(\frac{1}{\gamma_{\tau_{c}}}\frac{n_{c}-1-\rho}{n_{c}-1}I_{c}^{\ast}+\frac{1}{\gamma_{\tau_{c}}}\left(1+\rho\right)J_{c}^{\ast}\right),\label{eq:V}
\end{equation}
Under Assumptions \ref{assu:PAR} and \ref{assu:size}, there exists
some constants $c_{e}$ and $C_{e}$ such that $0<c_{e}<\frac{1}{\gamma_{\tau_{c}}}\frac{n_{c}-1-\rho}{n_{c}-1}<C_{e}<\infty$
and $0<c_{e}<\frac{1}{\gamma_{\tau_{c}}}\left(1+\rho\right)<C_{e}<\infty$.
In light of Lemma \ref{lem:mat_I-J}(iii), the eigenvalues of $V(\rho,\gamma)$
are uniformly bounded below by $c_{e}>0$. By Assumption \ref{assu:rank-new}
we have $\lambda_{min}\left(n^{-1}\underline{X}^{\prime}\underline{X}\right)\geq\xi_{X}>0$
and thus $\lambda_{min}\left(n^{-1}\underline{X}^{\prime}V(\rho,\gamma)\underline{X}\right)\geq c_{e}\xi_{X}>0$
uniformly for all $\gamma$, $\rho$, and $n$.

Recalling that $H=\left(\underline{X},z\right)$ and recalling the
expression for $K_{n}(\theta,\gamma)$ given in (\ref{eq:Kn}) the
matrix can be written as 
\begin{align}
K_{n}(\rho,\gamma) & =\frac{1}{n}\left(\begin{array}{cc}
\underline{X}^{\prime}V(\rho,\gamma)\underline{X} & \underline{X}^{\prime}V(\rho,\gamma)\mu^{*}\\
z^{\prime}V(\rho,\gamma)\underline{X} & z^{\prime}V(\rho,\gamma)\mu^{*}
\end{array}\right).\label{eq:Kn1}
\end{align}
Thus by the determinant for partitioned matrices, and using the shorthand
notation $V=V(\rho,\gamma)$,
\begin{align}
 & \left|\det\left(E\left[K_{n}(\rho,\gamma)\left|\mathcal{Z}_{n}\right.\right]\right)\right|\label{eq:Kn2}\\
= & \left|\det\left(\frac{1}{n}E\left[\underline{X}^{\prime}V\underline{X}\left|\mathcal{Z}_{n}\right.\right]\right)\right|\left|\det\left\{ E\left[\left(\frac{1}{n}z^{\prime}V\mu^{*}-\frac{1}{n}z^{\prime}V\underline{X}(\frac{1}{n}\underline{X}^{\prime}V\underline{X})^{-1}\frac{1}{n}\underline{X}^{\prime}V\mu^{*}\right)\left|\mathcal{Z}_{n}\right.\right]\right\} \right|\nonumber \\
\geqslant & \left(\xi_{X}c_{e}\right)^{p_{x}}\left|\frac{1}{n}E\left[(V^{1/2}z)^{\prime}Q_{V^{1/2}X}(V^{1/2}\mu^{*})\left|\mathcal{Z}_{n}\right.\right]\right|.\nonumber 
\end{align}
Observe furthermore that $E\left[(V^{1/2}z)^{\prime}Q_{V^{1/2}X}(V^{1/2}y_{2})\left|\mathcal{Z}_{n}\right.\right]=E\left[(V^{1/2}z)^{\prime}Q_{V^{1/2}X}(V^{1/2}\mu^{*})\left|\mathcal{Z}_{n}\right.\right]$
in light of \eqref{eq:Graham-v1-compact} and since $Q_{V^{1/2}X}$
is orthogonal to $V^{1/2}\underline{X}$ and $E\left[u_{i2}|\mathcal{Z}_{n}\right]=0$
by Theorem \ref{thm:decomp}. Thus 
\begin{align}
 & \inf_{\gamma\in\Gamma,\rho\in[-K_{\rho},K_{\rho}]}\left|\det\left(E\left[K_{n}(\rho,\gamma)\left|\mathcal{Z}_{n}\right.\right]\right)\right|\geq\label{eq:K_pos}\\
 & \left(\xi_{X}c_{e}\right)\left(\inf_{\gamma\in\Gamma,\rho\in[-K_{\rho},K_{\rho}]}n^{-1}\left|E\left[(V^{1/2}z)^{\prime}Q_{V^{1/2}X}(V^{1/2}y_{2})\left|\mathcal{Z}_{n}\right.\right]\right|\right)\nonumber \\
 & \geq\left(\xi_{X}c_{e}\right)^{p_{x}}K_{y}>0\nonumber 
\end{align}
in light of Assumption \ref{assu:rank-new}. This proves that $E\left[K_{n}(\rho,\gamma)\left|\mathcal{Z}_{n}\right.\right]$
in \eqref{eq:mom-1} has full rank for all admissible values of $\rho$
and $\gamma$, and consequently $E\left[m_{n}^{(l)}(\theta,\gamma)\right]=0$
if and only if $f_{1}=f_{1,0}$ and $\delta=\delta_{0}$. Therefore
$f_{1,0}$ and $\delta_{0}$ are identified from the linear moment
condition.

Once $f_{1,0}$ and $\delta_{0}$ are identified, $\rho$ can be identified
from the quadratic moment condition. To see this, note that when evaluated
at $\delta_{0}$ and $f_{1,0},$ $u^{+}(\theta,\gamma)$ in Equation
\eqref{eq:uplus0-1} becomes 
\begin{align*}
u^{+}(\theta,\gamma)|_{f_{1,0},\delta_{0}} & =\Omega(\gamma)^{-1/2}(I+\rho M)^{-1}(I+\rho_{0}M)(u_{1}-f_{1,0}u_{2})\\
 & =\Omega(\gamma)^{-1/2}\mathcal{M}(\rho)(u_{1}-f_{1,0}u_{2}).
\end{align*}
Then, 
\begin{align*}
 & \left.E\left[u^{+}(\theta,\gamma)^{\prime}Au^{+}(\theta,\gamma)|\mathcal{Z}_{n}\right]\right\vert \\
= & E\left[(u_{1}-f_{1,0}u_{2})^{\prime}\mathcal{M}\left(\rho\right)\Omega(\gamma)^{-1/2}A\Omega(\gamma)^{-1/2}\mathcal{M}\left(\rho\right)(u_{1}-f_{1,0}u_{2})|\mathcal{Z}_{n}\right]\\
= & E\left[\tr\left(\mathcal{M}\left(\rho\right)\Omega(\gamma)^{-1/2}A\Omega(\gamma)^{-1/2}\mathcal{M}\left(\rho\right)\Omega_{0}\right)|\mathcal{Z}_{n}\right].
\end{align*}
By Lemma \ref{lem:id_q} in Section \ref{sec:Auxiliary-Lemmas} in
the online appendix, the equation above equals 0 if and only if $\rho=\rho_{0}$
for $-1<\rho<1$. This establishes that for all admissible $\gamma$
we have $E\left[m_{n}\left(\theta_{0},\gamma\right)|\mathcal{Z}_{n}\right]=0\text{ a.s.}$
and $E\left[m_{n}\left(\theta,\gamma\right)|\mathcal{Z}_{n}\right]\neq0\text{ a.s.}$
for $\theta\neq\theta_{0}$.
\end{proof}
\pagebreak{}

\section*{Online Appendix for ``Differential Test Performance and Peer Effects''
by Guido M. Kuersteiner, Ingmar R. Prucha and Ying Zeng}

\setcounter{page}{1}

\section{\label{sec:Error_Term}Assumptions and Results for $u_{it}$}

We impose the following formal assumption about the variance of the
error structure.
\begin{assumption}
\label{assu:CondVar} Let $u_{it}=y_{it}^{*}-\kappa_{i}f_{t}$ and
let $i=n_{1}+\ldots+n_{c-1}+r$ with $1\leq r\leq n_{c}$, i.e., $i$
refers to student $r$ in class c. Let $u_{-i,t}=\left(u_{1,t},...,u_{i-1,t},u_{i+1,t},...,u_{n,t}\right)$.
Define the filtration $\mathcal{B}_{n,i,t}=\sigma\left(S,\alpha,\varphi,X,z,\zeta^{*},u_{-t},u_{-i,t}\right)\text{\,for \ensuremath{t=1,2}}$
where $u_{-t}$ represents information about the other test. Then
the conditional variance of $u_{it}$ is given by $E\left[u_{it}^{2}|\mathcal{B}_{n,i,t}\right]=\sigma_{t}^{2}\rho_{\tau_{c}}^{2},$where
$\tau_{c}$ is $\mathcal{Z}_{n}$-measurable and $\tau_{c}\in\{1,\ldots,J\}$,
and $c_{u}\leq\sigma_{t}^{2},\rho_{j}^{2}\leq C_{u}$ for some finite
positive bounds.
\end{assumption}
Assumption \ref{assu:CondVar} implies that for $t=1,2$, $\Sigma_{t}\equiv E\left[u_{t}u_{t}^{\prime}|\mathcal{Z}_{n}^{\ast}\right]=\sigma_{t}^{2}\diag_{c=1}^{C}\left(\rho_{\tau_{c}}^{2}I_{c}\right)$.
Observing that $f_{t}$ is a fixed constant, conditional on $S,\zeta^{*},\alpha,X$
and $z$, knowledge of $y_{it}^{*}$ is equivalent to knowledge of
$u_{it}$. Thus we have the following equivalence of information sets\footnote{See the proof of Lemma \ref{lem:u_Msbl_MomB} for more detail.}
\[
\mathcal{F}{}_{n,i,t}=\mathcal{B}_{n,i,t}\,\text{for}\,t=1,2.
\]
Assumption \ref{assu:CondVar} imposes a similar factor structure
on conditional variances as does Assumption \ref{assu:Proportional}
for conditional means. By Theorem \ref{thm:decomp} the errors $u_{it}$
are mutually uncorrelated over tests as well as cross-sectionally.
Consequently the variance covariance matrix of $u_{1}-f_{1}u_{2}$
is given by, 
\[
\Omega=\Var\left(u_{1}-f_{1}u_{2}|\mathcal{Z}_{n}^{*}\right)=\Sigma_{1}+f_{1}^{2}\Sigma_{2}=\left(\sigma_{1}^{2}+f_{1}^{2}\sigma_{2}^{2}\right)\diag_{c=1}^{C}\left(\rho_{\tau_{c}}^{2}I_{c}\right).
\]
Define $\gamma_{j}^{2}=\left(\sigma_{1}^{2}+f_{1}^{2}\sigma_{2}^{2}\right)\rho_{j}^{2}$,
for $j=1,...,J$ and $\gamma=\left(\gamma_{1},...,\gamma_{J}\right)'$.
Then $\Omega(\gamma)=\diag_{c=1}^{C}\left(\gamma_{\tau_{c}}^{2}I_{c}\right)$
and the conditions about $\Omega(\gamma)$ and $\gamma$ in Assumption
\ref{assu:PAR} hold.
\begin{lem}
\label{lem:u_Msbl_MomB}Let Assumption \ref{assu:Prop_Gen} hold.
Then, $u_{it}$ defined in Theorem \ref{thm:decomp} is a martingale
difference sequence w.r.t to the filtration $\mathcal{A}_{n,\nu}$
with \textup{$\mathcal{A}_{n,0}=\mathcal{Z}_{n}^{*}=\sigma\left(S,\alpha,\varphi,X,z,\zeta^{*}\right)$
and }
\begin{equation}
\mathcal{A}_{n,\left(t-1\right)n+i-1}=\sigma\left(S,\alpha,\varphi,X,z,\zeta^{*},\left\{ u_{j,t-1}\right\} _{j=1}^{n},\left\{ u_{jt}\right\} _{j=1}^{i-1}\right)\text{\,for \ensuremath{t=1,2;i=1,\ldots,n}}\label{eq:A_n,(t-1)n+i-1}
\end{equation}
and where $u_{j,0}=0$ to handle the case $t=1.$ In addition, $\sup_{i,t}E\left[\left|u_{it}\right|^{4+\eta}\right]<\infty$
for some $\eta>0$.
\end{lem}
\begin{proof}
First note that the filtrations $\mathcal{A}_{n,\nu}$ are increasing
in the sense $\mathcal{A}_{n,0}\subseteq\mathcal{A}_{n,1}....\subseteq\mathcal{A}_{n,\nu}\subseteq\mathcal{A}_{n,\nu+1}....$.
By construction $u_{it}=y_{it}^{*}-\kappa_{i}f_{t}$ where $\kappa_{i}f_{t}$
is $\mathcal{A}_{n,0}$-measurable. For general $i$ and $t$ note
that \textcolor{black}{since $u_{it}=y_{it}^{*}-\kappa_{i}f_{t}$
and $\kappa_{i}f_{t}$ is $\mathcal{Z}_{n}^{*}$-measurable it follows
that $\mathcal{B}_{n,i,t}\subseteq\mathcal{F}_{n,i,t}$. Since $y_{it}^{*}=\kappa_{i}f_{t}+u_{it}$
and $\kappa_{i}f_{t}$ is $\mathcal{Z}_{n}^{*}$-measurable it follows
that $\mathcal{F}_{n,i,t}\subseteq\mathcal{B}_{n,i,t}$. Consequently
$\mathcal{F}_{n,i,t}=\mathcal{B}_{n,i,t}.$ Observing that $\mathcal{A}_{n,\left(t-1\right)n+i-1}\subseteq\mathcal{B}_{n,i,t}$
it follows, using iterated expectations, that 
\begin{eqnarray*}
E\left[u_{it}|\mathcal{A}_{n,\left(t-1\right)n+i-1}\right] & = & E\left[E\left[u_{it}|\mathcal{B}_{n,i,t}\right]|\mathcal{A}_{n,\left(t-1\right)n+i-1}\right]=0
\end{eqnarray*}
 observing that $E\left[u_{it}|\mathcal{B}_{n,i,t}\right]=E\left[y_{it}^{*}-\kappa_{i}f_{t}|\mathcal{F}_{n,i,t}\right]=0$
in light of Assumption \ref{assu:Proportional}. Finally observe that
$u_{it}$ is measurable w.r.t. $\mathcal{A}_{n,\left(t-1\right)n+i}$
.}

Recall that by Assumption \ref{assu:Proportional} there exists a
random variable $y$ such that $\left|y_{it}^{*}\right|+\left|\kappa_{i}f_{t}\right|\leq y$
with $E\left[\left|y\right|^{4+\eta}|\mathcal{A}_{n,0}\right]\leq K_{y}<\infty.$
Hence\textcolor{blue}{{} }
\begin{align}
E\left[\left|u_{it}\right|^{4+\eta}|\mathcal{A}_{n,0}\right] & \leq E\left[\left(\left|y_{it}^{*}\right|+\left|\kappa_{i}f_{t}\right|\right)^{4+\eta}|\mathcal{A}_{n,0}\right]\leq E\left[\left|y\right|^{4+\eta}|\mathcal{A}_{n,0}\right]<K_{y}<\infty.\label{eq:u4bound}
\end{align}
 Since the bound on the RHS is uniform in $i$ and $t$ by assumption
it follows that $\sup_{i,t}E\left[\left|u_{it}\right|^{4+\eta}\right]<\infty.$
\end{proof}

\section{\label{sec:Consistency}Consistency of Parameter Estimators}

Formulating the moment conditions using transformed residuals as in
\eqref{eq:Moments_SI_Graham} requires consistent estimates of the
variance parameters $\gamma$. These can be obtained by first obtaining
consistent but inefficient parameter estimates. To do so, we set all
elements of $\gamma=(\gamma_{1},...,\gamma_{J})'$ to one. Theorem
\ref{thm:general_id} implies that the parameters $\theta=(\rho,f_{1},\delta^{\prime})^{\prime}$
are still identified by the moment conditions in \eqref{eq:Moments_SI_Graham}.

First step estimators $(\tilde{\rho},\tilde{f}_{1},\tilde{\delta}^{\prime})^{\prime}$
are defined in \eqref{eq:f_tilde}, \eqref{eq:delta_tilde} and \eqref{eq:rho_tilde}.\footnote{\citet{lee_gmm_2007} considers first step and efficient GMM estimators
using linear and quadratic moments.} To formalize the identification result for these non-efficient first
step estimators define the moment vector

\begin{equation}
m_{\epsilon}\left(\theta\right)=n^{-1/2}\left(\begin{array}{c}
H^{\prime}(I+\rho M)\epsilon^{+}\left(\theta\right)\\
\epsilon^{+}\left(\theta\right)^{\prime}A\epsilon^{+}\left(\theta\right)
\end{array}\right).\label{eq:Moment_simple}
\end{equation}
The next lemma shows that the moment conditions in \eqref{eq:Moment_simple}
are sufficient to identify the parameters $\theta=(\rho,f_{1},\delta^{\prime})^{\prime}$.
\begin{lem}
\label{lem:linear_id-1}Let Assumptions \ref{assu:Potential_Outcomes-cov}-\ref{assu:size},
and Assumption \ref{assu:rank-new}(i),(ii) be satisfied. Let $m_{\epsilon}(\theta)$
be defined as in \eqref{eq:Moment_simple}. Assume further that $n^{-1}\left|E\left[z^{\prime}Q_{X}y_{2}|\mathcal{Z}_{n}\right]\right|>K_{y}>0$
a.s . Then $E\left[m_{\epsilon}\left(\theta_{0}\right)\left|\mathcal{Z}_{n}\right.\right]=0$
and $E\left[m_{\epsilon}\left(\theta\right)\left|\mathcal{Z}_{n}\right.\right]\neq0$
a.s. if $\theta\neq\theta_{0}$ for all $\theta\in\Theta$.
\end{lem}
\begin{proof}[Proof of Lemma \ref{lem:linear_id-1}]
Observe that utilizing \eqref{eq:ydif} we have 
\begin{eqnarray*}
n^{-1}E\left[H'(I+\rho M)\epsilon^{+}\left(\rho,f_{1},\delta\right)\left|\mathcal{Z}_{n}\right.\right] & = & n^{-1}E\left[H'\left(y_{1}-f_{1}y_{2}-\underline{X}\delta\right)\left|\mathcal{Z}_{n}\right.\right]\\
 & = & n^{-1}E\left[H'\left(\mu^{\ast}(f_{1,0}-f_{1})+\underline{X}(\delta(f_{1},\rho_{0},\beta_{0})-\delta)\right)\left|\mathcal{Z}_{n}\right.\right].
\end{eqnarray*}
Let $\tilde{K}_{n}=n^{-1}H'\left(\underline{X},\mu^{*}\right)$, then
the linear moment condition can be expressed as
\begin{equation}
n^{-1}E\left[H'(I+\rho M)\epsilon^{+}\left(\rho,f_{1},\delta\right)\left|\mathcal{Z}_{n}\right.\right]=E\left[\tilde{K}_{n}\left|\mathcal{Z}_{n}\right.\right]\left(\begin{array}{c}
\delta(f_{1},\rho_{0},\beta_{0})-\delta\\
f_{1,0}-f_{1}
\end{array}\right)=0.\label{eq:K_CondExp3}
\end{equation}
Observe that $E\left[\tilde{K}_{n}\right]$ has full rank, as is readily
seen from the proof of Theorem \ref{thm:general_id} in light of (\ref{eq:K_pos})
with $V=I$. This implies that the only solution to \eqref{eq:K_CondExp3}
is $\delta=\delta_{0}$ and $f_{1}=f_{1,0}$. With $\mathcal{M}(\rho)=(I+\rho M)^{-1}(I+\rho_{0}M)$,
it now follows that 
\[
\epsilon^{+}\left(\rho,f_{1,0},\delta_{0}\right)=\mathcal{M}\left(\rho\right)(u_{1}-f_{1,0}u_{2})
\]
 such that the quadratic moment condition is 
\begin{equation}
E\left[\epsilon^{+}\left(\rho,f_{1,0},\delta_{0}\right)'A\epsilon^{+}\left(\rho,f_{1,0},\delta_{0}\right)\left|\mathcal{Z}_{n}\right.\right]=\tr\left(\mathcal{M}\left(\rho\right)A\mathcal{M}\left(\rho\right)\Omega_{0}\right)=0.\label{eq:MomQ_Simple}
\end{equation}
By Lemma \ref{lem:id_q} it follows that the only solution to \eqref{eq:MomQ_Simple}
is $\rho=\rho_{0}.$
\end{proof}
We next show that the initial estimators and the efficient estimators
are consistent. To establish consistency we impose assumptions related
to the convergence of sample averages and moments of sample averages
to well defined limits. These assumptions are adaptations of Assumptions
2, 3 and 5 in \citet{kuersteiner_dynamic_2020} to the present setting
which differs somewhat from the framework of sequential exogeneity
in that paper.

Let $i=n_{1}+\ldots n_{c-1}+r$ with $1\leq r\leq n_{c}$. Then, with
some abuse of notation, denote the $i$-th diagonal element of $\Omega(\gamma)=\diag_{c=1}^{C}\left(\gamma_{\tau_{c}}^{2}I_{c}\right)$
as $\gamma_{(i)}^{2}=\gamma_{\tau_{c}}^{2}$, noting that $\tau_{c}$
varies with $i$. Let the linear and quadratic moment vectors $m_{n}^{(l)}(\theta,\gamma)$
and $m_{n}^{(q)}(\theta,\gamma)$ be as defined in \eqref{eq:Moments_SI_Graham}.
Let $h_{i}=\left[h_{i1},...,h_{iq}\right]$ denote the $i$-th row
vector of $H$, and let $a_{ij}$ denote the the $(i,j)$-th element
of matrix $A$. Then note that with $u^{+}(\theta_{0},\gamma)=\Omega(\gamma)^{-1/2}(u_{1}-f_{1,0}u_{2})$,
the variance of $m_{n}^{(l)}(\theta_{0},\gamma)$ is 
\begin{equation}
\frac{1}{n}E\left[H^{\prime}\Omega(\gamma)^{-1}\Omega(\gamma_{0})H\right]=\frac{1}{n}\sum_{i=1}^{n}E\left[\frac{\gamma_{(i),0}^{2}}{\gamma_{(i)}^{2}}h_{i}^{\prime}h_{i}\right]
\end{equation}
and the variance of $m_{n}^{(q)}(\theta_{0},\gamma)$ is
\begin{equation}
\frac{2}{n}E\left[\tr\left(\Omega_{0}\Omega^{-1}A\Omega_{0}\Omega^{-1}A\right)\right]=\frac{2}{n}E\sum_{i=1}^{n}\sum_{j=1}^{n}\left[\frac{\gamma_{(i),0}^{2}}{\gamma_{(i)}^{2}}\frac{\gamma_{(j),0}^{2}}{\gamma_{(j)}^{2}}a_{ij}^{2}\right].
\end{equation}
The assumption below ensures that for any admissible $\gamma$, $\Var\left(m_{n}^{(l)}(\theta_{0},\gamma)\right)$
and $\Var\left(m_{n}^{(q)}(\theta_{0},\gamma)\right)$ and their corresponding
sample analogues converge, respectively, to finite positive definite
matrices.

\begin{assumption}
\label{assu:conv_Var_new} The following holds :
\[
n^{-1}\sum_{i=1}^{n}E\left[\frac{\gamma_{(i),0}^{2}}{\gamma_{(i)}^{2}}h_{i}^{\prime}h_{i}\right]\rightarrow V_{\gamma}^{h},\quad n^{-1}\sum_{i=1}^{n}\sum_{j=1}^{n}E\left[\frac{\gamma_{(i),0}^{2}}{\gamma_{(i)}^{2}}\frac{\gamma_{(j),0}^{2}}{\gamma_{(j)}^{2}}a_{ij}^{2}\right]\rightarrow V_{\gamma}^{a},
\]
where the elements of $V_{\gamma}^{h}$ and $V_{\gamma}^{a}$ are
finite , and 
\[
V_{n,\gamma}^{h}=n^{-1}\sum_{i=1}^{n}\frac{\gamma_{(i),0}^{2}}{\gamma_{(i)}^{2}}h_{i}^{\prime}h_{i}\overset{p}{\rightarrow}V_{\gamma}^{h},\quad V_{n,\gamma}^{a}=n^{-1}\sum_{i=1}^{n}\sum_{j=1}^{n}\frac{\gamma_{(i),0}^{2}}{\gamma_{(i)}^{2}}\frac{\gamma_{(j),0}^{2}}{\gamma_{(j)}^{2}}a_{ij}^{2}\overset{p}{\rightarrow}V_{\gamma}^{a}.
\]
The matrix $V_{\gamma}=\diag\left(V_{\gamma}^{h},2V_{\gamma}^{a}\right)$
is positive definite.
\end{assumption}
\begin{assumption}
\label{assu:conv_Mom} (i) Let $C\left(\theta,\gamma\right)$ be an
$n\times n$ matrix of the form $D$, $DP$, $DAD$, $DADP$, or $PDADP$,
where $D$ is an $n\times n$ positive diagonal matrix with elements
which are uniformly bounded and measurable w.r.t. $\mathcal{Z}_{n}$
and where $P=\diag_{c}\left(p_{c}I_{c}^{\ast}+q_{c}J_{c}^{\ast}\right)$
and where $p_{c}$ and $q_{c}$ are continuously differentiable functions
of the parameters $\theta$ and $\gamma$ as well as of variables
generating $\mathcal{Z}_{n}$, and $I_{c}^{\ast}$ and $J_{c}^{\ast}$
are defined in Appendix \ref{sec:Auxiliary-Lemmas}. Then
\[
\lim_{n\rightarrow\infty}\sup_{\theta,\gamma}\left|n^{-1}\Upsilon_{a}'C\left(\theta,\gamma\right)\Upsilon_{b}-n^{-1}E\left[\Upsilon_{a}'C\left(\theta,\gamma\right)\Upsilon_{b}\right]\right|=0\text{ a.s. }
\]
where $\Upsilon_{a},\Upsilon_{b}$ are selected from the set \textup{$\left\{ H,MH,\mu^{*},u_{1},u_{2}\right\} $.}

(ii) $\lim_{n\rightarrow\infty}\sup_{\theta,\gamma}\left\Vert E\left[n^{-1}\Upsilon_{a}'C\left(\theta,\gamma\right)\Upsilon_{b}\right]-\mathscr{U}_{a,b}\left(\theta,\gamma\right)\right\Vert =0$
where
\[
\mathscr{U}_{a,b}\left(\theta,\gamma\right)=\lim_{n\rightarrow\infty}E\left[n^{-1}\Upsilon_{a}'C\left(\theta,\gamma\right)\Upsilon_{b}\right]
\]
 is bounded and continuously differentiable in $\theta$ and $\gamma$.
\end{assumption}
In \citet{kuersteiner_efficient_2023} we show uniform convergence
under more primitive conditions in a closely related situation while
directly assuming it here in Assumption \ref{assu:conv_Mom} to save
space. Without further assumptions on cross-sectional dependence,
high level assumptions about the convergence of sample averages such
as $n^{-1}\underline{X}^{\prime}\underline{X}$ are required irrespective.
The next theorem establishes consistency of the GMM estimator.
\begin{thm}
\label{thm:Consistency}Let Assumptions \ref{assu:Potential_Outcomes-cov}-
\ref{assu:conv_Mom} be satisfied. Assume that $\bar{\gamma}_{n}\stackrel{p}{\rightarrow}\gamma$
for some sequence $\bar{\gamma}_{n}$ and some $\gamma\in\Gamma.$
Then,

(i) for $\hat{\theta}\left(\gamma\right)=\argmin_{\theta}Q_{n}\left(\theta,\gamma\right)$
with $Q_{n}\left(\theta,\gamma\right)$ defined in (\ref{eq:Def_Qn})
it follows that \textup{$\hat{\theta}\left(\bar{\gamma}_{n}\right)\stackrel{p}{\rightarrow}\theta_{0}$.}

(ii) for $\tilde{f}_{1}$ and $\tilde{\delta}$ defined in \eqref{eq:f_tilde}
and \eqref{eq:delta_tilde} and $\tilde{\rho}$ defined in \eqref{eq:rho_tilde}
it follows that $\tilde{f}_{1}\stackrel{p}{\rightarrow}f_{1,0}$,
$\tilde{\delta}\stackrel{p}{\rightarrow}\delta_{0}$ and $\tilde{\rho}\stackrel{p}{\rightarrow}\rho_{0}.$
\end{thm}
\begin{proof}
To prove part (i) of the theorem we establish the conditions of Lemma
3.1 in \citet{potscher_dynamic_1997}. For the linear moment conditions
use \eqref{eq:ml}. We have
\begin{eqnarray*}
n^{-1}H^{\prime}u^{+}(\theta,\gamma) & = & \left(f_{1,0}-f_{1}\right)n^{-1}H^{\prime}V(\rho,\gamma)\mu^{\ast}\\
 & + & \frac{1}{n}H^{\prime}V(\rho,\gamma)\underline{X}\left(\delta(f_{1},\rho_{0},\beta_{0})-\delta\right)\\
 & + & \frac{1}{n}H^{\prime}\Omega(\gamma)^{-1/2}\mathcal{M}(\rho)(u_{1}-f_{1}u_{2})\\
 & \equiv & I+II+III.
\end{eqnarray*}

Consider $I$ and $II$. Using Assumption \ref{assu:conv_Mom} and
noting that $V(\rho,\gamma)$ defined in \eqref{eq:V} satisfies the
assumption on $C(\theta,\gamma)$, hence $n^{-1}H^{\prime}V(\rho,\gamma)\underline{X}$
converges uniformly to 
\begin{equation}
\lim_{n\rightarrow\infty}E\left[n^{-1}H^{\prime}V(\rho,\gamma)\underline{X}\right]\equiv\mathscr{U}_{H,x}\left(\theta,\gamma\right)\label{eq:U_HX-1}
\end{equation}
and $n^{-1}H^{\prime}V(\rho,\gamma)\mu^{\ast}$ converges uniformly
to 
\begin{equation}
\lim_{n\rightarrow\infty}E\left[n^{-1}H^{\prime}V\left(\rho,\gamma\right)\mu^{*}\right]\equiv\mathscr{U}_{H,\mu}\left(\theta,\gamma\right).\label{eq:U_Hm-1}
\end{equation}

Consider $III$. From Corollary \ref{cor:pIqJ} and Assumption \ref{assu:conv_Mom},
we have $sup_{\theta,\gamma}|III|\rightarrow_{p}0$.

Now consider the quadratic moment conditions. Let $\vartheta=(\delta^{\prime},f_{1})^{\prime}$,
$\vartheta_{0}=(\delta^{\prime}(f_{1},\rho_{0},\beta_{0}),f_{1,0})^{\prime}$
and $\Delta_{1}=[\underline{X},\mu^{\ast}]$. We point out that $\vartheta_{0}$
is a function of $f_{1}$. We omit this dependence for notational
convenience and because $f_{1}$ is pinned down by linear moment conditions
for $f_{1}$ directly, or in other words by equations related to the
second component of $\vartheta$. In light of (\ref{eq:uplus0-1})
we have
\begin{eqnarray}
n^{-1}u^{+}(\theta,\gamma)Au^{+}(\theta,\gamma) & = & n^{-1}\left(\vartheta-\vartheta_{0}\right)'\Delta_{1}'V(\rho,\gamma)'AV(\rho,\gamma)\Delta_{1}\left(\vartheta-\vartheta_{0}\right)\nonumber \\
 &  & -n^{-1}2\left(\vartheta-\vartheta_{0}\right)'\Delta_{1}'V(\rho,\gamma)'A\Omega(\gamma)^{-1/2}\mathcal{M}(\rho)\left(u_{1}-f_{1}u_{2}\right)\nonumber \\
 &  & +n^{-1}\left(u_{1}-f_{1}u_{2}\right)^{\prime}\mathcal{M}(\rho)^{\prime}\Omega(\gamma)^{-1/2}A\Omega(\gamma)^{-1/2}\mathcal{M}(\rho)\left(u_{1}-f_{1}u_{2}\right)\nonumber \\
 & = & I+II+III.\label{eq:UAU}
\end{eqnarray}
By Assumption \ref{assu:conv_Mom}, $n^{-1}\Delta_{1}'V(\rho,\gamma)'AV(\rho,\gamma)\Delta_{1}$
in $I$ converges uniformly to 
\begin{equation}
\lim_{n\rightarrow\infty}n^{-1}E\left[\Delta_{1}'V(\rho,\gamma)'AV(\rho,\gamma)\Delta_{1}\right]\equiv\mathscr{U}_{A,\Upsilon}\left(\theta,\gamma\right)\label{eq:U_AU}
\end{equation}
By Corollary \ref{cor:pIqJ} and Assumption \ref{assu:conv_Mom},
$II$ converges uniformly to 0. Consider $III$, let $B(\theta,\gamma)=\mathcal{M}(\rho)^{\prime}\Omega(\gamma)^{-1/2}A\Omega(\gamma)^{-1/2}\mathcal{M}(\rho)$,
then 
\begin{align*}
\frac{1}{n}(u_{1}-f_{1}u_{2})^{\prime}B(\theta,\gamma)(u_{1}-f_{1}u_{2}) & =\frac{1}{n}u_{1}^{\prime}B(\theta,\gamma)u_{1}+\frac{f_{1}^{2}}{n}u_{2}^{\prime}B(\theta,\gamma)u_{2}-\frac{2}{n}f_{1}u_{1}^{\prime}B(\theta,\gamma)u_{2}.
\end{align*}
Each term on the R.H.S. converges to its mean uniformly in $\theta$
and $\gamma$ under Corollary \ref{cor:pIqJ} and Assumption \ref{assu:conv_Mom},
and the mean is finite and continuously differentiable in $\theta,\gamma$
. Consequently, $\frac{1}{n}(u_{1}-f_{1}u_{2})^{\prime}B(\theta,\gamma)(u_{1}-f_{1}u_{2})$
converges uniformly to the limit of its mean 
\begin{align*}
E\left[\frac{1}{n}(u_{1}-f_{1}u_{2})^{\prime}B(\theta,\gamma)(u_{1}-f_{1}u_{2})\right] & =\tr\left[\frac{1}{n}B(\theta,\gamma)\left(\Sigma_{1}+f_{1}^{2}\Sigma_{2}\right)\right]\\
 & =\frac{1}{n}\tr\left(B(\theta,\gamma)\Omega_{0}\right)+\frac{1}{n}\left(f_{1}^{2}-f_{1,0}^{2}\right)\tr\left(B(\theta,\gamma)\Sigma_{2}\right),
\end{align*}
observing that $\Omega_{0}=\Sigma_{1}+f_{1,0}^{2}\Sigma_{2}$. Next
observe that $B(\theta,\gamma)=\mathcal{M}(\rho)^{\prime}\Omega(\gamma)^{-1/2}A\Omega(\gamma)^{-1/2}\mathcal{M}(\rho)$
falls into the class of $C(\theta,\gamma)$ matrices defined in Assumption
\ref{assu:conv_Mom}(i). Thus by Assumption \ref{assu:conv_Mom},
there exist bounded and continuously differentiable functions in $\theta$
and $\gamma$ denoted by $\mathscr{U}_{A,\Omega}(\theta,\gamma)$
and $\mathscr{U}_{A,\Sigma}(\theta,\gamma)$ such that
\begin{equation}
\lim_{n\rightarrow\infty}\sup_{\theta,\gamma}\left\Vert E\left[\frac{1}{n}\tr\left(B(\theta,\gamma)\Omega_{0}\right)\right]-\mathscr{U}_{A,\Omega}(\theta,\gamma)\right\Vert =0\label{eq:Uao}
\end{equation}
 and 
\begin{equation}
\lim_{n\rightarrow\infty}\sup_{\theta,\gamma}\left\Vert E\left[\frac{1}{n}\tr\left(B(\theta,\gamma)\Sigma_{2}\right)\right]-\mathscr{U}_{A,\Sigma}(\theta,\gamma)\right\Vert =0\label{eq:uas}
\end{equation}

From uniform convergence and \eqref{eq:U_HX-1}, \eqref{eq:U_AU},
\eqref{eq:Uao} and \eqref{eq:uas} it follows that 
\begin{align}
\lim_{n\rightarrow\infty}n^{-1/2}E\left[m_{n}\left(\theta,\gamma\right)\right] & =\left[\begin{array}{c}
\left(\mathscr{U}_{H,x}\left(\theta,\gamma\right),\mathscr{U}_{H,\mu}\left(\theta,\gamma\right)\right)\left(\begin{array}{c}
\delta(f_{1},\rho_{0},\beta_{0})-\delta\\
f_{1,0}-f_{1}
\end{array}\right)\\
\left(\vartheta-\vartheta_{0}\right)'\mathscr{U}_{A,\Upsilon}\left(\theta,\gamma\right)\left(\vartheta-\vartheta_{0}\right)+\mathscr{U}_{A,\Omega}\left(\theta,\gamma\right)+\left(f_{1}^{2}-f_{1,0}^{2}\right)\mathscr{U}_{A,\Sigma}(\theta,\gamma)
\end{array}\right]\label{eq:Mom_lim}\\
 & \equiv\left[\begin{array}{c}
\mathfrak{m}_{l}\left(\theta,\gamma\right)\\
\mathfrak{m}_{q}\left(\theta,\gamma\right)
\end{array}\right]\equiv\mathfrak{m}\left(\theta,\gamma\right).\nonumber 
\end{align}
By uniform convergence of $n^{-1/2}m_{n}\left(\theta,\gamma\right)$
it also follows that for $\bar{\gamma}_{n}\rightarrow\gamma_{*}$
\begin{equation}
\sup_{\theta}\left\Vert n^{-1/2}m_{n}\left(\theta,\bar{\gamma}_{n}\right)-\mathfrak{m}\left(\theta,\gamma_{*}\right)\right\Vert \rightarrow0\,i.p.\label{eq:unifcon}
\end{equation}
We proceed to show, building on our results for finite $n$, that
$\theta_{0}$ is also the unique solution vector of the limiting moment
condition $\mathfrak{m}\left(\theta,\gamma_{*}\right)=0$. Observe
that from (\ref{eq:Kn}), (\ref{eq:U_HX-1}) and (\ref{eq:U_Hm-1})
we have $\left[\mathscr{U}_{H,x}\left(\theta,\gamma_{*}\right),\mathscr{U}_{H,\mu}\left(\theta,\gamma_{*}\right)\right]=\lim_{n\rightarrow\infty}E\left[K_{n}(\rho,\gamma_{*})\right],$
and by (\ref{eq:K_pos}) we have $\inf_{\gamma\in\Gamma,\rho}\left|\det\left(E\left[K_{n}(\rho,\gamma)\right]\right)\right|\geq C_{K}>0$
for some $C_{K}>0$ and all $n.$ Since $\det(.)$ is a continuous
function it follows that for any $\theta,\gamma_{*}$ we have 
\begin{align*}
\left|\det\left(\mathscr{U}_{H,x}\left(\theta,\gamma_{*}\right),\mathscr{U}_{H,\mu}\left(\theta,\gamma_{*}\right)\right)\right| & =\left|\det\left(\lim_{n\rightarrow\infty}E\left[K_{n}(\rho,\gamma_{*})\right]\right)\right|\\
 & =\lim_{n\rightarrow\infty}\left|\det\left(E\left[K_{n}(\rho,\gamma_{*})\right]\right)\right|\geq C_{K}>0.
\end{align*}
 In light of this for any $\rho$ and $\gamma_{*}$
\[
\mathfrak{m}_{l}\left(\theta,\gamma_{*}\right)=\mathfrak{m}_{l}\left(\rho,f_{1},\delta,\gamma_{*}\right)=0
\]
 if and only if $f_{1}=f_{1,0}$ and $\delta=\delta_{0}$. Note that
$\delta(f_{1},\rho_{0},\beta_{0})$ equals $\delta_{0}$ when $f_{1}=f_{1,0}$.
In all , $\mathfrak{m}_{l}\left(\theta,\gamma_{*}\right)=0$ if and
only if $f_{1}=f_{1,0}$ and $\delta=\delta_{0}$.

To show that $\mathfrak{m}\left(\theta,\gamma_{*}\right)=0$ if and
only if $\theta=(\rho,f_{1},\delta')'=\theta_{0}=(\rho_{0},f_{1,0},\delta_{0}^{\prime})^{\prime}$
it thus suffices to show that $\mathfrak{m}_{q}\left(\rho,f_{1,0},\delta_{0},\gamma_{*}\right)=0$
if and only if $\rho=\rho_{0}$. Recall that 
\begin{align*}
\mathfrak{m}_{q}\left(\rho,f_{1,0},\delta_{0},\gamma_{*}\right) & =\mathscr{U}_{A,\Omega}\left(\rho,f_{1,0},\delta_{0},\gamma_{*}\right)\\
 & =\lim_{n\rightarrow\infty}\frac{1}{n}\tr\left(\mathcal{M}\left(\rho\right)\Omega(\gamma_{\ast})^{-1/2}A\Omega(\gamma_{\ast})^{-1/2}\mathcal{M}\left(\rho\right)\Omega_{0}\right).
\end{align*}
By Lemma \ref{lem:id_q}, $\mathfrak{m}_{q}\left(\rho,f_{1,0},\delta_{0},\gamma_{*}\right)=0$
if and only if $\rho=\rho_{0}$, and consequently $\mathfrak{m}\left(\theta,\gamma_{*}\right)=0$
if and only if $\theta=\theta_{0}$. Assumption \ref{assu:conv_Mom}
implies that $\mathfrak{m}\left(\theta,\gamma_{*}\right)$ and thus
$\mathfrak{m}\left(\theta,\gamma_{*}\right)^{\prime}\Xi^{-1}\mathfrak{m}\left(\theta,\gamma_{*}\right)$
is continuous. Since the parameter space is compact it follows that
$\theta_{0}=[$$\rho_{0},f_{1,0},\delta_{0}^{\prime}]^{\prime}$ is
the identifiable unique minimizer of $\mathfrak{m}\left(\theta,\gamma_{*}\right)^{\prime}\Xi^{-1}\mathfrak{m}\left(\theta,\gamma_{*}\right)$;
cp. the discussion after Definition 3.1 in \citet{potscher_dynamic_1997}.
Consistency now follows immediately from Lemma 3.1 in \citet{potscher_dynamic_1997}.

For part (ii) of the theorem we first consider $\tilde{f}_{1}$. Observing
that in light of \eqref{eq:lin_IV} we have $y_{1}=f_{1,0}y_{2}+\underline{X}\delta_{0}+(I+\rho_{0}M)(u_{1}-f_{1,0}u_{2})$
and that $Q_{X}\underline{X}=0$, it follows from \eqref{eq:f_tilde}
that 
\begin{eqnarray}
\tilde{f}_{1} & = & \frac{z'Q_{X}y_{1}}{z'Q_{X}y_{2}}=f_{1,0}+\frac{n^{-1}z'Q_{X}(I+\rho_{0}M)(u_{1}-f_{1,0}u_{2})}{n^{-1}z'Q_{X}y_{2}}.\label{eq:f_tilde_2}
\end{eqnarray}
 It follows that $n^{-1}z'Q_{X}(I+\rho_{0}M)(u_{1}-f_{1,0}u_{2})\rightarrow_{p}0$
since $n^{-1}z'Q_{X}(I+\rho_{0}M)u_{t}\rightarrow_{p}0$ for $t=1,2$
by Assumption \ref{assu:conv_Mom}.

Observe that for $\rho=0$,$\gamma=(1,...,1)^{\prime}$, $V\left(\rho,\gamma\right)=I_{n}$
and $Q_{V^{1/2}X}=Q_{X}.$ Using \eqref{eq:Graham-v1-compact} implies
that 
\[
n^{-1}z'Q_{X}y_{2}=n^{-1}z'Q_{X}\mu^{*}+n^{-1}z'Q_{X}(I+\rho_{0}M)u_{2}
\]
and $E\left[n^{-1}z'Q_{X}y_{2}|\mathcal{Z}_{n}\right]=E\left[n^{-1}z'Q_{X}\mu^{*}|\mathcal{Z}_{n}\right].$
As argued above, $n^{-1}z'Q_{X}(I+\rho_{0}M)u_{2}\rightarrow_{p}0.$
By Assumption \ref{assu:conv_Mom} it follows that $n^{-1}\left|z'Q_{X}\mu^{*}-E\left[z'Q_{X}\mu^{*}\right]\right|\rightarrow0$
a.s. and $n^{-1}\left|E\left[z'Q_{X}\mu^{*}\right]-\mathcal{\mathscr{U}}_{z\mu}\right|\rightarrow0$
a.s. Assumption \ref{assu:rank-new}(iii) implies that $n^{-1}\left|E\left[z'Q_{X}\mu^{*}\right]\right|\geq K_{y}>0$
which in turn implies that $\left|\mathcal{\mathscr{U}}_{z\mu}\right|\geq K_{y}>0$
. This shows that the denominator on the r.h.s. of \eqref{eq:f_tilde_2}
converges to a positive constant. Having shown that the numerator
converges to zero it follows that $\tilde{f}_{1}\rightarrow_{p}f_{1,0}.$

Then, turning to 
\begin{eqnarray*}
\tilde{\delta} & = & \left(\underline{X}'\underline{X}\right)^{-1}\underline{X}'\left(y_{1}-y_{2}\tilde{f}_{1}\right)\\
 & = & \left(\underline{X}'\underline{X}\right)^{-1}\underline{X}'\left(\underline{X}\delta_{0}+\left(I_{n}+\rho_{0}M\right)\left(u_{1}-u_{2}f_{1,0}\right)\right)\\
 & + & \left(\underline{X}'\underline{X}\right)^{-1}\underline{X}'y_{2}\left(f_{1,0}-\tilde{f}_{1}\right)\\
 & \equiv & I+II
\end{eqnarray*}
where we used \eqref{eq:lin_IV} for $y_{1}-f_{1,0}y_{2}$. By Assumptions
\ref{assu:rank-new} and \ref{assu:conv_Mom} as well as \eqref{eq:Graham-v1-compact}
it follows that $n^{-1}\underline{X}'\underline{X}\rightarrow_{p}\mathscr{U}_{XX}$
where $\mathscr{U}_{XX}$ is positive definite and symmetric and $n^{-1}\underline{X}'y_{2}\rightarrow_{p}\mathscr{U}_{Xy}$
where $\mathscr{U}_{Xy}$ is bounded. This implies that $II\rightarrow_{p}0.$
For $I$ recall from above that $n^{-1}\underline{X}'\left(I_{n}+\rho_{0}M\right)\left(u_{1}-u_{2}f_{1,0}\right)\rightarrow_{p}0.$
This shows that $\tilde{\delta}\rightarrow_{p}\delta_{0}.$

Now consider the quadratic moment conditions. Let $\Upsilon_{1}=\left[X,y_{2}\right]$,
$\tilde{\vartheta}=\left(\tilde{\delta}^{\prime},\tilde{f}_{1}\right)^{\prime}$,
and observe that utilizing \eqref{eq:lin_IV} we have 
\begin{eqnarray*}
\epsilon^{+}\left(\rho,\tilde{f}_{1},\tilde{\delta}\right) & = & \left(I+\rho M\right)^{-1}[y_{1}-y_{2}\tilde{f}_{1}-\underline{X}\tilde{\delta}]\\
 & = & \left(I+\rho M\right)^{-1}[(I+\rho_{0}M)(u_{1}-f_{1,0}u_{2})-\Upsilon_{1}(\tilde{\vartheta}-\vartheta_{0})],
\end{eqnarray*}
 where, with slight abuse of notation relative to earlier definitions,
we now understand $\vartheta_{0}=(\delta_{0}^{\prime},f_{1,0})^{\prime}.$
Consequently
\begin{eqnarray*}
n^{-1}\epsilon^{+}(\rho,\tilde{f}_{1},\tilde{\delta})'A\epsilon^{+}(\rho,\tilde{f}_{1},\tilde{\delta}) & = & n^{-1}\left(\tilde{\vartheta}-\vartheta_{0}\right)'\Upsilon_{1}'\left(I+\rho M\right)^{-1}A\left(I+\rho M\right)^{-1}\Upsilon_{1}\left(\tilde{\vartheta}-\vartheta_{0}\right)\\
 &  & -n^{-1}2\left(\tilde{\vartheta}-\vartheta_{0}\right)'\Upsilon_{1}'\left(I+\rho M\right)^{-1}A\mathcal{M}\left(\rho\right)\left(u_{1}-f_{1,0}u_{2}\right)\\
 &  & +n^{-1}\left(u_{1}-f_{1,0}u_{2}\right)'\mathcal{M}\left(\rho\right)A\mathcal{M}\left(\rho\right)\left(u_{1}-f_{1,0}u_{2}\right)\\
 & = & I+II+III.
\end{eqnarray*}
By Assumption \ref{assu:conv_Mom} it follows that $\sup_{\rho}\left|n^{-1}\Upsilon_{1}'(I+\rho M)^{-1}A(I+\rho M)^{-1}\Upsilon_{1}-\mathcal{\mathscr{U_{\mathnormal{\Upsilon_{1}\Upsilon_{1}}}\left(\rho\right))}}\right|\rightarrow0$
a.s. for some uniformly bounded $U_{\mathnormal{\Upsilon_{1}\Upsilon_{1}}}\left(\rho\right)$.
Thus $I$ converges to zero uniformly in probability. Uniform convergence
in probability to zero for the second term, $II,$ follows by Assumption
\ref{assu:conv_Mom}.

Finally, using Lemma \ref{lem:id_q} and Assumption \ref{assu:conv_Mom}
it follows that 
\[
\sup_{\rho}\left|n^{-1}\left(u_{1}-f_{1,0}u_{2}\right)'\mathcal{M}\left(\rho\right)A\mathcal{M}\left(\rho\right)\left(u_{1}-f_{1,0}u_{2}\right)-n^{-1}\tr\left(\mathcal{M}\left(\rho\right)A\mathcal{M}\left(\rho\right)\Omega_{0}\right)\right|\rightarrow0\text{ a.s.}
\]
Then the consistency of $\tilde{\rho}$ follows from Lemma 3.1 in
\citet{potscher_dynamic_1997}, and analogous arguments as those used
for part (i) of the theorem.
\end{proof}

\section{\label{sec:AsyNormal}Asymptotic Normality of Parameter Estimators}

In the following we develop an asymptotically justified inference
theory for our proposed GMM estimator. For that purpose we first discuss
consistent estimators for the variance parameters $\gamma=(\gamma_{1},...,\gamma_{J})'$,
which are treated as nuisance parameters in the objective function
of our GMM estimator.
\begin{assumption}
\label{assu:GA1} Let $\omega_{j}\equiv\lim{}_{n\rightarrow\infty}n^{-1}\sum_{c=1}^{C}n_{c}1(\tau_{c}=j)$
and assume that $\omega_{j}$ exists and $\omega_{j}>0$ for all $j=1,...,J$.
\end{assumption}
Assumption \ref{assu:GA1} guarantees that there are enough data to
estimate $\gamma_{j}^{2}$ for each $j$ by requiring that the fraction
$\omega_{j}$ of all students in classrooms of type $j$ is asymptotically
non-negligible for all types of classrooms.

The first step inefficient but consistent estimators $\tilde{f}_{1}$,
$\tilde{\delta}$ and $\tilde{\rho}$ defined in \eqref{eq:f_tilde},
\eqref{eq:delta_tilde} and \eqref{eq:rho_tilde} can be used to obtain
consistent estimates of the variance parameters $\gamma.$ Form the
residuals 
\begin{equation}
\tilde{\epsilon}=\epsilon^{+}(\tilde{\theta})=\left(I_{n}+\tilde{\rho}M\right)^{-1}\left(y_{1}-y_{2}\tilde{f}_{1}-\underline{X}\tilde{\delta}\right)\label{eq:epstilde}
\end{equation}
and organize the residuals by classrooms $c$ as $\tilde{\epsilon}=\left(\tilde{\epsilon}_{1}^{\prime},....,\tilde{\epsilon}_{C}^{\prime}\right)^{\prime}$.
Then construct the variance estimators 
\begin{equation}
\hat{\gamma}_{j}^{2}=\left(N_{j}-p_{x}-1\right)^{-1}\sum_{c=1}^{C}\tilde{\epsilon}_{c}^{\prime}\tilde{\epsilon}_{c}1\left(\tau_{c}=j\right),\label{eq:varsigmahat}
\end{equation}
where $N_{j}=\sum_{c=1}^{C}n_{c}1\left(\tau_{c}=j\right)$ is the
number of students in classrooms of type $j$, and where $p_{x}$
denotes the number of columns in $\underline{X}$. Now set
\begin{equation}
\hat{\gamma}=\left(\hat{\gamma}_{1},\ldots,\hat{\gamma}_{J}\right)'.\label{eq:gam_hat}
\end{equation}

Efficient GMM estimators can now be formed by plugging $\hat{\gamma}$
into $u^{+}(\theta,\gamma)$ defined in \eqref{eq:uplus-1-1}. Recall
that $\Omega(\gamma)=\diag_{c=1}^{C}(\gamma_{\tau_{c}}^{2}I_{c}).$
In the following we denote with $K_{u}$ a generic finite constant
(which is taken, w.o.l.o.g., to be greater than one) and which is
invariant over $t=1,2$, $i=1,\ldots,n$, $n\geq1$. The central limit
theorem in Proposition 3 of \citet{kuersteiner_dynamic_2020} can
be applied to obtain distributional approximations for estimators
based on $u^{+}$. The central limit theorem of \citet{kuersteiner_dynamic_2020}
requires that $E\left[\left\vert u_{it}u_{is}\right\vert ^{1+\eta_{u}}|\mathcal{Z}_{n}\right]\leqslant K_{u}$
with $\eta_{u}>0$ for $t,s=1,2\text{ }$ and all $\ensuremath{i}$.
This condition is implied by Assumption \ref{assu:Proportional},
as Lemma \ref{lem:u_Msbl_MomB} demonstrates. When setting 
\begin{equation}
\tilde{\gamma}=\left(1,1,...,1\right)^{\prime}\label{eq:gam_tild}
\end{equation}
 it follows that $\Omega(\tilde{\gamma})=I_{n}$. The feasible efficient
GMM estimator for $\theta_{0}=\left(\rho_{0},f_{1,0},\delta_{0}'\right)'$
is defined as $\hat{\theta}\left(\hat{\gamma}\right)=\argmin_{\theta}Q_{n}\left(\theta,\hat{\gamma}\right)$
where $Q_{n}\left(\theta,\gamma\right)$ is given in \eqref{eq:Def_Qn}
and the inefficient counterpart $\hat{\theta}\left(\tilde{\gamma}\right)$
is defined in an analogous way. The next theorem establishes the limiting
distribution of $\hat{\theta}\left(\tilde{\gamma}\right)$ and $\hat{\theta}\left(\hat{\gamma}\right).$
The limiting distribution of $\hat{\theta}\left(\tilde{\gamma}\right)$
for $\tilde{\gamma}$ defined in \eqref{eq:gam_tild} provides a distributional
result for an inefficient GMM estimator. The limiting distribution
of the first step estimators defined in \eqref{eq:f_tilde}, \eqref{eq:delta_tilde}
and \eqref{eq:rho_tilde} can be deduced from the theorem by working
out the limiting distribution of $m_{\epsilon}\left(\theta\right).$
We omit the details.
\begin{thm}
\label{thm:CLT1}Let Assumptions \ref{assu:Potential_Outcomes-cov}-\ref{assu:GA1}
hold. Let $\hat{\gamma}$ be as defined in \eqref{eq:gam_hat} and
$\tilde{\gamma}$ as defined in \eqref{eq:gam_tild}. For $\gamma_{n}\rightarrow_{p}\gamma$
let $\plim_{n\rightarrow\infty}n^{-1/2}\partial m_{n}\left(\hat{\theta}\left(\gamma_{n}\right),\gamma_{n}\right)/\partial\theta=G(\gamma)$
and $\Xi=\plim_{n\rightarrow\infty}\Xi_{n}$, where $\Xi_{n}$ is
defined in \eqref{eq:Xi_n}. Then $\hat{\gamma}$$\rightarrow_{p}\gamma_{0}$,
and for $\gamma_{n}=\hat{\gamma}$ or $\gamma_{n}=\tilde{\gamma}$,
\[
\sqrt{n}\left(\hat{\theta}\left(\gamma_{n}\right)-\theta_{0}\right)\rightarrow_{d}N(0,\Psi_{\gamma})
\]
where the limiting variance covariance matrix $\Psi_{\gamma}$ has
the form 
\[
\Psi_{\gamma}=\left(G'\Xi^{-1}G\right)^{-1}G'\Xi^{-1}V_{\gamma}\Xi^{-1}G\left(G'\Xi^{-1}G\right)^{-1}.
\]
with $V_{\gamma}$ defined in Assumption \ref{assu:conv_Var_new}
and $G$ is the shorthand notation for $G(\gamma)$. When $\gamma_{n}=\hat{\gamma}$
then \textup{$V_{\gamma}=\Xi$ and $\Psi_{\gamma}=\left(G(\gamma_{0})^{\prime}\Xi^{-1}G(\gamma_{0})\right)^{-1}$.}
\end{thm}
In Appendix \ref{sec:variance} we provide explicit formulas for the
derivatives $G$. Together with the expressions for $\Xi_{n}$ it
is then easy to construct sample based estimators $\hat{G}$ and $\hat{\Xi}$.
Letting $\hat{\Psi}_{\gamma}=\left(\hat{G}^{\prime}\hat{\Xi}^{-1}\hat{G}\right)^{-1}$
it follows from a standard Slutsky argument that $\hat{\Psi}_{\gamma}^{-1/2}\sqrt{n}\left(\hat{\theta}\left(\hat{\gamma}\right)-\theta_{0}\right)\rightarrow_{d}N\left(0,I\right).$

\begin{proof}[Proof of Theorem \ref{thm:CLT1}]

First show consistency of $\hat{\gamma}$. Define $\Upsilon_{1}=\left[\underline{X},y_{2}\right]$
and $\vartheta=\left(\delta^{\prime},f_{1}\right)^{\prime}$. Observing
from \eqref{eq:lin_IV} that $y_{1}-f_{1,0}y_{2}-\underline{X}\delta_{0}=(I+\rho_{0}M)(u_{1}-f_{1,0}u_{2})$
we have 
\begin{eqnarray*}
\tilde{\epsilon} & = & \left(I_{n}+\tilde{\rho}M\right)^{-1}\left(y_{1}-\Upsilon_{1}\tilde{\vartheta}\right)\\
 & = & \left(u_{1}-u_{2}f_{1,0}\right)+\left(\mathcal{M}\left(\tilde{\rho}\right)-I\right)\left(u_{1}-u_{2}f_{1,0}\right)\\
 & + & (I+\tilde{\rho}M)\Upsilon_{1}\left(\vartheta_{0}-\tilde{\vartheta}\right)
\end{eqnarray*}
such that 
\begin{eqnarray*}
\sum_{c=1}^{C}\tilde{\epsilon}_{c}'\tilde{\epsilon}_{c}1(\tau_{c}=j) & = & \sum_{c=1}^{C}\left(u_{c1}-u_{c2}f_{1,0}\right)'\left(u_{c1}-u_{c2}f_{1,0}\right)1(\tau_{c}=j)\\
 &  & +\sum_{c=1}^{C}\left(u_{c1}-u_{c2}f_{1,0}\right)'\left(\mathcal{M}_{c}\left(\tilde{\rho}\right)-I\right)^{2}\left(u_{c1}-u_{c2}f_{1,0}\right)1(\tau_{c}=j)\\
 &  & +\sum_{c=1}^{C}\left(\vartheta_{0}-\tilde{\vartheta}\right)'\Upsilon_{1,c}'\left(I_{c}+\tilde{\rho}M_{c}^{\prime}\right)^{-1}\left(I_{c}+\tilde{\rho}M_{c}\right)^{-1}\Upsilon_{1,c}\left(\vartheta_{0}-\tilde{\vartheta}\right)1(\tau_{c}=j)\\
 &  & +2\sum_{c=1}^{C}\left(u_{c1}-u_{c2}f_{1,0}\right)'\left(\mathcal{M}_{c}\left(\tilde{\rho}\right)-I_{c}\right)\left(u_{c1}-u_{c2}f_{1,0}\right)1(\tau_{c}=j)\\
 &  & +2\sum_{c=1}^{C}\left(u_{c1}-u_{c2}f_{1,0}\right)'\left(\mathcal{M}_{c}\left(\tilde{\rho}\right)-I_{c}\right)'\left(I_{n}+\tilde{\rho}M_{c}\right)^{-1}\Upsilon_{1,c}\left(\vartheta_{0}-\tilde{\vartheta}\right)1(\tau_{c}=j)\\
 &  & +2\sum_{c=1}^{C}\left(u_{c1}-u_{c2}f_{1,0}\right)'\left(I_{c}+\tilde{\rho}M_{c}\right)^{-1}\Upsilon_{1,c}\left(\vartheta_{0}-\tilde{\vartheta}\right)1(\tau_{c}=j)\\
 & = & I+II+III+IV+V+VI.
\end{eqnarray*}
Now consider the components of $\left(N_{j}-p_{x}-1\right)^{-1}\sum_{c=1}^{C}\tilde{\epsilon}_{c}'\tilde{\epsilon}_{c}1(\tau_{c}=j)$.
For $I$, note that 
\begin{eqnarray*}
n^{-1}E\left[I\right] & = & n^{-1}\sum_{c=1}^{C}E\left[E\left[\left(u_{c1}-u_{c2}f_{1,0}\right)'\left(u_{c1}-u_{c2}f_{1,0}\right)|\mathcal{Z}_{n}^{*}\right]1(\tau_{c}=j)\right]\\
 & = & \left(\sigma_{1,0}^{2}+f_{1,0}^{2}\sigma_{2,0}^{2}\right)n^{-1}\sum_{c=1}^{C}n_{c}\rho_{\tau_{c},0}^{2}1(\tau_{c}=j)\\
 & = & \left(\sigma_{1,0}^{2}+f_{1,0}^{2}\sigma_{2,0}^{2}\right)\rho_{j,0}^{2}n^{-1}\sum_{c=1}^{C}n_{c}1(\tau_{c}=j)\\
 & = & \gamma_{j,0}^{2}(N_{j}/n)\rightarrow\gamma_{j,0}^{2}\omega_{j}
\end{eqnarray*}
recalling that $N_{j}/n\rightarrow\omega_{j}$ such that 
\[
\left(N_{j}-p_{x}-1\right)^{-1}E\left[I\right]=[n/(N_{j}-p_{x}-1)]n^{-1}E\left[I\right]\rightarrow\left(\sigma_{1,0}^{2}+f_{1,0}^{2}\sigma_{2,0}^{2}\right)\rho_{j,0}^{2}=\gamma_{j,0}^{2}.
\]
 By Assumption \ref{assu:conv_Mom} it follows that $n^{-1}I\rightarrow\gamma_{j,0}^{2}$
a.s. for $j=1,...,J$. For $II,$ $IV$ and $V$ consider 
\begin{align*}
\left(\mathcal{M}\left(\tilde{\rho}\right)-I\right) & =\left(I+\tilde{\rho}M\right)^{-1}M\left(\rho_{0}-\tilde{\rho}\right)\\
 & =\diag_{c}\left(-\frac{1}{\left(n_{c}-1-\tilde{\rho}\right)}I_{c}^{*}+\frac{1}{1+\tilde{\rho}}J_{c}^{*}\right)\left(\rho_{0}-\tilde{\rho}\right)
\end{align*}
 such that for $II$ 
\[
\left(\mathcal{M}\left(\tilde{\rho}\right)-I\right)^{2}=\diag_{c}\left(\frac{1}{\left(n_{c}-1-\tilde{\rho}\right)^{2}}I_{c}^{*}+\frac{1}{\left(1+\tilde{\rho}\right)^{2}}J_{c}^{*}\right)\left(\rho_{0}-\tilde{\rho}\right)^{2}.
\]
Then using the notation 
\[
C\left(\rho\right)=\diag_{c}\left(\frac{1(\tau_{c}=j)}{\left(n_{c}-1-\rho\right)^{2}}I_{c}^{*}+\frac{1(\tau_{c}=j)}{\left(1+\rho\right)^{2}}J_{c}^{*}\right)
\]
 in line with Assumption \ref{assu:conv_Mom} shows that 
\begin{eqnarray*}
n^{-1}II & = & \left(\rho_{0}-\tilde{\rho}\right)^{2}n^{-1}\left(u_{1}-u_{2}f_{1,0}\right)'C\left(\tilde{\rho}\right)\left(u_{1}-u_{2}f_{1,0}\right)\\
 & \leq & \left(\rho_{0}-\tilde{\rho}\right)^{2}\sup_{\rho}\left|n^{-1}\left(u_{1}-u_{2}f_{1,0}\right)'C\left(\rho\right)\left(u_{1}-u_{2}f_{1,0}\right)-E\left[n^{-1}\tr\left(C\left(\rho\right)\Omega_{0}\right)\right]\right|\\
 &  & +\left(\rho_{0}-\tilde{\rho}\right)^{2}\sup_{\rho}\left|\left[E\left[n^{-1}\tr\left(C\left(\rho\right)\Omega_{0}\right)\right]-\lim_{n\rightarrow\infty}E\left[n^{-1}\tr\left(C\left(\rho\right)\Omega_{0}\right)\right]\right]\right|\\
 &  & +\left(\rho_{0}-\tilde{\rho}\right)^{2}\sup_{\rho}\left|\lim_{n\rightarrow\infty}E\left[n^{-1}\tr\left(C\left(\rho\right)\Omega_{0}\right)\right]\right|\\
 & = & o_{p}\left(1\right),
\end{eqnarray*}
observing that $\rho_{0}-\tilde{\rho}=o_{p}\left(1\right)$ by Theorem
\ref{thm:Consistency}(ii), that the first and second supremum converges
to zero a.s. by Assumption \ref{assu:conv_Mom}(i) and that $\lim_{n}E\left[n^{-1}\tr\left(C\left(\rho\right)\Omega_{0}\right)\right]$
is uniformly bounded by Assumption \ref{assu:conv_Mom}(ii). Arguments
analogous to these show that the remaining terms of $n^{-1}\sum_{c=1}^{C}\tilde{\epsilon}_{c}'\tilde{\epsilon}_{c}1(\tau_{c}=j)$
converge to zero in probability. The details are omitted. This establishes
that 
\[
\hat{\gamma}_{j}^{2}=\left(N_{j}-p_{x}-1\right)^{-1}\sum_{c=1}^{C}\tilde{\epsilon}_{c}'\tilde{\epsilon}_{c}1(\tau_{c}=j)\rightarrow_{p}\left(\sigma_{0,1}^{2}+f_{1,0}^{2}\sigma_{0,2}^{2}\right)\rho_{j,0}^{2}=\gamma_{j,0}^{2}.
\]

To obtain the limiting distribution note that $Q_{n}\left(\theta,\gamma\right)$
defined in \eqref{eq:Def_Qn} can be written as 
\[
Q_{n}\left(\theta,\gamma_{n}\right)=\frac{1}{n}m_{n}\left(\theta,\gamma_{n}\right)'\hat{\Xi}_{n}^{-1}m_{n}\left(\theta,\gamma_{n}\right)
\]
where $\hat{\Xi}_{n}$ is defined after \eqref{eq:Xi_n}. Using a
mean value expansion of $m_{n}(\hat{\theta},\gamma_{n})$ around $\theta_{0}$
we obtain 
\begin{align*}
0 & =\frac{\sqrt{n}}{2}\frac{\partial Q_{n}\left(\hat{\theta},\gamma_{n}\right)}{\partial\theta}=\frac{\partial m_{n}(\hat{\theta},\gamma_{n})^{\prime}}{\sqrt{n}\partial\theta}\hat{\Xi}_{n}^{-1}m_{n}(\hat{\theta},\gamma_{n})\\
 & =\frac{\partial m_{n}(\hat{\theta},\gamma_{n})^{\prime}}{\sqrt{n}\partial\theta}\hat{\Xi}_{n}^{-1}m_{n}(\theta_{0},\gamma_{n})+\frac{\partial m_{n}(\hat{\theta},\gamma_{n})^{\prime}}{\sqrt{n}\partial\theta}\hat{\Xi}_{n}^{-1}\frac{\partial m_{n}(\bar{\theta},\gamma_{n})}{\sqrt{n}\partial\theta}\sqrt{n}\left(\hat{\theta}-\theta_{0}\right)
\end{align*}
where $\left\Vert \bar{\theta}-\theta_{0}\right\Vert \leq\left\Vert \hat{\theta}-\theta_{0}\right\Vert $
and with some abuse of notation it is understood that the between
value $\bar{\theta}$ differs among rows of $m_{n}\left(.\right).$
Using the explicit derivatives given in Appendix \ref{sec:variance}
it follows from Assumption \ref{assu:conv_Mom}(i) and (ii) that for
any sequence $\gamma_{n}$ with $\gamma_{n}\rightarrow_{p}\gamma_{*}$
we have $n^{-1/2}\partial m_{n}\left(\hat{\theta},\gamma_{n}\right)/\partial\theta\rightarrow_{p}G(\gamma_{\ast})$
and $n^{-1/2}\partial m_{n}\left(\bar{\theta},\gamma_{n}\right)/\partial\theta\rightarrow_{p}G(\gamma_{\ast})$
where $G(\gamma_{\ast})$ is fixed matrix with full column rank and
$\hat{\Xi}_{n}\rightarrow_{p}\Xi$ with $\Xi$ a full-rank square
matrix. This implies that 
\[
\sqrt{n}\left(\hat{\theta}-\theta_{0}\right)=-\left(\left(G'\Xi^{-1}G\right)^{-1}G'\Xi^{-1}+o_{p}\left(1\right)\right)m_{n}\left(\theta_{0},\gamma_{n}\right).
\]
 It remains to establish the limiting distribution of $m_{n}\left(\theta_{0},\gamma_{n}\right)$
where $\gamma_{n}$ is a possibly random sequence with $\gamma_{n}\rightarrow_{p}\gamma_{*}.$
To apply the results of \citet{kuersteiner_dynamic_2020} note that
the moment vector $m_{n}\left(.\right)$ can be represented as the
sum over the linear and quadratic terms by defining $h_{i}^{*}=\left(h_{i1},\ldots,h_{iq},0\right)$
and $a_{ij}^{\ast}=\left(0,....,0,a_{ij}\right)$, where $(h_{i1},\ldots,h_{iq})$
denotes the $i$-th row of $H$ and $a_{ij}$ the $i,j$-th element
of $A.$ Then, $m_{n}\left(\theta,\gamma\right)=n^{-1/2}\sum_{i=1}^{n}h_{i}^{*\prime}u_{i}^{+}(\theta,\gamma)+n^{-1/2}\sum_{j,i=1}^{n}a_{ij}^{*\prime}u_{i}^{+}(\theta,\gamma)u_{j}^{+}(\theta,\gamma)$
has the same form as Equation (23) in \citet{kuersteiner_dynamic_2020}.
It is then sufficient to check that the conditions for Proposition
3 in \citet{kuersteiner_dynamic_2020} hold. Proposition 3 requires
that Assumptions 1-3 of \citet{kuersteiner_dynamic_2020} hold. We
abbreviate assumptions in \citet{kuersteiner_dynamic_2020} with A-KP
to avoid confusion with assumptions in this paper.

The moment bounds of A-KP 1(i) hold by Assumption \ref{assu:Proportional}
and Theorem \ref{thm:decomp} which implies that $u_{it}$ and $\mu_{i}^{\ast}$
have bounded moments, and Assumption \ref{assu:rank-new}(i) which
bounds the moments of $x_{it}$ and $z_{it}.$ A-KP 1(ii), Eq. (15),
holds by Lemma \ref{lem:u_Msbl_MomB}, A-KP 1(ii), Eq. (16) holds
by Assumption \ref{assu:CondVar} and A-KP 1(ii), Eq. (17) holds by
Lemma \ref{lem:u_Msbl_MomB}, and the parametric restrictions on the
conditional cross-sectional variances hold by Assumption \ref{assu:CondVar}.

Now turning to A-KP 2, the $\sigma$-field that corresponds to $\mathcal{B}_{n,t}$
is $\mathcal{A}_{n,\left(t-1\right)n}$ defined in \eqref{eq:A_n,(t-1)n+i-1}.
Since $\mathcal{Z}_{n}^{*}\subset\mathcal{A}_{n,\left(t-1\right)n}$,
the fact that $h_{i}$ and $a_{ij}$ are assumed measurable w.r.t
to $\mathcal{Z}_{n}\subset\mathcal{Z}_{n}^{*}$ and elements of $h_{i}$
and the row sums of the absolute elements (and powers thereof) of
$A$ are uniformly bounded by Assumptions \ref{assu:quaA}(i)(iii)
and \ref{assu:rank-new}, A-KP 2(i) and (ii) hold. We also normalize
$f_{2}=1$ and restrict the parameter space of $f_{1}$ to a compact
interval in Assumption \ref{assu:PAR} such that A-KP 2(iii) holds.

Finally, consider A-KP 3 which is identical to Assumption \ref{assu:conv_Mom}.
In summary, all conditions of Proposition 3 in \citet{kuersteiner_dynamic_2020}
hold. It follows that $m_{n}\left(\theta_{0},\gamma\right)\rightarrow^{d}V_{\gamma}^{1/2}\xi\text{ }$
where $\xi\sim N\left(0,I_{q+1}\right)$ and $V_{\gamma}$ is the
limit of the variance-covariance matrix of $m_{n}(\theta_{0},\gamma)$
defined in \eqref{eq:Moments_SI_Graham}. For any conformable, matrix
$\mathscr{A}$ it follows that $\mathscr{A}m_{n}\left(\theta_{0},\gamma\right)\rightarrow^{d}\left(\mathscr{A}V\mathscr{A}^{\prime}\right)^{1/2}\xi$.
By choosing $\mathscr{A}=\left(G'\Xi^{-1}G\right)^{-1}G'\Xi^{-1}$
the claim of the theorem is proven.
\end{proof}

\section{Auxiliary Lemmas\label{sec:Auxiliary-Lemmas}}

To prove Theorems \ref{thm:Consistency} and \ref{thm:CLT1}, we establish
two preliminary lemmas. We also provide proofs for Lemma \ref{lem:id_nocov}
and Corollary \ref{cor:TS_ID} in this section. Note that many matrices
in this paper, e.g., $M$, $\Omega(\gamma)$, $\Sigma_{t}$ defined
in Assumption \ref{assu:CondVar} can be written in the form of $\diag_{c=1}^{C}\left(p_{c}I_{c}^{\ast}+q_{c}J_{c}^{\ast}\right),$where
$I_{c}^{\ast}=I_{c}-\mathbf{1}_{c}\mathbf{1}_{c}^{\prime}/n_{c}$,
$J_{c}^{\ast}=\mathbf{1}_{c}\mathbf{1}_{c}^{\prime}/n_{c}$ are the
residual projection matrix and the projection matrix onto $\mathbf{1}_{c}$
respectively. To see this note that 
\begin{align}
M & =\diag_{c=1}^{C}\left(-\frac{1}{n_{c}-1}I_{c}^{\ast}+J_{c}^{\ast}\right),\label{eq:M_IJ_Rep}\\
I+\rho M & =\diag_{c=1}^{C}\left(\frac{n_{c}-1-\rho}{n_{c}-1}I_{c}^{\ast}+\left(1+\rho\right)J_{c}^{\ast}\right),\label{eq:I-rhoM_IJ_Rep}\\
\Omega(\gamma) & =\diag_{c=1}^{C}\left(\gamma_{\tau_{c}}^{2}I_{c}\right)=\diag_{c=1}^{C}\left(\gamma_{\tau_{c}}^{2}I_{c}^{\ast}+\gamma_{\tau_{c}}^{2}J_{c}^{\ast}\right),\label{eq:Omega_IJ_Rep}\\
\Sigma_{t} & =\sigma_{0,t}^{2}\diag_{c=1}^{C}\left(\rho_{0,\tau_{c}}^{2}I_{c}\right)=\diag_{c=1}^{C}\left(\rho_{0,\tau_{c}}^{2}\sigma_{0,t}^{2}I_{c}^{\ast}+\rho_{0,\tau_{c}}^{2}\sigma_{0,t}^{2}J_{c}^{\ast}\right),\,t=1,2.\label{eq:Sigma_IJ_Rep}
\end{align}

The following lemma summarizes some well understood key properties
of these matrices for the convenience of the reader. Proofs are straight
forward and omitted, and can e.g., be found in \citet{kuersteiner_efficient_2023}.
\begin{lem}
\label{lem:mat_I-J} The matrices $I_{c}^{\ast}$, $J_{c}^{\ast}$
and \textup{$pI_{c}^{\ast}+qJ_{c}^{\ast}$} have the following properties:

(i) The matrices $I_{c}^{\ast}$, $J_{c}^{\ast}$ are symmetric, idempotent,
orthogonal, $I_{c}^{*}+J_{c}^{*}=I$ and $I_{c}^{\ast}\mathbf{1}_{c}=0$,
$J_{c}^{\ast}\mathbf{1}_{c}=\mathbf{1}_{c}$.

(ii) $\det\left(pI_{c}^{\ast}+qJ_{c}^{\ast}\right)=p^{(n_{c}-1)}q$.

(iii) The eigenvalues of $pI_{c}^{\ast}+qJ_{c}^{\ast}$ are $p$ ($n_{c}-1$
times) and $q$ (once).

(iv) For $p_{1}$, \textup{$p_{2}$, $q_{1}$} and $q_{2}$ arbitrary
constants, $(p_{1}I_{c}^{\ast}+q_{1}J_{c}^{\ast})(p_{2}I_{c}^{\ast}+q_{2}J_{c}^{\ast})=(p_{1}p_{2}I_{c}^{\ast}+q_{1}q_{2}J_{c}^{\ast})$,
and thus multiplication of matrices of the form $pI_{c}^{\ast}+qJ_{c}^{\ast}$
is commutative and associative.

(v) If $p\neq0,$ $q\neq0$, $(pI_{c}^{\ast}+qJ_{c}^{\ast})$ is invertible
with $(pI_{c}^{\ast}+qJ_{c}^{\ast})^{-1}=p^{-1}I_{c}^{\ast}+q^{-1}J_{c}^{\ast}$.

(vi) $\tr(pI_{c}^{\ast}+qJ_{c}^{\ast})=p(n_{c}-1)+q$.

(vii) Both $I_{c}^{\ast}$ and $J_{c}^{\ast}$ have row and column
sum uniformly bounded in absolute value. When $p$ and $q$ are bounded
in absolute value, the row and column sums of $pI_{c}^{\ast}+qJ_{c}^{\ast}$
are uniformly bounded in absolute value.
\end{lem}
\begin{cor}
\label{cor:pIqJ}Suppose Assumptions \ref{assu:PAR}, \ref{assu:quaA},
\ref{assu:size}, \ref{assu:rank-new}, \ref{assu:CondVar}, and \ref{assu:GA1}
hold. Then $I+\rho M$, $\Omega(\gamma)$, and $\Sigma_{t}$ are non-singular.
If $S(\theta,\gamma)$ is an element or product matrix of elements
from 
\begin{equation}
\left\{ M,(I+\rho M),(I+\rho M)^{-1},\Omega(\gamma),\Omega(\gamma)^{-1},\Sigma_{t}^{-1},\Sigma_{t}\right\} ,\label{eq:Sset}
\end{equation}
then $S(\theta,\gamma)$ can be written in the form of $\diag_{c=1}\{p_{c}I_{c}^{\ast}+q_{c}J_{c}^{\ast}\}$with
$p_{c}$, $q_{c}$ continuously differentiable in $\theta$ and $\gamma$.
The matrix $S(\theta,\gamma)A$ is measurable w.r.t $\mathcal{Z}_{n}$
. The row and column sums of $S(\theta,\gamma)$ and $S(\theta,\gamma)A$
are uniformly bounded in absolute value. .
\end{cor}
\begin{proof}
By Assumptions \ref{assu:PAR}, \ref{assu:size} and \ref{assu:CondVar},
$\left(n_{c}-1-\rho\right)/\left(n_{c}-1\right)$, $\left(1+\rho\right)$,
$\gamma_{j}^{2}$, $\rho_{0,\tau_{c}}^{2}$, and $\sigma_{0,t}^{2}$
in \eqref{eq:I-rhoM_IJ_Rep}, \eqref{eq:Omega_IJ_Rep}, \eqref{eq:Sigma_IJ_Rep}
are positive and uniformly bounded away from 0. Hence by Lemma \ref{lem:mat_I-J}(v),
$I+\rho M$, $\Omega(\gamma)$ and $\Sigma_{t}$ are non-singular.
An inspection of Equations \eqref{eq:M_IJ_Rep}-\eqref{eq:Sigma_IJ_Rep}
shows that all matrices in the set can be written in the form of $\diag_{c=1}^{C}\left(p_{c}I_{c}^{\ast}+q_{c}J_{c}^{\ast}\right),$
with $p_{c}$ and $q_{c}$ continuously differentiable in $\theta$
and $\gamma$ and uniformly bounded. Also note that by Assumption
\ref{assu:quaA}, the row and column sums of $A$ are uniformly bounded
in absolute value. Utilizing remark A.1 in \citet{kelejian_estimation_2004}
the corollary thus follows from Lemma \ref{lem:mat_I-J}.
\end{proof}
\begin{lem}
\label{lem:id_q}Suppose Assumptions \ref{assu:PAR}, \ref{assu:quaA},
\ref{assu:size} and \ref{assu:CondVar} hold. Consider the block
diagonal matrix $\mathcal{M}\left(\rho\right)=\diag{}_{c=1}^{C}\left(\mathcal{M}_{c}\left(\rho\right)\right)$
with $\mathcal{M}_{c}\left(\rho\right)=\left(I_{c}+\rho M_{c}\right)^{-1}\left(I_{c}+\rho_{0}M_{c}\right)$.
Then for all $\gamma\in\Gamma$,
\[
\tr\left(\mathcal{M}\left(\rho\right)\Omega(\gamma)^{-1/2}A\Omega(\gamma)^{-1/2}\mathcal{M}\left(\rho\right)\Omega_{0}\right)=0
\]
and 
\[
\lim_{n\rightarrow\infty}\frac{1}{n}\tr\left(\mathcal{M}\left(\rho\right)\Omega(\gamma)^{-1/2}A\Omega(\gamma)^{-1/2}\mathcal{M}\left(\rho\right)\Omega_{0}\right)=0
\]
each has only one solution for $\rho\in\left(-1,1\right)$. This solution
is given by $\rho=\rho_{0}$.
\end{lem}
\begin{proof}
Let $D(\rho,\gamma)=\mathcal{M}\left(\rho\right)\Omega(\gamma)^{-1/2}A\Omega(\gamma)^{-1/2}\mathcal{M}\left(\rho\right)\Omega_{0}$
and observe that $D(\rho,\gamma)=\diag{}_{c=1}^{C}\left(D_{c}(\rho,\gamma)\right)$
with
\begin{eqnarray*}
D_{c}(\rho,\gamma) & = & \mathcal{M}_{c}\left(\rho\right)\left(\gamma_{\tau_{c}}^{2}I_{c}\right)^{-1/2}A_{c}\left(\gamma_{\tau_{c}}^{2}I_{c}\right)^{-1/2}\mathcal{M}_{c}\left(\rho\right)\left(\gamma_{0,\tau_{c}}^{2}I_{c}\right)\\
 & = & \frac{\gamma_{\tau_{c},0}^{2}}{\gamma_{\tau_{c}}^{2}}\mathcal{M}_{c}\left(\rho\right)A_{c}\mathcal{M}_{c}\left(\rho\right).
\end{eqnarray*}
Noting that $\tr\left(A_{c}\right)=0$ and hence $\tr\left(A_{c}I_{c}^{\ast}\right)=\tr\left(A_{c}(I-J_{c}^{\ast})\right)=-\tr\left(A_{c}J_{c}^{\ast}\right)$,
and in light of Lemma \ref{lem:mat_I-J}
\begin{eqnarray*}
\mathcal{M}_{c}\left(\rho\right) & = & \left(I_{c}+\rho M_{c}\right)^{-1}\left(I_{c}+\rho_{0}M_{c}\right)=\frac{n_{c}-1-\rho_{0}}{n_{c}-1-\rho}I_{c}^{\ast}+\frac{1+\rho_{0}}{1+\rho}J_{c}^{\ast}.
\end{eqnarray*}
We have
\begin{eqnarray*}
\tr\left(\mathcal{M}_{c}\left(\rho\right)A_{c}\mathcal{M}_{c}\left(\rho\right)\right) & = & \tr\left(A_{c}\left[\left(\frac{n_{c}-1-\rho_{0}}{n_{c}-1-\rho}\right)^{2}I_{c}^{\ast}+\left(\frac{1+\rho_{0}}{1+\rho}\right)^{2}J_{c}^{\ast}\right]\right)\\
 & = & \left[\left(\frac{1+\rho_{0}}{1+\rho}\right)^{2}-\left(\frac{n_{c}-1-\rho_{0}}{n_{c}-1-\rho}\right)^{2}\right]\tr\left(A_{c}J_{c}^{\ast}\right)\\
 & = & (\rho_{0}-\rho)C(\rho,n_{c})\tr\left(A_{c}J_{c}^{\ast}\right)
\end{eqnarray*}
where 
\[
C(\rho,n_{c})=\left[\left(\frac{1+\rho_{0}}{1+\rho}\right)+\left(\frac{n_{c}-1-\rho_{0}}{n_{c}-1-\rho}\right)\right]\left[\frac{1}{1+\rho}+\frac{1}{n_{c}-1-\rho}\right].
\]
Hence 
\[
\tr\left(D(\rho,\gamma)\right)=\sum_{c=1}^{C}\tr\left(D_{c}(\rho,\gamma)\right)=(\rho_{0}-\rho)\sum_{c=1}^{C}\frac{\gamma_{\tau_{c},0}^{2}}{\gamma_{\tau_{c}}^{2}}C(\rho,n_{c})\tr\left(A_{c}J_{c}^{\ast}\right).
\]
Clearly, at $\rho=\rho_{0}$, for all $\gamma\in\Gamma$, we have
$\tr\left(D(\rho_{0},\gamma)\right)=0$ for all $n$.

Next observe that $\tr\left(A_{c}J_{c}^{*}\right)=\mathbf{1}_{c}^{'}A_{c}\mathbf{1}_{c}/n_{c}\geqslant\text{\underbar{K}}_{a}>0$
and that for any $n_{c}\geqslant2$ and $-1<\rho<1$, both $n_{c}-1-\rho$
and $1+\rho$ are positive, and therefore $C(\rho,n_{c})$ is positive
for all $\rho\in\left(-1,1\right)$. Therefore, $\tr\left[\mathcal{M}_{c}\left(\rho\right)A_{c}\mathcal{M}_{c}\left(\rho\right)\right]<0$
if $\rho>\rho_{0}$, $\tr\left[\mathcal{M}_{c}\left(\rho\right)A_{c}\mathcal{M}_{c}\left(\rho\right)\right]>0$
if $\rho<\rho_{0}$ and $\tr\left[\mathcal{M}_{c}\left(\rho\right)A_{c}\mathcal{M}_{c}\left(\rho\right)\right]$
has the same sign for all $c$. In all, for any value of $\gamma$
the only solution for 
\[
\tr\left(\mathcal{M}\left(\rho\right)\Omega(\gamma,\theta)^{-1/2}A\Omega(\gamma,\theta)^{-1/2}\mathcal{M}\left(\rho\right)\Omega_{0}\right)=0
\]
 is $\rho=\rho_{0}$. Furthermore, by Assumptions \ref{assu:PAR}
and \ref{assu:size} there exists some constants $c_{D}$ and $c_{\gamma}$
such that for any $\rho_{0}$ in the interior of the parameter space,
$C(\rho,n_{c})>c_{D}>0$ and $\frac{\gamma_{\tau_{c},0}^{2}}{\gamma_{\tau_{c}}^{2}}>c_{\gamma}>0$.
This implies that, 
\[
\left|\tr(D(\rho))\right|\geqslant|\rho_{0}-\rho|\sum_{c=1}^{C}\text{\underbar{K}}_{a}c_{D}c_{\gamma},
\]
\[
\left|\frac{1}{n}\tr(D(\rho))\right|\geqslant|\rho_{0}-\rho|\text{\underbar{K}}_{a}c_{D}c_{\gamma}\frac{C}{n}.
\]
By Assumption \ref{assu:size}, $C/n\geq1/\bar{n}_{c}>0$ such that
the conclusion follows.
\end{proof}
Proofs of Additional Identification Results.
\begin{proof}[Proof of Lemma \ref{lem:id_nocov} and Corollary \ref{cor:TS_ID}]

Lemma \ref{lem:id_nocov} is a special case of Theorem \ref{thm:general_id}
without covariates. Note that Assumptions \ref{assu:Potential_Outcomes-cov},
\ref{assu:Prop_Gen}, \ref{assu:PAR}, \ref{assu:quaA}, \ref{assu:size}
and \ref{assu:CondVar} hold. Without covariates Equation \eqref{eq:Kn1}
reduces to $K_{n}\left(\rho,\gamma\right)=z'V\left(\rho,\gamma\right)\mu^{*}$
where $V\left(\rho,\gamma\right)=\Omega\left(\gamma\right)^{-1/2}\left(I+\rho M\right)^{-1}.$

Recall that $z=\left[z_{1}^{\prime},...,z_{c}^{\prime},...,z_{C}^{\prime}\right]^{\prime}$
with $z_{c}=\grave{z}_{c}\mathbf{1}_{c}$. Utilizing Lemma \ref{lem:mat_I-J},

\[
(I_{c}+\rho M_{c})^{-1}z_{c}=\left(\frac{n_{c}-1-\rho}{n_{c}-1}I_{c}^{\ast}+\frac{1}{1+\rho}J_{c}^{\ast}\right)\grave{z}_{c}\mathbf{1}_{c}=\frac{1}{1+\rho}\grave{z}_{c}\mathbf{1}_{c}=\frac{z_{c}}{1+\rho}
\]
leads to 
\begin{align*}
n^{-1}E\left[y_{2}^{\prime}\Omega(\gamma)^{-1/2}(I+\rho M)^{-1}z|\mathcal{Z}_{n}\right] & =\frac{1}{1+\rho}n^{-1}E\left[y_{2}^{\prime}\Omega(\gamma)^{-1/2}z|\mathcal{Z}_{n}\right]\geq\frac{1}{1+\rho}K_{y}>0
\end{align*}
where the inequality follows from the conditions imposed in the lemma,
in particular \eqref{eq:Graham_ID_Cond}. With $\rho\in(-1,1)$, $E\left[K_{n}\left(\rho,\gamma\right)|\mathcal{Z}_{n}\right]\neq0$
such that \eqref{eq:mom-1} only has one solution, $f_{1}=f_{1,0}$.
Identification of $\rho_{0}$ from the quadratic moment is not affected
by the absence of covariates. Thus Lemma \ref{lem:id_nocov} follows
from Theorem \ref{thm:general_id}.

To prove Corollary \ref{cor:TS_ID} observe that when $z=\mathbf{1}_{n}$,
$n^{-1}\sum_{c=1}^{C}E\left[y_{c2}^{\prime}\mathbf{1}_{c}\right]>K_{y}>0$
implies that \eqref{eq:Graham_ID_Cond} holds, observing that 
\[
n^{-1}\left|E\left[y_{2}^{\prime}\Omega(\gamma)^{-1/2}\mathbf{1}_{n}\left|\mathcal{Z}_{n}\right.\right]\right|=n^{-1}E\sum_{c=1}^{C}\left[\frac{y_{c2}^{\prime}\mathbf{1}_{c}}{\gamma_{\tau_{c}}}\left|\mathcal{Z}_{n}\right.\right]\geq\frac{1}{K_{\gamma}}\frac{1}{n}\sum_{c=1}^{C}E\left[y_{c2}^{\prime}\mathbf{1}_{c}\left|\mathcal{Z}_{n}\right.\right]\geq\frac{K_{y}}{K_{\gamma}}>0.
\]
\end{proof}

\section{Gradient Vector\label{sec:variance}}

In this section, we derive the gradient vector for the two-step-GMM
estimator, which is a component of the limiting variance covariance
matrix in Theorem \ref{thm:CLT1}. Recall that $\theta=\left(\rho,f_{1},\delta'\right)'$,
\begin{align*}
u^{+}(\theta,\gamma) & =\Omega(\gamma)^{-1/2}(I+\rho M)^{-1}\left(y_{1}-f_{1}y_{2}-\underline{X}\delta\right)\\
 & =\Omega(\gamma)^{-1/2}(I+\rho M)^{-1}\left((f_{10}-f_{1})\mu^{\ast}+\underline{X}(\delta(f_{1},\rho_{0},\beta_{0})-\delta)+(I+\rho_{0}M)(u_{1}-f_{1}u_{2})\right)
\end{align*}
 and that $y_{t}=\mu^{*}f_{t,0}+\underline{X}\beta_{0,t}+(I+\rho_{0}M)u_{t}$.

\begin{align*}
\frac{\partial m_{n}(\theta,\gamma)/\sqrt{n}}{\partial\theta^{\prime}} & =\frac{1}{n}\left[\begin{array}{c}
H^{\prime}\frac{\partial u^{+}}{\partial\theta^{\prime}}\\
2u^{+\prime}A\frac{\partial u^{+}}{\partial\theta^{\prime}}
\end{array}\right]=\frac{1}{n}\left[\begin{array}{ccc}
H^{\prime}\frac{\partial u^{+}}{\partial\rho} & H^{\prime}\frac{\partial u^{+}}{\partial f_{1}} & H^{\prime}\frac{\partial u^{+}}{\partial\delta^{\prime}}\\
2u^{+\prime}A\frac{\partial u^{+}}{\partial\rho} & 2u^{+\prime}A\frac{\partial u^{+}}{\partial f_{1}} & 2u^{+\prime}A\frac{\partial u^{+}}{\partial\delta^{\prime}}
\end{array}\right].
\end{align*}

The derivative of $u^{+}(\theta,\gamma)$ w.r.t $\theta=(\rho,f_{1},\delta^{\prime})^{\prime}$
are
\begin{align*}
\partial\frac{u^{+}(\theta,\gamma)}{\partial f_{1}} & =-\Omega(\gamma)^{-1/2}(I+\rho M)^{-1}y_{2}\\
 & =-\Omega(\gamma)^{-1/2}(I+\rho M)^{-1}[\mu^{\ast}+\underline{X}\beta_{0,2}+(I+\rho_{0}M)u_{2}].
\end{align*}
\begin{align*}
\partial\frac{u^{+}(\theta,\gamma)}{\partial\rho} & =-\Omega(\gamma)^{-1/2}(I+\rho M)^{-1}M(I+\rho M)^{-1}(y_{1}-f_{1}y_{2}-\underline{X}\delta)\\
 & =-\Omega(\gamma)^{-1/2}(I+\rho M)^{-1}M(I+\rho M)^{-1}\\
 & \times\left((f_{10}-f_{1})\mu^{\ast}+\underline{X}(\delta(f_{1},\rho_{0},\beta_{0})-\delta)+(I+\rho_{0}M)(u_{1}-f_{1}u_{2})\right)
\end{align*}
\begin{align*}
\partial\frac{u^{+}(\theta,\gamma)}{\partial\delta} & =-\Omega(\gamma)^{-1/2}(I+\rho M)^{-1}\underline{X}.
\end{align*}

It thus is obvious that $u^{+}(\theta,\gamma)$ and each element of
$\frac{u^{+}(\theta,\gamma)}{\partial\theta}$ can be written as linear
combinations of $S(\theta,\gamma)\Upsilon$, where $S(\theta,\gamma)$
is an element or the product of elements of the matrix set in \eqref{eq:Sset}
and $\Upsilon\in\left\{ H,\mu^{*},u_{1},u_{2}\right\} $. Applying
Corollary \ref{cor:pIqJ} and Assumption \ref{assu:conv_Mom}, the
gradient vector converges uniformly to its expected value.

\section{GMM with Randomly Missing Variables\label{sec:Randomly_Missing}}

If class-level characteristics like teacher's race are missing and
the variable is controlled in the model, then the whole class is dropped
from the sample. Under the assumption that these characteristics are
missing randomly, dropping these classes should not affect the consistency
of our estimator. Below we discuss the case when individual level
variables (control variable or outcome variable) are missing.

Suppose class $c$ has $n_{c}$ students, and $n_{c}^{obs}\leqslant n_{c}$
of them are observed. which We define the case where $n_{c}^{obs}=n_{c}$
as having non-missing values for all variables. Without loss of generality,
we order the vectors for the class such that observed students are
at the top and unobserved are at the bottom. That is, $y_{c1}=(y_{c1}^{obs\prime},y_{c1}^{unobs\prime})^{\prime}$,
where $y_{c1}$ is the $n_{c}\times1$ vector of test score 1 for
all students in class $c$, $y_{c1}^{obs}$ is the $n_{c}^{obs}\times1$
vector of test score 1 for all observed students, and $y_{c1}^{unobs}$
is the $n_{c}^{unobs}\times1$ vector of test score 1 for all unobserved
students. Note that $n_{c}^{unobs}=n_{c}-n_{c}^{obs}$. Similarly,
$y_{c2}=(y_{c2}^{obs\prime},y_{c2}^{unobs\prime})^{\prime}$, $x_{c}=(x_{c}^{obs\prime},x_{c}^{unobs\prime})^{\prime}$,
etc.

The model for class $c$ is

\begin{equation}
y_{1c}-f_{1}y_{2c}=\text{\underbar{X}}_{c}\delta+(I_{c}+\rho M_{c})(u_{1c}-f_{1}u_{2c}).\label{eq:main}
\end{equation}
where $\underline{X}_{c}=\left(v_{c}^{c},w_{1,c}^{c},w_{2,c}^{c},v_{c}^{p},w_{1,c}^{p},w_{2,c}^{p},M_{c}v_{c}^{p},M_{c}w_{1,c}^{p},M_{c}w_{2,c}^{p}\right)$.
Let $x_{c}^{c}=[v_{c}^{c},w_{1,c}^{c},w_{2,c}^{c}]$ be class level
variables, $x_{c}^{p}=[v_{c}^{p},w_{1,c}^{p},w_{2,c}^{p}]$ be student
personal level characteristics. The model for class $c$ then can
be written more explicitly as

\begin{equation}
y_{1c}-f_{1}y_{2c}=x_{c}^{c}\delta_{c}+x_{c}^{p}\delta_{p}+M_{c}x_{c}^{p}\delta_{m}+(I_{c}+\rho M_{c})(u_{1c}-f_{1}u_{2c}).\label{eq:main-1}
\end{equation}
The composite error term satisfies $E(u_{1c}-f_{1}u_{2c}|\mathcal{Z}_{n}^{\ast})=0$
and $\textrm{var}(u_{1c}-f_{1,0}u_{2c}|\mathcal{Z}_{n}^{\ast})=\gamma_{\tau_{c}}^{2}I_{c}$.

Let $D_{c}=[I_{n_{c}^{obs}\times n_{c}^{obs}},0_{n_{c}^{obs}\times n_{c}^{nob}}]$
be the $n_{c}^{obs}\times n_{c}$ matrix that transforms the matrix
for the whole class to that for the observed students. That is $D_{c}y_{c1}=y_{c1}^{obs}$,
etc. Premultiplying both sides of \eqref{eq:main-1} by $D_{c}$,
we have
\[
D_{c}y_{1c}-f_{1}D_{c}y_{2c}=D_{c}x_{c}^{c}\delta_{c}+D_{c}x_{c}^{p}\delta_{p}+D_{c}M_{c}x_{c}^{p}\delta_{m}+D_{c}(I_{c}+\rho M_{c})(u_{1c}-f_{1}u_{2c}).
\]
It is easy to see that $D_{c}y_{1c}=y_{1c}^{obs}$, $D_{c}y_{2c}=y_{2c}^{obs}$,
$D_{c}x_{c}^{c}=x_{c}^{c,obs}$ and $D_{c}x_{c}^{p}=x_{c}^{p,obs}$,
which are all observed.

For $D_{c}M_{c}x_{c}^{p}$, observing that $D_{c}\iota_{c}=\iota_{n^{obs}}$
, $D_{c}I_{c}=D_{c}$, and $M_{c}=\frac{\iota_{c}\iota_{c}^{\prime}-I_{c}}{n_{c}-1}$,
we have
\begin{equation}
D_{c}M_{c}x_{c}^{p}=D_{c}\frac{\iota_{c}\iota_{c}^{\prime}-I_{c}}{n_{c}-1}x_{c}^{p}=\frac{1}{n_{c}-1}\left(\iota_{n_{c}^{nobs}}\iota_{c}^{\prime}x_{c}^{p}-x_{c}^{p,obs}\right)=\frac{1}{n_{c}-1}(\iota_{n_{c}^{obs}}n_{c}\bar{x}_{c}^{p}-x_{c}^{p,obs}),
\end{equation}
where $\bar{x}_{c}^{p}=\frac{1}{n_{c}}\sum_{r=1}^{n_{c}}x_{rc}^{p}$
is the average characteristics of all students. Let $\bar{x}_{c}^{p,obs}$
be the averages of all observed $x$. That is, as long as a variable
in $x^{p}$ is observed for a student, it will be used to calculate
the average even if the student has other variables missing and counted
as unobserved. For example, some classes may have variable ``$girl$''
observed for all students, while race and test information may be
missing for some. In this way we increase the efficiency of our approximation
as observing a missing value for a particular variable is much less
common than having missing values for all of the variables related
to a particular student. Then $\bar{girl}_{c}^{p,obs}=\bar{girl}_{c}^{p}$
as we use all observed values of $"girl"$. Assuming that observations
are missing at random it follows that $E\left[\bar{x}_{c}^{p}-\bar{x}_{c}^{p,obs}\right]=0$
and the difference between the two is very small as personal characteristics
don't have that many missing values. See Table \ref{tab:checkmiss}
for details. We can therefore approximate $D_{c}M_{c}x_{c}^{p}$ with
$\tilde{x}_{c}^{m,obs}=\frac{1}{n_{c}-1}(\iota_{n_{c}^{obs}}n_{c}\bar{x}_{c}^{p,obs}-x_{c}^{p,obs})$.
In all, 
\begin{align*}
y_{1c}^{obs}-f_{1}y_{2c}^{obs} & =x_{c}^{c,obs}\delta_{c}+x_{c}^{p,obs}\delta_{p}+\tilde{x}_{c}^{m,obs}\delta_{m}+v_{c}=\text{\underbar{x}}_{c}^{obs}\delta+v_{c},
\end{align*}
where $\text{\underbar{x}}_{c}^{obs}=\left(x_{c}^{c,obs},x_{c}^{p,obs},\tilde{x}_{c}^{m,obs}\right)$
, $\delta=(\delta_{c}^{\prime},\delta_{p}^{\prime},\delta_{m}^{\prime})^{\prime}$,
$v_{c}=\xi_{c}+\epsilon_{c}$, with
\[
\epsilon_{c}=D_{c}(I_{c}+\rho M_{c})(u_{1c}-f_{1}u_{2c}),
\]
\[
\xi_{c}=\left(D_{c}M_{c}x_{c}^{p}-\tilde{x}_{c}^{m,obs}\right)\delta_{m}=\frac{n_{c}}{n_{c}-1}\iota_{n_{c}^{obs}}(\bar{x}_{c}^{p}-\bar{x}_{c}^{p,obs})\delta_{m}.
\]

If observations are missing randomly, we have 
\[
E\left[\epsilon_{c}|x_{c}^{c,obs},x_{c}^{p,obs},\tilde{x}_{c}^{m,obs}\right]=0,
\]
\begin{equation}
\Omega_{c}^{obs}\equiv E\left[\epsilon_{c}\epsilon_{c}^{\prime}\right]=\sigma_{c}^{2}\left((\frac{n_{c}-1-\rho}{n_{c}-1})^{2}I_{n^{obs}}+\frac{1}{n_{c}}((1+\rho)^{2}-(\frac{n_{c}-1-\rho}{n_{c}-1})^{2})\iota_{n^{obs}}\iota_{n^{obs}}^{\prime}\right).\label{eq:omega}
\end{equation}
Note that $E\left[v_{c}|\text{\underbar{x}}_{c}^{obs}\right]$ and
$\textrm{var}(v_{c}|\text{\underbar{x}}_{c}^{obs})$ requires additional
assumptions on the distribution of $\xi_{c}$, in particular $E\left[\bar{x}_{c}^{p}-\bar{x}_{c}^{p,obs}\vert\bar{x}_{c}^{p,obs}\right]=0$.
In practice, we ignore the term $\xi_{c}$. As long as the number
of missing $x^{p}$ is small, as Table \ref{tab:checkmiss} shows,
the bias is small and negligible.

The formula for the variance in \eqref{eq:omega} points to the importance
for adjusting for original class size $n_{c}$ even if observations
are missing randomly. We can then define 
\[
u_{c}^{+obs}=(\Omega_{c}^{obs})^{-1/2}\left(y_{1c}^{obs}-f_{1}y_{2c}^{obs}-\text{\underbar{x}}_{c}^{obs}\delta\right)
\]
and continue with our GMM estimation strategy.

\section{Additional Empirical Results}

\subsection{Additional Specification Tests}

Our analysis thus far has documented that the word and reading scores
are highly correlated both contemporaneously as well as over the course
of one school year. This evidence is consistent with restrictions
imposed by Assumption \ref{assu:Proportional} and the hypothesis
that closely related scores depend on the same unobserved components.
A caveat is that closely related scores may remove too much common
variation when quasi-differences of these scores are used and thus
render statistical evidence about peer effects imprecise. Panel A
of Table \ref{tab:sdratio_all} reports pseudo $R^{2}$ measures of
model estimates for word and reading scores and model specifications
(1) through (4) with column (1) only controlling for school fixed
effects, and column (4) being the full specification with school and
class type effects, student and teacher characteristics as well as
peer characteristics. Results for other pairs of scores are presented
in panels B-F as a comparison.

The observed pseudo $R^{2}$ measures are all below 0.2 for the word-read
pair. As expected the pseudo $R^{2}$ generally increases with larger
models, or in other words by moving from Column (1) to (4). Compared
to other pairs of scores, the word-read score pair often has the lowest
$R^{2}$ which is an indication that the differential score in these
specifications already controls for most of the variation explained
by the covariates. We take this as further evidence that word-read
best fits the model assumptions we impose.

Table \ref{tab:varqxy} presents standard deviations of $Q_{x}y$
for reading and word scores and $Q_{x}(y_{1}-f_{1}y_{2})$ for $y_{1}$
being word and $y_{2}$ being reading scores for Grade K to 3. The
left and right panels are for small and regular classes respectively.
Control variables for specifications (1)-(4) are the same as those
in Table \ref{tab:sdratio_all}, as discussed above. We see that in
Grade K to 3, the ratio of standard deviations of $Q_{x}(y_{1}-f_{1}y_{2})$
to that of $Q_{x}y$ is 1/5 to 1/4.This shows that quasi-differencing
removes a significant part of the variation in the test scores, consistent
with the term $\mu^{*}$ playing a major role in determining test
performance. At the same time, the results in Table \ref{tab:varqxy}
indicate that the inclusion of controls does not eliminate much of
the remaining variation not captured by fixed effects. The conclusion
is that after differencing there is between 20 to 30 percent of test
score variation remaining for the identification of $\rho$.

\subsection{Effects of Individual Characteristics}

We also provide the coefficients of peer averages of individual characteristics
(age, black, free lunch, and girl) in Table \ref{tab:ave}. These
four variables all pertain to test-invariant individual characteristics.
The estimated coefficients can be expressed as $\rho(\beta_{v1}^{p}-f_{1}\beta_{v2}^{p})$,
where $\rho\beta_{v1}^{p}$ and $\rho\beta_{v2}^{p}$ represent the
exogenous peer effect of the variable on tests $t=1$ and $t=2$ respectively.
Thus, these coefficients capture the differential exogenous peer effects
across the two tests. We also scale the coefficients for dummies $girl$,
$black$ and $free\,lunch$ by 1/100 so that the coefficients are
the impact on the outcome of 1 percentage point changes in share of
peers who are girl/black/free lunch receivers. Except for kindergarten
where the presence of girls among peers positively impacts reading
more than word study, none of the four variables show significant
differential exogenous peer effects. Overall, the results indicate
small and mostly insignificant peer effects on differential scores.
This is consistent with a scenario where $y_{2}$ is essentially sufficient
to control for both observed and unobserved test invariant determinants
of test performance. In fact, when $f_{1}=1$ and $\beta_{v1}^{p}=\beta_{v2}^{p}$
these effects would be zero in our model and the overall evidence
points to the conclusion that this is approximately the case in the
sample as well.

 \begin{table}[htpb] \begin{center} \caption{Number of Missing Values for Each Variable} \label{tab:checkmiss}  \begin{threeparttable}  \begin{tabular}{ccccc} \hline \hline & grade K & grade 1 & grade 2 & grade 3 \\ \hline 
\multicolumn{5}{c}{\textit{student characteristics}} \\
poor&8&32&62&83\\
black&3&29&89&14\\
girl&0&13&7&0\\
age&8&37&7&20\\
\multicolumn{5}{c}{\textit{teacher characteristics}} \\
tchfemale&0&19&60&50\\
tchblack&43&19&60&50\\
tchms&21&19&60&65\\
tchyears&21&19&101&50\\
\multicolumn{5}{c}{\textit{test scores}} \\
read&536&434&763&802\\
list&488&273&797&728\\
math&454&231&775&725\\
word&474&857&493&453\\
\multicolumn{5}{c}{\textit{Any vairable missing for the test pairs and controls}} \\
list\&read&643&528&929&966\\
math\&read&589&480&907&946\\
read\&word&608&1,093&883&914\\
list\&math&570&332&926&878\\
list\&word&562&965&928&863\\
math\&word&555&929&904&842\\
\hline sample size&6,325&6,829&6,840&6,801\\
 \hline \hline \end{tabular} \begin{tablenotes} \footnotesize
\item 1. School ID, class ID, class size and class type have no missing values.
\end{tablenotes} \end{threeparttable} \end{center} \end{table}

 \begin{table}[ht] \begin{center} \caption{Summary Statistics} \label{tab:sum}  \begin{threeparttable}  \begin{tabular}{rrrrrrrrrrrr} \hline \hline & \multicolumn{2}{c}{Grade K} & & \multicolumn{2}{c}{Grade 1} & & \multicolumn{2}{c}{Grade 2} & & \multicolumn{2}{c}{Grade 3} \\ \cline{2-3}\cline{5-6}\cline{8-9}\cline{11-12} & mean & sd & & mean & sd & & mean & sd & & mean & sd \\ \hline 
\multicolumn{12}{c}{\textit{indiviudal characteristics}} \\
free lunch&0.560&0.496&&0.588&0.492&&0.592&0.491&&0.578&0.494\\
black&0.326&0.469&&0.327&0.469&&0.347&0.476&&0.333&0.471\\
girl&0.486&0.500&&0.480&0.500&&0.484&0.500&&0.480&0.500\\
age&6.011&0.349&&7.158&0.492&&8.242&0.557&&9.285&0.586\\
\multicolumn{12}{c}{\textit{teacher characteristics}} \\
black teacher&0.165&0.371&&0.174&0.379&&0.204&0.403&&0.209&0.406\\
master degree&0.347&0.476&&0.346&0.476&&0.373&0.484&&0.442&0.497\\
years of experience&9.258&5.809&&11.633&8.937&&13.145&8.655&&13.933&8.615\\
\multicolumn{12}{c}{\textit{SAT test scores}} \\
math&485.377&47.698&&530.528&43.109&&580.613&44.574&&617.970&39.841\\
read&436.725&31.706&&520.787&55.192&&583.935&46.043&&615.422&38.563\\
list&537.475&33.140&&567.487&33.674&&595.476&34.908&&624.119&32.265\\
word&434.179&36.759&&513.436&53.316&&582.986&50.658&&610.136&45.041\\
\multicolumn{12}{c}{\textit{Sample Size}} \\
number of students&6325&&&6829&&&6840&&&6801&\\
number of classes&325&&&339&&&340&&&336&\\
\# small&127&&&124&&&133&&&140&\\
\# regular&99&&&115&&&100&&&89&\\
\# regular/aide&99&&&100&&&107&&&107&\\
number of schools&79&&&76&&&75&&&75&\\
 \hline \hline \end{tabular} \begin{tablenotes} \footnotesize
\item 1. Mean and standard deviations of student characteristics, teacher charcteristics, and SAT scores, as well as number of students and classes in each grade.
\end{tablenotes} \end{threeparttable} \end{center} \end{table}

 \begin{table}[htpb] \begin{center} \caption{Correlation between other paris} \label{tab:corr_all}  \begin{threeparttable}  \begin{tabular}{cccccc} \hline \hline & \multicolumn{2}{c}{Rank Correlation} & & \multicolumn{2}{c}{Pseudo $ R^2 $ of 2SLS} \\ \cline{2-3}\cline{5-6} & $ y_1,y_2 $ & $ Q_X y_1,Q_X y_2 $ & & $ y_1,\hat{y}_1 $ & $ y_2,\hat{y}_2 $ \\ \hline 
\multicolumn{6}{c}{\textit{$ y_1 =math $, $ y_2 =listening $ }} \\
Grade K&0.657&0.584&&0.728&0.638\\
Grade 1&0.721&0.654&&0.766&0.716\\
Grade 2&0.678&0.611&&0.713&0.705\\
Grade 3&0.649&0.610&&0.689&0.673\\
\multicolumn{6}{c}{\textit{$ y_1 =reading $, $ y_2 =listening $}} \\
Grade K&0.621&0.554&&0.686&0.691\\
Grade 1&0.624&0.534&&0.704&0.601\\
Grade 2&0.667&0.601&&0.709&0.699\\
Grade 3&0.652&0.623&&0.690&0.683\\
\multicolumn{6}{c}{\textit{$ y_1 =reading $, $ y_2 =math $}} \\
Grade K&0.741&0.672&&0.697&0.763\\
Grade 1&0.737&0.669&&0.782&0.753\\
Grade 2&0.724&0.673&&0.756&0.758\\
Grade 3&0.741&0.706&&0.759&0.769\\
\multicolumn{6}{c}{\textit{$ y_1 =word $, $ y_2 =listening $}} \\
Grade K&0.558&0.477&&0.637&0.610\\
Grade 1&0.570&0.483&&0.645&0.562\\
Grade 2&0.597&0.533&&0.639&0.639\\
Grade 3&0.524&0.478&&0.581&0.552\\
\multicolumn{6}{c}{\textit{$ y_1 =word $, $ y_2 =math $}} \\
Grade K&0.656&0.579&&0.643&0.727\\
Grade 1&0.686&0.619&&0.726&0.715\\
Grade 2&0.625&0.564&&0.659&0.670\\
Grade 3&0.683&0.643&&0.712&0.708\\
 \hline \hline \end{tabular} \begin{tablenotes} \footnotesize
\item 1. The left panel is Spearman's rank correlation between $ y_1 $, $ y_2 $ and between $ Q_{X}y_{1} $, $ Q_{X}y_{2} $, where $ Q_{X}=I-\underbar{X}(\underbar{X}^{\prime}\underbar{X})^{-1}\underbar{X}^{\prime} $, $ y_{1} $ and $ y_{2} $ are raw SAT scores specified for each panel, $ \underbar{X} $ are the complete set of control variables, including school fixed effects, class type fixed effects, student characteristics, teacher characteristics and peer characteristics. \item 2. The right panel is pseudo $ R^{2} $, i.e., Spearman's rank correlation between $ y_{1} $ and $ \hat{y}_{1} $ in Column 3, and  between $ y_{2} $ and $ \hat{y}_{2} $ in Column 4. Here $ \hat{y}_{1}=\tilde{f}_{1}y_{2}+\text{\ensuremath{\underbar{X}}}\tilde{\delta} $, with $ \tilde{f_{1}} $ and $ \tilde{\delta} $ defined in (\ref{eq:f_tilde}) and (\ref{eq:delta_tilde}), $ \underbar{X} $ is the complete set of controls as described in note 1. That is,  $ \hat{y}_{1} $ is the predicted value of $ y_{1} $ from 2SLS for $ y_{1}=f_{1}y_{2}+\underbar{X}\delta+\epsilon $ with $ y_2 $ instrumented by the constant term,  $ \hat{y}_{2} $ is defined in a similar manner by reversing the order of $ y_1, y_2 $.
\end{tablenotes} \end{threeparttable} \end{center} \end{table}

 \begin{table}[htpb] \begin{center} \caption{Estimates of $ f_1 $  and $ \rho $ for alternate pairings} \label{tab:f1rho_all}  \begin{threeparttable}  \begin{tabular}{cccccccccc} \hline \hline & \multicolumn{4}{c}{$ f_1 $} & & \multicolumn{4}{c}{$ \rho $} \\ \cline{2-5}\cline{7-10} & (1) & (2) & (3) & (4) & & (1) & (2) & (3) & (4) \\ \hline 
\multicolumn{10}{c}{\textit{$ y_1 =math $, $ y_2=listening $ }} \\
Grade K&0.915&0.912&0.833&0.779&&0.721&0.691&0.694&0.688\\
&(0.009)&(0.008)&(0.019)&(0.148)&&(0.079)&(0.073)&(0.075)&(0.078)\\
Grade 1&0.939&0.936&0.949&0.946&&0.627&0.587&0.579&0.578\\
&(0.008)&(0.007)&(0.011)&(0.063)&&(0.069)&(0.073)&(0.073)&(0.074)\\
Grade 2&0.946&0.944&0.993&1.109&&0.637&0.621&0.554&0.525\\
&(0.005)&(0.005)&(0.011)&(0.068)&&(0.073)&(0.072)&(0.071)&(0.068)\\
Grade 3&0.973&0.974&1.040&1.113&&0.747&0.728&0.714&0.703\\
&(0.002)&(0.002)&(0.009)&(0.068)&&(0.088)&(0.087)&(0.086)&(0.090)\\
\multicolumn{10}{c}{\textit{$ y_1 =reading $, $ y_2 =listening $}} \\
Grade K&0.812&0.812&0.822&0.989&&0.434&0.415&0.407&0.377\\
&(0.005)&(0.005)&(0.014)&(0.107)&&(0.079)&(0.076)&(0.078)&(0.083)\\
Grade 1&0.944&0.943&1.028&0.978&&0.322&0.277&0.256&0.218\\
&(0.007)&(0.007)&(0.015)&(0.074)&&(0.051)&(0.049)&(0.048)&(0.046)\\
Grade 2&0.970&0.970&1.108&1.187&&0.568&0.553&0.571&0.561\\
&(0.005)&(0.005)&(0.012)&(0.070)&&(0.073)&(0.071)&(0.074)&(0.074)\\
Grade 3&0.987&0.987&1.093&1.151&&0.606&0.572&0.613&0.623\\
&(0.003)&(0.003)&(0.010)&(0.060)&&(0.091)&(0.088)&(0.096)&(0.101)\\
\multicolumn{10}{c}{\textit{$ y_1 = reading $, $ y_2=math $}} \\
Grade K&0.888&0.890&0.986&1.297&&0.492&0.485&0.519&0.575\\
&(0.004)&(0.004)&(0.018)&(0.236)&&(0.065)&(0.065)&(0.066)&(0.076)\\
Grade 1&1.005&1.007&1.092&1.032&&0.441&0.434&0.434&0.360\\
&(0.005)&(0.006)&(0.014)&(0.077)&&(0.063)&(0.064)&(0.061)&(0.061)\\
Grade 2&1.026&1.027&1.120&1.077&&0.481&0.479&0.475&0.444\\
&(0.006)&(0.006)&(0.012)&(0.055)&&(0.071)&(0.071)&(0.066)&(0.066)\\
Grade 3&1.015&1.015&1.052&1.036&&0.462&0.460&0.452&0.440\\
&(0.003)&(0.003)&(0.009)&(0.051)&&(0.078)&(0.078)&(0.082)&(0.086)\\
 \hline \hline \end{tabular} \begin{tablenotes} \footnotesize
\item 1. To be continued on the next page.
\end{tablenotes} \end{threeparttable} \end{center} \end{table}

\addtocounter{table}{-1}
 \begin{table}[htpb] \begin{center} \caption{Estimates of $ f_1 $  and $ \rho $ for alternate pairings(Continued)} \label{tab:_tem}  \begin{threeparttable}  \begin{tabular}{cccccccccc} \hline \hline & \multicolumn{4}{c}{$ f_1 $} & & \multicolumn{4}{c}{$ \rho $} \\ \cline{2-5}\cline{7-10} & (1) & (2) & (3) & (4) & & (1) & (2) & (3) & (4) \\ \hline 
\multicolumn{10}{c}{\textit{$ y_1=word $, $ y_2 =listening $}} \\
Grade K&0.817&0.816&0.833&0.954&&0.398&0.378&0.373&0.359\\
&(0.005)&(0.004)&(0.016)&(0.116)&&(0.066)&(0.065)&(0.065)&(0.068)\\
Grade 1&0.934&0.932&1.023&0.964&&0.385&0.340&0.308&0.257\\
&(0.008)&(0.008)&(0.016)&(0.081)&&(0.054)&(0.052)&(0.048)&(0.047)\\
Grade 2&0.978&0.977&1.145&1.284&&0.537&0.521&0.512&0.497\\
&(0.006)&(0.007)&(0.014)&(0.085)&&(0.070)&(0.068)&(0.070)&(0.069)\\
Grade 3&0.974&0.974&1.107&1.145&&0.512&0.482&0.518&0.522\\
&(0.004)&(0.004)&(0.012)&(0.069)&&(0.095)&(0.091)&(0.095)&(0.097)\\
\multicolumn{10}{c}{\textit{$ y_1 =word $, $ y_2 =math $}} \\
Grade K&0.893&0.895&1.001&1.261&&0.517&0.513&0.538&0.578\\
&(0.006)&(0.006)&(0.021)&(0.242)&&(0.067)&(0.067)&(0.070)&(0.081)\\
Grade 1&0.994&0.994&1.077&1.027&&0.487&0.480&0.469&0.388\\
&(0.004)&(0.005)&(0.014)&(0.089)&&(0.066)&(0.065)&(0.062)&(0.064)\\
Grade 2&1.034&1.034&1.159&1.160&&0.493&0.492&0.478&0.459\\
&(0.007)&(0.007)&(0.014)&(0.074)&&(0.071)&(0.071)&(0.066)&(0.065)\\
Grade 3&1.001&1.001&1.064&1.031&&0.371&0.366&0.370&0.357\\
&(0.004)&(0.004)&(0.010)&(0.054)&&(0.074)&(0.072)&(0.074)&(0.078)\\
\hline school FE & Y & Y & Y & Y & & Y & Y & Y & Y \\ class type FE & & Y & Y & Y & & & Y & Y & Y \\ stu\&tch char & & & Y & Y & & & & Y & Y \\ peer char & & & & Y & & & & & Y\\ \hline \hline \end{tabular} \begin{tablenotes} \footnotesize
\item 1. Estimates and standard errors (in the parentheses) for $ f_1 $  and $ \rho $. Standard errors are clustered at the classroom level. All models allow for heteroscedasticity across small and regular classes. Estimates are adjusted for missing observations. \item 2. As indicated at the bottom of the table, the models may control for (1) school fixed effects (excluding one); (2) class type fixed effects (excluding small); (3) a student's characteristics (free lunch, black, girl, age) and teacher characteristics (black teacher, master, years of experience); (4) peer averages of student characteristics.
\end{tablenotes} \end{threeparttable} \end{center} \end{table}

 \begin{table}[htpb] \begin{center} \caption{Pseudo $ R^2 $: $ 1-var(u_1-f_1 u_2)/var(y_1-f_1y_2) $} \label{tab:sdratio_all}  \begin{threeparttable}  \begin{tabular}{cccccccccc} \hline \hline & \multicolumn{4}{c}{Small Classes} & & \multicolumn{4}{c}{Regualr Classes(w/wo Aide)} \\ \cline{2-5}\cline{7-10} & (1) & (2) & (3) & (4) & & (1) & (2) & (3) & (4) \\ \hline 
\multicolumn{10}{c}{\textbf{Panel A: } $ \mathbf{y_1 =word, y_2=reading} $} \\
Grade K&0.104&0.104&0.107&0.113&&0.140&0.140&0.142&0.146\\
Grade 1&0.088&0.088&0.092&0.090&&0.093&0.093&0.098&0.093\\
Grade 2&0.064&0.064&0.077&0.082&&0.065&0.065&0.074&0.078\\
Grade 3&0.093&0.093&0.092&0.096&&0.033&0.034&0.035&0.038\\
\multicolumn{10}{c}{\textit{Panel B: $ y_1 =math $, $ y_2=listening $ }} \\
Grade K&0.179&0.180&0.206&0.216&&0.177&0.179&0.201&0.208\\
Grade 1&0.138&0.140&0.143&0.144&&0.113&0.115&0.123&0.124\\
Grade 2&0.116&0.117&0.123&0.119&&0.104&0.106&0.118&0.109\\
Grade 3&0.135&0.136&0.147&0.144&&0.085&0.087&0.092&0.089\\
\multicolumn{10}{c}{\textit{Panel C: $ y_1 =reading $, $ y_2 =listening $}} \\
Grade K&0.143&0.144&0.153&0.150&&0.145&0.146&0.150&0.147\\
Grade 1&0.131&0.133&0.154&0.163&&0.113&0.116&0.148&0.161\\
Grade 2&0.087&0.087&0.144&0.141&&0.083&0.084&0.125&0.125\\
Grade 3&0.073&0.074&0.133&0.133&&0.024&0.027&0.077&0.078\\
\multicolumn{10}{c}{\textit{Panel D: $ y_1 = reading $, $ y_2=math $}} \\
Grade K&0.148&0.148&0.177&0.219&&0.126&0.127&0.155&0.192\\
Grade 1&0.105&0.105&0.130&0.137&&0.070&0.071&0.119&0.132\\
Grade 2&0.092&0.092&0.134&0.137&&0.086&0.087&0.110&0.109\\
Grade 3&0.083&0.083&0.107&0.106&&0.050&0.050&0.085&0.086\\
\multicolumn{10}{c}{\textit{Panel E: $ y_1=word $, $ y_2 =listening $}} \\
Grade K&0.124&0.125&0.132&0.129&&0.132&0.133&0.136&0.134\\
Grade 1&0.105&0.106&0.122&0.132&&0.068&0.071&0.094&0.108\\
Grade 2&0.084&0.084&0.131&0.132&&0.070&0.072&0.107&0.109\\
Grade 3&0.090&0.090&0.132&0.131&&0.030&0.033&0.065&0.066\\
\multicolumn{10}{c}{\textit{Panel F: $ y_1 =word $, $ y_2 =math $}} \\
Grade K&0.138&0.139&0.164&0.200&&0.128&0.129&0.154&0.184\\
Grade 1&0.092&0.092&0.114&0.118&&0.058&0.059&0.101&0.109\\
Grade 2&0.089&0.089&0.124&0.126&&0.081&0.081&0.106&0.107\\
Grade 3&0.085&0.085&0.107&0.106&&0.039&0.040&0.071&0.070\\
\hline school FE & Y & Y & Y & Y & & Y & Y & Y & Y \\ class type FE & & Y & Y & Y & & & Y & Y & Y \\ stu\&tch char & & & Y & Y & & & & Y & Y \\ peer char & & & & Y & & & & & Y\\ \hline \hline \end{tabular} \begin{tablenotes} \footnotesize
\item 1. The table reports one minus the ratio of the variance of $ u_1-f_1 u_2 $ to that of $ y_1-f_1 y_2 $ for small classes (in the left panel) and regular (with or without aide) classes (in the right panel). All models allow for heteroscedasticity across small and regular classes. Estimates are adjusted for missing observations. \item 2. As indicated at the bottom of the table, the models may control for (1) school fixed effects (excluding one); (2) class type fixed effects (excluding small); (3) a student's characteristics (free lunch, black, girl, age) and teacher characteristics (black teacher, master, years of experience); (4) peer averages of student characteristics.
\end{tablenotes} \end{threeparttable} \end{center} \end{table}

 \begin{table}[htpb] \begin{center} \caption{Standard Deviations of $ Q_{x}y $ and $ Q_{x}(y_1-f_1 y_2) $ } \label{tab:varqxy}  \begin{threeparttable}  \begin{tabular}{cccccccccc} \hline \hline & \multicolumn{4}{c}{Small Classes} & & \multicolumn{4}{c}{Regular Classes} \\ \cline{2-5}\cline{7-10} & (1) & (2) & (3) & (4) & & (1) & (2) & (3) & (4) \\ \hline 
\multicolumn{10}{c}{\textit{Grade K}} \\
$ Q_{x}y $, y=read&62.748&62.169&33.699&26.827&&66.321&64.827&35.534&29.143\\
$ Q_{x}y $, y=word&65.196&64.629&38.017&31.883&&68.723&67.294&39.693&33.916\\
$ Q_{x}(y_1-f_1 y_2) $&14.131&14.131&14.103&14.302&&14.705&14.707&14.651&14.831\\
\multicolumn{10}{c}{\textit{Grade 1}} \\
$ Q_{x}y $, y=read&83.720&82.984&57.172&46.261&&88.372&86.584&58.298&49.448\\
$ Q_{x}y $, y=word&81.314&80.624&57.574&46.620&&85.880&84.192&57.174&47.750\\
$ Q_{x}(y_1-f_1 y_2) $&21.475&21.470&21.350&21.239&&22.352&22.346&22.183&22.048\\
\multicolumn{10}{c}{\textit{Grade 2}} \\
$ Q_{x}y $, y=read&93.586&92.295&57.595&39.409&&87.831&85.406&54.519&41.740\\
$ Q_{x}y $, y=word&96.864&95.555&62.274&44.613&&90.149&87.782&59.070&46.967\\
$ Q_{x}(y_1-f_1 y_2) $&24.071&24.070&24.070&24.312&&24.496&24.497&24.539&24.796\\
\multicolumn{10}{c}{\textit{Grade 3}} \\
$ Q_{x}y $, y=read&102.223&100.223&52.813&33.948&&85.025&82.325&51.772&38.114\\
$ Q_{x}y $, y=word&102.419&100.495&56.860&39.320&&86.959&84.387&56.262&43.608\\
$ Q_{x}(y_1-f_1 y_2) $&20.365&20.367&20.330&20.323&&20.839&20.833&20.730&20.721\\
 \hline \hline \end{tabular} \begin{tablenotes} \footnotesize
\item 1. Standard deviations of $ Q_x y $ and $ Q_x (y_1-f_1 y_2) $ for small (on the left) and regular classes (on the right). For $ Q_x y $, the test score $ y $ is reading or word score For $ Q_x (y_1-f_1 y_2) $, $ y_1 $ and $ y_2 $ are word and reading scores respectively. The control variables for models (1)-(4) are the same as those for the main models in Table \ref{tab:f1rho}. The estimates of $ f_1 $ are from our efficient GMM estimations for each specifications and presented in Table \ref{tab:f1rho}.
\end{tablenotes} \end{threeparttable} \end{center} \end{table}

 \begin{table}[htpb] \begin{center} \caption{Estimates of Exogenous Peer Effects for the Word-Reading Pair} \label{tab:ave}  \begin{threeparttable}  \begin{tabular}{ccccc} \hline \hline & \hspace{1.2cm}Age & \hspace{1.2cm}Black & \hspace{1.2cm}Poor & \hspace{1.2cm}Girl \\ \hline 
Grade K&\hspace{1.2cm}4.126\hspace{1.2cm}&\hspace{1.2cm}0.028\hspace{1.2cm}&\hspace{1.2cm}0.005\hspace{1.2cm}&\hspace{1.2cm}-0.067\hspace{1.2cm}\\
&\hspace{1.2cm}(3.268)\hspace{1.2cm}&\hspace{1.2cm}(0.060)\hspace{1.2cm}&\hspace{1.2cm}(0.023)\hspace{1.2cm}&\hspace{1.2cm}(0.028)**\hspace{1.2cm}\\
Grade 1&\hspace{1.2cm}0.701\hspace{1.2cm}&\hspace{1.2cm}0.032\hspace{1.2cm}&\hspace{1.2cm}-0.037\hspace{1.2cm}&\hspace{1.2cm}0.036\hspace{1.2cm}\\
&\hspace{1.2cm}(3.307)\hspace{1.2cm}&\hspace{1.2cm}(0.063)\hspace{1.2cm}&\hspace{1.2cm}(0.042)\hspace{1.2cm}&\hspace{1.2cm}(0.048)\hspace{1.2cm}\\
Grade 2&\hspace{1.2cm}-3.204\hspace{1.2cm}&\hspace{1.2cm}0.144\hspace{1.2cm}&\hspace{1.2cm}-0.027\hspace{1.2cm}&\hspace{1.2cm}-0.113\hspace{1.2cm}\\
&\hspace{1.2cm}(2.997)\hspace{1.2cm}&\hspace{1.2cm}(0.093)\hspace{1.2cm}&\hspace{1.2cm}(0.053)\hspace{1.2cm}&\hspace{1.2cm}(0.058)*\hspace{1.2cm}\\
Grade 3&\hspace{1.2cm}2.157\hspace{1.2cm}&\hspace{1.2cm}-0.029\hspace{1.2cm}&\hspace{1.2cm}-0.036\hspace{1.2cm}&\hspace{1.2cm}0.025\hspace{1.2cm}\\
&\hspace{1.2cm}(2.486)\hspace{1.2cm}&\hspace{1.2cm}(0.061)\hspace{1.2cm}&\hspace{1.2cm}(0.039)\hspace{1.2cm}&\hspace{1.2cm}(0.044)\hspace{1.2cm}\\
 \hline \hline \end{tabular} \begin{tablenotes} \footnotesize
\item 1. Estimates and standard errors (in the parentheses) for exogenous peer effects of peer's age, race, free lunch status and gender. Coefficients for dummies black, free lunch and girl are divided by 100 so they can be interpreted as the impact of one percentage point change in these variables. Note that these four variables are all test-invariant personal characteristics. The reported coefficients are $ \rho(\beta_{v1}^p-f_1\beta_{v2}^p) $. \item 2. * 0.1, ** 0.05, *** 0.01. Standard errors are clustered at the classroom level. All models allow for heteroscedasticity across small and regular classes. Estimates are adjusted for missing observations. \item 3. The scores are raw SAT scores in word study skills(word) and reading (read) for $ y_1 $ and $ y_2 $ respectively. \item 4. All specifications control for (1) school fixed effects (excluding one); (2) class type fixed effects (excluding small); (3) a student's characteristics (free lunch, black, girl, age) and teacher characteristics (black teacher, master, years of experience); (4) peer averages of student characteristics.
\end{tablenotes} \end{threeparttable} \end{center} \end{table}

\end{document}